%% file: paper.tex
\documentclass[runningheads]{llncs}

\usepackage{acro}
\usepackage{xspace}
\usepackage{bm}
\usepackage{csquotes}
\usepackage[bookmarks=false]{hyperref}
\usepackage{mathtools}
\usepackage{tikz}
\usepackage{amssymb}
\usepackage{graphicx}
\usepackage{makecell}
\usepackage{xltabular}
\usepackage{thm-restate}
\usepackage{verbatim}
\usepackage{cleveref}

\usepackage{xspace}
\usepackage{tikz}
\usepackage[linesnumbered]{algorithm2e} 
\usetikzlibrary{intersections,graphs,trees,fit,shapes.geometric,quotes,arrows.meta,backgrounds,matrix} 
\usepackage{csquotes}

\usepackage{booktabs}
\usepackage{array}
\makeatletter
\newcommand{\thickhline}{%
    \noalign {\ifnum 0=`}\fi \hrule height 1.5pt
    \futurelet \reserved@a \@xhline
}
\newcolumntype{"}{@{\hskip\tabcolsep\vrule width 1pt\hskip\tabcolsep}}
\makeatother

\newcommand{\Z}{\mathbb{Z}}
\newcommand{\R}{\mathbb{R}}
\newcommand{\Oh}[1]{O(#1)}

\newcommand{\etal}{et~al.\xspace}
\newcommand{\oh}[1]{o\left(#1\right)}
\newcommand{\norm}[1]{\left\lVert#1\right\rVert}
\newcommand{\eps}{\varepsilon}
\newcommand{\I}{\mathcal{I}}
\newcommand{\closedinterval}[1]{\left[#1\right]}

\newcommand*{\shortautoref}[1]{%
  \begingroup
    \def\sectionautorefname{Sec.}%
    \def\theoremautorefname{Thm.}%
    \def\corollaryautorefname{Cor.}%
    \def\lemmaautorefname{Lem.}%
    \def\propositionautorefname{Prop.}%
    \autoref{#1}%
  \endgroup
}

\def\O{O} 
\def\pmax{p_{\max}} 
\def\smax{s_{\max}} 
\def\dmax{d_{\max}} 
\def\wmax{w_{\max}} 
\def\nmax{n_{\max}} 
\def\m{m} 
\def\n{\norm{n}_1} 
\def\opt{\mathrm{OPT}} 
\crefname{cond}{Condition}{Conditions}
\creflabelformat{cond}{#2{\upshape#1}#3} 
\crefname{cons}{Constraint}{Constraints}
\creflabelformat{cons}{#2{\upshape#1}#3} 
\crefname{ineq}{Inequality}{Inequalities}
\creflabelformat{ineq}{#2{\upshape(#1)}#3} 
\crefname{equal}{Equality}{Equalities}
\creflabelformat{equal}{#2{\upshape(#1)}#3}

\renewcommand{\matrix}[1]{\begin{pmatrix}#1\end{pmatrix}}

\newcommand{\KS}{\textsc{Knapsack}\xspace}

\newcommand{\N}{\mathbb{N}}
\newcommand{\Nwithzero}{\mathbb{N}_0}
\newcommand{\Nwithoutzero}{\mathbb{N}_{>0}}
\newcommand{\mswbp}{\textsc{MinSum-Weighted-BinPacking}\xspace}

\newcommand{\rr}[1]{\multicolumn{1}{c|}{#1}}
\newcommand{\lrr}[1]{\multicolumn{1}{|c|}{#1}}

\newcommand{\one}{1}

\renewcommand{\P}{P}
\newcommand{\Q}{Q}
\newcommand{\unrel}{R}

\newcommand{\jobtypes}{d} 
\renewcommand{\dj}{d_j}

\newcommand{\class}{\textup{class}}
\newcommand{\setup}{s_j}
\newcommand{\capacity}{\textup{cap}}
\renewcommand{\vec}{\textup{vec}}

\newcommand{\cmax}{C_{\max}}
\newcommand{\cmin}{C_{\min}}
\newcommand{\envy}{C_{\textup{envy}}}

\newcommand{\sumwu}{\sum w_jU_j}

\newcommand{\sumwc}{\sum w_jC_j}

\newcommand{\elltwo}{\ell_2}

\newcommand{\tf}[3]{${#1|#2|#3}$}

\newcommand{\assip}{\textsc{Assignment ILP}\xspace}
\newcommand{\confip}{\textsc{Configuration ILP}\xspace}

\def\mlm{m} 
\def\nlm{\n} 
\def\M{\closedinterval{\mlm}} 
\def\J{\closedinterval{\nlm}} 
\def\jk{\mathcal{J}_k} 
\def\D{\closedinterval{\jobtypes}} 
\newcommand{\slotmilpdelta}{\operatorname{slot-MILP^{\delta}}}
\newcommand{\slotmilp}{\operatorname{slot-MILP}}
\newcommand{\enc}[1]{\langle#1\rangle}
\newcommand{\epsfrac}{\frac{1}{\eps}}
\newcommand{\fromvertex}{v_1}
\newcommand{\tovertex}{v_2}
\newcommand{\thirdvertex}{v_3}

\newcommand{\dprt}{\left( m \jobtypes \pmax \right)^{\O(\epsfrac)}}
\newcommand{\swaprt}{\epsfrac m^2 + \log(\jobtypes \pmax)}
\newcommand{\localsearchrt}{m^3 \jobtypes^3 \pmax \epsfrac \log(m \jobtypes \pmax)}

\newcommand{\additiveapproxrt}{ (m \pmax) ^{\O(\epsfrac)}}
\newcommand{\originaladditiveapproxrt}{ m \n^{\O(\epsfrac)}}

\renewcommand{\dj}{C_j\leq d_j}

{%
  \leavevmode\nobreak\par
\begin{samepage}
	\begin{list}%
		{}%
		{%
			\def\labelstyle{\itshape}
			\setlength{\topsep}{0pt}%
			\settowidth{\labelwidth}{\labelstyle Parameter:}%
			\setlength{\leftmargin}{\labelwidth}%
			\addtolength{\leftmargin}{\labelsep}%
			\setlength{\itemsep}{0pt}%
			\setlength{\parsep}{0pt}%
		}%
	}%
	{%
	\end{list}%
\end{samepage}
}

\usetikzlibrary{intersections,graphs,trees,fit,shapes.geometric,quotes,arrows.meta,backgrounds,matrix} 

\usetikzlibrary{decorations.pathreplacing,calligraphy}

\usepackage{caption}
\usepackage{subcaption}

\tikzset{
    cross/.pic = {
    \draw[rotate = 45] (-#1,0) -- (#1,0);
    \draw[rotate = 45] (0,-#1) -- (0, #1);
    }
}

\makeatletter
\def\hlinewd#1{%
\noalign{\ifnum0=`}\fi\hrule \@height #1 \futurelet
\reserved@a\@xhline}
\makeatother

\usetikzlibrary{positioning}
\usetikzlibrary{calc}

\newcommand{\pqmr}{PQ-R\xspace}
\newcommand{\pqmrs}{PQ-Rs\xspace}

\DeclarePairedDelimiterX\sett[2]{\lbrace}{\rbrace}{ #1 \,\delimsize| \,\mathopen{} #2 }

\usepackage{graphicx}
\begin{document}
%
\title{Exact and Approximate High-Multiplicity Scheduling on Identical Machines\thanks{Supported by the German Research Foundation (DFG) project JA 612/25-1.}}
\titlerunning{High-Multiplicity Scheduling on Identical Machines}
%
\author{Klaus Jansen \and Kai Kahler \and Esther Zwanger}
\authorrunning{Jansen et al.}
%
\institute{Department of Computer Science, Kiel University, Kiel, Germany\\\email{\{kj,kka\}@informatik.uni-kiel.de} \\ \email{stu222121@mail.uni-kiel.de}}
\maketitle              
\begin{abstract}
Goemans and Rothvoss (SODA'14) gave a framework for solving problems which can be described as finding a point in $\text{int.cone}(P\cap\Z^N)\cap Q$, where $P,Q\subset\R^N$ are (bounded) polyhedra. The running time for solving such a problem is $\enc{P}^{2^{\Oh{N}}}\enc{Q}^{\Oh{1}}$.
    This framework can be used to solve various scheduling problems, but the encoding length $\enc{P}$ usually involves parameters like the makespan or deadlines (which can be very large compared to the processing times). We describe three tools to improve the framework by Goemans and Rothvoss:
    \begin{itemize}
        \item Problem-specific preprocessing can be used to greatly reduce $\enc{P}$.
        \item By solving a certain LP relaxation, one can obtain bounds for the points in $P$. Combined with the classical result by Frank and Tardos (J.\ Comb.\ '87), these yield a more compact encoding of $P$ in general.
        \item A result by Jansen and Klein (SODA'17) changes the running time of the algorithm by Goemans and Rothvoss to $|V|^{2^{\Oh{N}}}\enc{P}^{\Oh{1}}\enc{Q}^{\Oh{1}}$, where $V$ is the set of vertices of the convex hull of $P\cap\Z^N$. We provide a new bound for $|V|$ that is similar to the one by Berndt \etal (SOSA'21) but better for our setting; this gives an alternative way to improve the framework.
    \end{itemize}
    For example, applied to the scheduling problems \tf{\P}{}{\{\cmax,\cmin,\envy\}}, these tools improve the running time from $(\log(\cmax))^{2^{\Oh{d}}}\enc{I}^{\Oh{1}}$ to the possibly much better $(\log(\pmax))^{2^{\Oh{d}}}\enc{I}^{\Oh{1}}$. Here, $\pmax$ is the largest processing time, $d$ is the number of different processing times, $\cmax$ is the makespan and $\enc{I}$ is the encoding length of the instance. 
    
    On the complexity side, we use reductions from the literature to provide new parameterized lower bounds for \tf{\P}{}{\cmax}. Finally, we show that the big open question asked by Mnich and van Bevern (Comput.\ Oper.\ Res.\ '18) whether \tf{\P}{}{\cmax} is FPT w.r.t. the number of job types $d$ has the same answer as the question whether \tf{\Q}{}{\cmax} is FPT w.r.t. the number of job and machine types $d+\tau$ (all in high-multiplicity encoding). The same holds for objective $\cmin$.
    

\end{abstract}

\section{Introduction}\label{sect:introduction}
Distributing jobs or tasks among workers or machines is one of the most natural optimization problems arising in practice. Such \emph{scheduling problems} have been systematically studied for over half a century (see e.g.\ \cite{LS20} for a great introduction to scheduling). Different machine models, job characteristics and objectives make for a plethora of different problems; some of them can be solved efficiently by scheduling jobs in a certain order, but many are NP-hard. Hence, a large part of research on scheduling is concerned with approximation algorithms; on the other hand, there are also many pseudo-polynomial algorithms. Another -- regarding the whole history of scheduling -- more recent research direction are \emph{parameterized algorithms}. Such algorithms encapsulate the complexity in certain parameters of the problem, and should these parameters be quite small, the algorithms can be very efficient~\cite{CFKLMPPS15}. 

In this work, we primarily focus on the parameter $d$, the number of different \emph{job types}. The exact definition may vary as the problems become more complex (with weights, classes or deadlines), but in a simple setting where a job $j$ only has a processing time $p_j$, the parameter $d$ describes the number of different processing times. This parameter choice seems natural in an industrialized world, where workers and machines have a large degree of specialization and hence do not work on too many different types of jobs or tasks. 
We use the standard three-field notation introduced by Graham \etal~\cite{GLLR79}; for a detailed description of our problems and parameters, see \autoref{sect:preliminaries}.

With the result by Goemans and Rothvoss~\cite{GR20} (or the one by Jansen and Klein~\cite{JK20}), one can solve problems that can be modelled by two polyhedra $P$ and $Q$ (with $P$ bounded) such that solutions of the problem directly correspond to a point $y\in\text{int.cone}(P\cap\Z^N)\cap Q$. Many scheduling problems fall into this category, e.g.\ makespan minimization on identical machines (\tf{\P}{}{\cmax}). However, when described by such a \emph{PQ-representation} (in short \pqmr), the makespan $\cmax$ appears in the system $Ax\leq b$ that represents the polytope $P$ and thus greatly influences the running time. 

\paragraph*{Our Contributions}
We describe three tools that can be used to improve the framework by Goemans and Rothvoss~\cite{GR20} such that the running time does not depend on $\norm{b}_\infty$, which in the case of scheduling problems eliminates parameters like the makespan or the largest due date from the running time.

The first tool (covered in \autoref{sect:toolone}) is a balancing result by Govzmann \etal~\cite{GMO23} that allows pre-scheduling many of the jobs, leaving us with a makespan value that is bounded by $2d\pmax$, where $\pmax$ is the largest processing time. Applying this preprocessing before using the algorithm by Goemans and Rothvoss yields the following running times for solving \tf{\P}{}{\{\cmax,\cmin,\envy\}}:
\begin{restatable}{theorem}{pcmaxminenvybalancing}\label{thm:pcmaxminenvy}
    The optimization problems \tf{\P}{}{\{\cmax,\cmin\}} can be solved in time $(d\log(\pmax)+\log(2d\pmax^2))^{2^{\Oh{d}}}\enc{I}^{\Oh{1}}$ and the optimization problem \tf{\P}{}{\envy} in time $\pmax(d\log(\pmax)+\log(2d\pmax^2))^{2^{\Oh{d}}}\enc{I}^{\Oh{1}}$.
\end{restatable}
We also show lower bounds for \tf{\P}{}{\cmax} that somewhat resemble the running time of our algorithm, though there is still a gap (for a proof, see \autoref{sect:complexity}):
\begin{restatable}{theorem}{exactlowerbound}
	Let $\eps>0$. Unless the ETH fails, \tf{\P}{}{\cmax} cannot be solved in time $\enc{I}^{\Oh{d^{1-\eps}}}\pmax^{\Oh{1}}$, $\enc{I}^{\Oh{1}}d^{\Oh{d^{1-\eps}}}\pmax^{\Oh{1}}$, $\enc{I}^{\oh{\frac{d}{\log(d)}}}\pmax^{\Oh{1}}$ or $\pmax^{O(d^{1-\eps})}$.
\end{restatable}
Note that the encoding length $\enc{I}$ includes $\log(\pmax)$ and $d$. In \autoref{sect:additiveapprox}, we show that the balancing result can be used to speed up an additive approximation scheme by Buchem \etal~\cite{buchem} in the case where the number of jobs $\n$ is larger than the number of machines $m$ times $\pmax$:
\begin{restatable}{theorem}{thmadditiveapprox}\label{thm:additiveapprox}
    There is an additive approximation scheme with error at most $\eps \pmax$ and running time $\additiveapproxrt$ for the problems \tf{\P}{}{\{\cmax,\cmin,\envy\}}.
\end{restatable}
In \autoref{sect:omittedproofs}, we give a matching lower bound for objective $\cmax$ via ETH:
\begin{restatable}{theorem}{thmlowerboundapprox}\label{thm:lowerboundapprox}
Let $\delta>0$. Then \tf{\P}{}{\cmax} cannot be approximated with additive error at most $\eps \pmax$ in time $(m\pmax)^{\Oh{(\frac{1}{\varepsilon})^{1-\delta}}}$, unless the ETH fails.
\end{restatable}

The second tool (covered in \autoref{sect:tooltwo}) involves solving a special relaxation of an integer linear program (ILP) associated with a given \pqmr. A proximity result by Cslovjecsek \etal~\cite{CEHRW20} then gives an upper bound for the points in $P$, which allows us to reduce the coefficients in the system $Ax\leq b$ describing $P$ with the famous result by Frank and Tardos~\cite{FT87}. Together, this then yields an algorithm for solving \pqmrs in a running time independent of $\norm{b}_\infty$:\footnote{Of course, the entries in $b$ still appear in the encoding length of $P$, but only logarithmically.}
\begin{restatable}{theorem}{maintheorem}\label{thm:main}
    If a problem has a \pqmr $(P,Q,m)$ given by a polytope $P=\left\{\left.x\in\R_{\geq0}^{N}\,\right|\,A^{(P)}x= b^{(P)}\right\}$ and $Q=\left\{\left.x\in\R_{\geq0}^{N}\,\right|\,A^{(Q)}x= b^{(Q)}\right\}$ with $M^{(P)}$ and $M^{(Q)}$ constraints, respectively, it can be solved in time
    \begin{align*}
        &\left(\left(M^{(P)} M^{(Q)}\log\left(\max\left\{\norm{A^{(P)}}_\infty, \norm{A^{(Q)}}_\infty\right\}\right)\right)^{2^{\Oh{N}}}\right.
        +2^{\Oh{(M^{(Q)})^2}}\bigg)\\&(\enc{P}\enc{Q}\log(m))^{\Oh{1}}.
    \end{align*}
\end{restatable}

The third tool (covered in \autoref{sect:toolthree}) is an upper bound for the number of vertices of the integer hull of a polytope that is similar to the one by Berndt \etal~\cite{BJK21} but better for our specific purpose:
\begin{restatable}{theorem}{vertexboundimproved}\label{thm:vertexboundimproved}
    The integer hull of a polytope $P=\left\{\left.x\in\R^N_{\geq0}\,\right|\,Ax= b\right\}$ has at most $N^M\Oh{M\log(M\Delta)}^N$ vertices, where $M$ is the number of constraints and $\Delta=\norm{A}_\infty$.
\end{restatable}
Note that this does not depend on $\norm{b}_\infty$. Previously known bounds either depend on $\norm{b}_\infty$ or are exponential in $\log(\norm{A}_\infty)$. Combined with the algorithm by Jansen and Klein~\cite{JK20}, this gives an alternative way of solving \pqmrs:
\begin{restatable}{theorem}{mainvertices}\label{thm:mainvertices}
    Let $P=\left\{\left.x\in\R_{\geq0}^{N}\,\right|\,A^{(P)}x= b^{(P)}\right\}$ with $M^{(P)}$ constraints and $Q=\left\{\left.x\in\R_{\geq0}^{N}\,\right|\,A^{(Q)}x= b^{(Q)}\right\}$ with $M^{(Q)}$ constraints. Then the \pqmr $(P,Q,m)$ can be solved in time
    \begin{align*}
        \left(N^{M^{(P)}}M^{(P)}\log\left(\norm{A^{(P)}}_\infty\right)\right)^{2^{\Oh{N}}}(\enc{P}\enc{Q}\log(m))^{\Oh{1}}.
    \end{align*}
\end{restatable}

Both \autoref{thm:main} and \autoref{thm:mainvertices} can be used to obtain algorithms for various scheduling problems that are more efficient than a straightforward application of the results by Goemans and Rothvoss~\cite{GR20} and Jansen and Klein~\cite{JK20} with previously known bounds for the vertices of the integer hull.
In \autoref{sect:applications}, we give \pqmrs for various scheduling problems, \mswbp (MSWBP) and uniform $n$-fold ILPs. 

The theorems also yield $(\log(\pmax))^{2^{\Oh{d}}}\enc{I}^{\Oh{1}}$-time (respectively for objective $\envy$ with an additional $\pmax$-factor) algorithms for \tf{\P}{}{\{\cmax,\cmin,\envy\}}, but with worse constants than in \autoref{thm:pcmaxminenvy}. 
One might ask why such running times are interesting, as the parameters $\pmax$ and $d$ are still entangled. But with an inequality that can be found in an exercise from \cite{CFKLMPPS15} (Hint 3.18) and has been used by Koutecký and Zink~\cite{KZ20}, one can bound $(\log(\Delta))^{2^{\Oh{d}}}$ by $2^{2^{\Oh{d}}}\Delta^{\oh{1}}$ (see \autoref{lem:paralgbound}). 
So \autoref{thm:pcmaxminenvy} almost answers the open question by Mnich and van Bevern~\cite{MB18} for an algorithm solving \tf{\P}{}{\cmax} that is fixed-parameter tractable (FPT) w.r.t. $d$. This question had been partially answered by Koutecký and Zink~\cite{KZ20}, who gave an $f(d)n^{\oh{1}}\enc{I}^{\Oh{1}}$-time algorithm, where $f$ is some computable function. 
Actually, one can analyze the function $(\log(\Delta))^{2^{\Oh{d}}}$ even further. Define the tower function for any $\alpha\in\N$, $\beta\in\R$ recursively as follows: $\textup{Tower}(\alpha,\beta)=\beta$ if $\alpha=0$ and $\textup{Tower}(\alpha,\beta)=2^{\textup{Tower}(\alpha-1,\beta)}$ otherwise.
For example, $\text{Tower}(1,\beta)=2^\beta$ and $\text{Tower}(2,\beta)=2^{2^\beta}$. We generalize the inequality used by Koutecký and Zink~\cite{KZ20} to the following:
\begin{lemma}\label{lem:logbound}
	For $\alpha,\beta,\gamma\in\N_{>0}$, we have $\beta^\gamma\leq \max\left\{\textup{Tower}(\alpha,\gamma)^\gamma, \beta^{\log^{(\alpha)}(\beta)}\right\}.$
\end{lemma}
\begin{proof}
	Case 1: $\gamma\geq \log^{(\alpha)}(\beta)$. Then $\textup{Tower}(\alpha,\gamma) \geq \beta$ and $\beta^\gamma\leq \textup{Tower}(\alpha,\gamma)^\gamma$.
	
	Case 2: $\gamma< \log^{(\alpha)}(\beta)$. Then $\beta^\gamma<\beta^{\log^{(\alpha)}(\beta)}$.
\end{proof}
So inserting $\beta=\log(\pmax)$, $\gamma=2^{\Oh{d}}$, we get for any $\alpha\in\N_{>0}$:
\[\log(\pmax)^{2^{\Oh{d}}}\leq \max\left\{\text{Tower}(\alpha,2^{\Oh{d}})^{2^{\Oh{d}}}, \log(\pmax)^{\log^{(\alpha+1)}(\pmax)}\right\}\]
Interestingly enough, $\log^{(\alpha+1)}(\pmax)\leq 1$ for $\alpha=4$, as long as $\pmax\leq 2^{65536}$. So one could argue that for any reasonable value of $\pmax$ ($\leq 2^{65536}$), the running time $\log(\pmax)^{2^{\Oh{d}}}\enc{I}^{\Oh{1}}$ is FPT w.r.t. $d$. Of course, it is a bit hypocritical to argue about \enquote{reasonable} values (which would be encountered in practice) if at the same time the running time then includes a tower of height $5$. Also, the assumption $\pmax\leq 2^{65536}$ directly yields an FPT running time. It still seems to us that this is an interesting observation, as this might be as close to an FPT running time as one can get.

In \autoref{sect:complexity}, we give a reduction from \tf{\Q}{}{\{\cmax,\cmin\}} to \tf{\P}{}{\{\cmax,\cmin\}}, producing instances with $d+\tau$ job types, where $\tau$ is the number of machine types. This shows that the big open question whether \tf{\P}{}{\cmax} is FPT w.r.t. parameter $d$ (see \cite{MB18}) and the question whether \tf{\Q}{}{\{\cmin} is FPT w.r.t. parameters $d$ and $\tau$ have the same answer:
\begin{restatable}{theorem}{thmfptrelated}\label{thm:fptrelated}
    The following statements are equivalent:
    \begin{enumerate}
        \item \tf{\P}{}{\{\cmax,\cmin\}} is FPT w.r.t. parameter $d$.
        \item \tf{\Q}{}{\{\cmax,\cmin\}} is FPT w.r.t. parameters $d$ and $\tau$.
    \end{enumerate}
\end{restatable}

\paragraph*{Related Work}
In their paper~\cite{GR20}, Goemans and Rothvoss showed the following result:
\begin{theorem}[Goemans \& Rothvoss~\cite{GR20}]\label{thm:goemansrothvoss}
    Given rational polyhedra $P,Q\subset\R^N$ where $P$ is bounded, one can find a vector $y\in\text{int.cone}(P\cap\Z^N)\cap Q$ and a vector $x\in\N^{P\cap\Z^N}$ such that $y=\sum_{c\in P\cap\Z^N}cx_c$ in time $\enc{P}^{2^{\Oh{N}}}\enc{Q}^{\Oh{1}}$ or decide that no such $y$ exists. Here, $\enc{P}$ and $\enc{Q}$ are the encoding lengths of $P$ and $Q$, respectively. 
\end{theorem}
They then briefly showed how this can be used directly to solve various scheduling problems. For example, they get an algorithm for \tf{\P}{}{\cmax} that runs in time $(\log(\cmax))^{2^{\Oh{d}}}\enc{I}^{\Oh{1}}$. The structural result has been generalized by Jansen and Klein:
\begin{theorem}[Jansen \& Klein~\cite{GR20}]\label{thm:jansenklein}
    Given rational polyhedra $P,Q\subset\R^N$ where $P$ is bounded, one can find a vector $y\in\text{int.cone}(P\cap\Z^N)\cap Q$ and a vector $x\in\N^{P\cap\Z^N}$ such that $y=\sum_{c\in P\cap\Z^N}cx_c$ in time $|V|^{2^{\Oh{N}}}\enc{P}^{\Oh{1}}\enc{Q}^{\Oh{1}}$ or decide that no such $y$ exists. Here, $\enc{P}$ and $\enc{Q}$ are the encoding lengths of $P$ and $Q$, respectively and $V$ is the set of vertices of the integer hull of $P$.
\end{theorem}
Kowalik \etal~\cite{KLMPS24} showed that the doubly-exponential dependency on $N$ is necessary, unless the ETH fails.
Knop \etal~\cite{KKLMO21} gave a framework for solving configuration ILPs and depending on the objective, the algorithms have different running times. As part of their framework, they also generalize the algorithm by Goemans and Rothvoss~\cite{GR20} to multiple types of P-polytopes, which allows modelling uniform and unrelated machines quite efficiently. The right-hand-sides of the systems describing the polytopes (and hence such terms as the makespan or due dates) still appear in the running time in the same way as they do in the algorithm by Goemans and Rothvoss~\cite{GR20}. Knop \etal~\cite{KKLMO19} then show how this framework can be used to obtain parameterized algorithms for various scheduling problems. For identical machines however, their running time that is not parameterized by $\pmax$ does not provide an improvement over the running time by Goemans and Rothvoss~\cite{GR20}. 
When one considers $\pmax$ a parameter, there are plenty of other algorithms, mostly using ILPs as a subroutine. Mnich and Wiese~\cite{MW15} use a balancing result in combination with solving an ILP in small dimension. Knop and Koutecký~\cite{KK18} use the normal assignment ILP and solve it with an $n$-fold algorithm. The framework by Knop \etal~\cite{KKLMO23} uses huge $n$-fold ILPs and proximity; a similar approach was described by Koutecký and Zink~\cite{KZ20}. Govzmann \etal~\cite{GMO23} bound the coefficients in the configurations and use an algorithm for ILPs with few constraints. 

As mentioned above, Koutecký and Zink~\cite{KZ20} gave an algorithm for \tf{\P}{}{\cmax} with running time $2^{2^{\Oh{d}}}\n^{\Oh{1}}$, brilliantly combining the algorithm by Goemans and Rothvoss~\cite{GR20} with the result by Frank and Tardos~\cite{FT87} to bound the coefficients in terms of $\n$ (the total number of jobs). Moreover, they stated in a footnote (without proof) that the result by Goemans and Rothvoss gives a $(\log(\pmax))^{f(d)}(\log(\n))^{\Oh{1}}$-time algorithm for \tf{\P}{}{\cmax}. The running time stated by Goemans and Rothvoss~\cite{GR20} is $(\log(\max\{\cmax,\n\}))^{f(d)}$ and it is clear that $\n$ does not have to be part of the basis, as $\n$ only appears in the encoding of $Q$, but it is not clear at all how the dependency on $\cmax$ can be avoided if one uses the algorithm by Goemans and Rothvoss. Using either of our three tools, however, yields a running time like the one claimed by Koutecký and Zink~\cite{KZ20}. As mentioned above, Mnich and van Bevern~\cite{MB18} stated the parameterized complexity of \tf{\P}{}{\cmax} in high-multiplicity encoding with parameter $d$ as an open problem.

\section{Preliminaries}\label{sect:preliminaries}
In this section, we introduce the studied problems, notation and terminology. We denote by $[k]:=\{1,\hdots,k\}$ the numbers up to $k$ and by $v_{\max}:=\norm{v}_{\infty}=\max_{i=1,\hdots,N}\{|v_i|\}$ the largest (absolute) entry in a vector $v\in\R^N$. As it is custom in parameterized complexity, we sometimes write $f(k)$ to represent \emph{any} computable function that only depends on the parameter $k$. Hence, it might not be explicitly specified and it might change, e.g.\ by consuming additional factors that depend only on $k$. A problem is fixed-parameter tractable (FPT) w.r.t. parameter $k$ if it can be solved in time $f(k)\enc{I}^{\Oh{1}}$. For any problem instance, we always assume that all input numbers are integral. This includes the matrix and right-hand-side entries in polyhedra.


\paragraph*{Applications}
To show how \autoref{thm:main} and \autoref{thm:mainvertices} can be useful, we apply them to various problems, mainly from the area of scheduling, to obtain faster algorithms.  
Following the notation introduced by Graham \etal~\cite{GLLR79}, we denote scheduling problems by a triple \tf{\alpha}{\beta}{\gamma}, where $\alpha$ describes the machine setting, $\beta$ describes a list of additional job constraints and $\gamma$ describes the objective. In the main part of this paper, we only consider the problems \tf{\P}{}{\{\cmax,\cmin,\envy\}}; other problems that appear in \autoref{sect:applications} are described later. For the problems \tf{\P}{}{\{\cmax,\cmin,\envy\}}, we are given a job vector $n\in\N^d_{>0}$, a corresponding processing time vector $p\in\N^d_{>0}$ and a number $m$ of identical machines. Formally, the task is to define an assignment $\pi:[d]\mapsto([m],[\n])$ that maps job types to machines (and multiplicities) such that the maximum load $\cmax=\max_{i\in[m]}\{L_{\pi,i}\}$ over the machines is minimized (for objective $\cmax$), the minimum load $\cmin=\min_{i\in[m]}\{L_{\pi,i}\}$ is maximized (for objective $\cmin$) or the difference between $\cmax$ and $\cmin$ is minimized (for objective $\envy$).\footnote{Note that for the problems considered in this paper, the order or starting times do not matter; only the load values (and the schedulability of early jobs) are important.} The load $L_{\pi,i}$ of a machine $i\in[m]$ w.r.t. schedule $\pi$ is defined as the sum of all processing times of the jobs assigned to $i$, i.e., $L_{\pi,i}:=\sum_{\substack{j\in[d]\\\pi(j)=(i,k)}}k p_j$. Of course, we may not schedule more jobs of a type than we are given. 
An important concept in scheduling are \emph{configurations}. A configuration is a selection of jobs (represented by a $d$-dimensional vector) that can be scheduled together on a machine. In the simple case of \tf{\P}{}{\cmax}, the set of configurations is defined as $\mathcal{C}:=\{c\in\N^d\,|\,p^Tc\leq u\}$ when all jobs have to be completed by time $u$.

As the reader might have noticed, the input to our scheduling problems is given by vectors of dimension $d$ and multiplicities. Classically, in the low-multiplicity encoding, the $\n$ jobs would all be listed together with their processing times and other characteristics, even if all jobs were identical. In this paper, all inputs are high-multiplicity encoded. The encoding length of an instance (or similarly for a polytope) $I$ is denoted by~$\enc{I}$. e.g.\ for \tf{\P}{}{\cmax}, we have $\enc{I}\leq \Oh{d\log(\pmax+\nmax)+\log(m)}$ as opposed to the classical encoding, which might have size $\Oh{\n\log(\pmax)}$. Note that one can assume that $m\leq\n$, because more than $\n$ machines are not necessary, so in the classical encoding, the part that encodes the number of machines vanishes. Depending on the given instance, $\n$ might be exponential in $\enc{I}$ and $m$ might be as well. This is the reason why the $f(d)\n^{\Oh{1}}$-time algorithm by Koutecký and Zink~\cite{KZ20} is not FPT in the high-multiplicity setting.

Note that the scheduling problems described above are optimization problems. However, the framework can only really model decision problems. Fortunately, the optimization problems can be easily solved by solving several decision problems as part of a binary search over some interval $[0,v]$. The key ingredient is that the optimum is always bounded by some value $v$ that depends on the numbers in the input and hence $\log(v)\leq\enc{I}^{\Oh{1}}$.\footnote{Note e.g.\ in the case of \tf{\P}{}{\cmax} that $\log(v)=\log(\sum_{j=1}^d p_jn_j)\leq\log(d\pmax\nmax)\leq\enc{I}^{\Oh{1}}$.} 

Throughout this paper, we only consider the decision problems, and as our algorithms all have an $\enc{I}^{\Oh{1}}$-term in the running time anyway, this directly also yields algorithms with the same running time for the optimization problem, except in the $\envy$-case. There, we get an additional factor, which is $\pmax$. Hence, we will always assume that we are given thresholds $\ell$ and/or $u$ for the load values when solving these problems. Clearly, (the logarithms of) these are then also part of the encoding. A detailed justification for only considering decision problems is given in \autoref{sect:omittedproofs}.

\paragraph*{Reducing Coefficients}
If the coefficients of an inequality (or equality) are much larger compared to the number of variables and the variables are also bounded by a rather small value, the coefficients can be reduced with the following classical result by Frank and Tardos: 
\begin{theorem}[Frank \& Tardos~\cite{FT87}]\label{thm:franktardos}
    For every $w\in\R^N$ and $\Delta\in\N$, there exists a $\bar{w}\in\Z^N$ such that $\norm{\bar{w}}_{\infty}\leq (N\Delta)^{\Oh{N^3}}$ and $\text{sign}(w^Tx)=\text{sign}(\bar{w}^Tx)$ for every $x\in\Z^N$ with $\norm{x}_1\leq \Delta-1$. Moreover, $\bar{w}$ can be computed in time $N^{\Oh{1}}$.
\end{theorem}

\begin{restatable}[Stated in a similar form by Etscheidt \etal~\cite{EKMR17}]{corollary}{franktardos}\label{cor:franktardos}
    For every $w\in\R^N$, $b\in\R$, $\Delta\in\N$, one can compute $\bar{w}\in\Z^N$, $\bar{b}\in\Z$ with $\norm{\bar{w}}_{\infty},|\bar{b}|\leq (N\Delta)^{\Oh{N^3}}$ in time $N^{\Oh{1}}$ such that for every $x\in[-\Delta,\Delta]^N$, $w^Tx\leq b \iff \bar{w}^Tx\leq \bar{b}$.
\end{restatable}
A proof is given in \autoref{sect:omittedproofs}.
Note that due to the equivalence, zeros are mapped to zeros, negative entries are mapped to negative entries and positive entries are mapped to positive entries. This can be seen by setting $x$ to be a unit vector in \autoref{thm:franktardos}.


Another way to greatly reduce a given problem instance (and the set of configurations) is provided by the following result by Govzman \etal~\cite{GMO23}. It is the key ingredient used in \autoref{sect:additiveapprox} to improve the running time of the additive approximation scheme from~\cite{buchem} in a high-multiplicity setting and can also be used to reduce the coefficients in a \pqmr. For completeness, we include a proof in \autoref{sect:omittedproofs}, as the result has not been published yet. A similar but exponential bound had already been shown by Mnich and Wiese~\cite{MW15}.
\begin{restatable}[Govzmann \etal \cite{GMO23}]{lemma}{balancinglemma}\label{lem:balancing}
	For \tf{\P}{}{\{\cmax,\cmin,\envy\}}, there exists a kernel where the number of jobs of a specific type on a specific machine is bounded by $2 \pmax$. So the load of every machine is bounded by $2 \pmax^2d$. The kernelization runs in $\O(d)$ time. 
\end{restatable}

\paragraph*{\pqmrs}
\enquote{\pqmrs} are formally defined as follows:
\begin{definition}[PQ-Representation]\label{def:pqrep}
    Given a problem $X$, a \emph{\pqmr} of $X$ is a triple $(P,Q,m)$ of polytopes $P,Q\subset\R^N$ and a number $m$ such that for every instance $I$ of $X$, $I$ is positive if and only if there exists a $y=y^{(1)}+\hdots+y^{(m)}\in\text{int.cone}(P\cap\Z^N)\cap Q$ such that for all $i\in[m]$, $y^{(i)}\in P\cap\Z^N$.
\end{definition}
Problems that have a \pqmr can be solved with the algorithm by Goemans and Rothvoss~\cite{GR20} or the algorithm by Jansen and Klein~\cite{JK20}:
\begin{restatable}[See \cite{JK20} and \cite{GR20}]{proposition}{thmpqrep}\label{prop:jkgr}
    A problem with a \pqmr $(P,Q,m)$, where $P,Q\subset\R^N$, can be solved in $|V|^{2^{\Oh{N}}}\enc{P}^{\Oh{1}}\enc{Q}^{\Oh{1}}(\log(m))^{\Oh{1}}$ or in $\enc{P}^{2^{\Oh{N}}}\enc{Q}^{\Oh{1}}(\log(m))^{\Oh{1}}$,
    where $V$ is the set of vertices of the integer hull of $P$.
\end{restatable}

\paragraph*{Bounds on the Vertices of the Integer Hull}
Clearly, the number of vertices of the integer hull plays a big role in the running time of \autoref{prop:jkgr}. Formally, the integer hull of a polytope $P\subset\R^N$ is defined as the convex hull of $P\cap\Z^N$. There are several upper and lower bounds for the number of vertices of the integer hull, e.g.\ the one by Berndt \etal~\cite{BJK21}, which improved the one by Aliev \etal~\cite{ADEOW18}:
\begin{theorem}[Berndt \etal~\cite{BJK21}]\label{thm:berndtgeneral}
    The number of vertices of the integer hull of a polytope $P=\left\{\left.x\in\R^N_{\geq0}\,\right|\,Ax= b\right\}$ is at most $(NM\log(M\Delta))^{\Oh{M\log(\sqrt{M}\Delta)}},$ where $\Delta=\norm{A}_\infty$ and $M$ is the number of constraints.
\end{theorem} 
Note that this does not depend on the right-hand-side but has a stronger dependency on the largest coefficient in the matrix compared to the older result by Hartmann~\cite{Hartmann89} and Cook \etal~\cite{CHKM92}, which appears to be nearly optimal in general~\cite{Zolotykh06}:
\begin{theorem}[Hartmann~\cite{Hartmann89} and Cook \etal~\cite{CHKM92}]\label{thm:hartmann}
    Consider a rational polyhedron $P=\{x\in\R^N\,|\,Ax\leq b\}$ with $A\in\Z^{M\times N}$ and let $\Delta=\norm{A}_{\infty}$. Then the integer hull of $P$ has at most $(MN\phi)^{\Oh{N}}=(MN\log(\max\{\Delta,\norm{b}_{\infty}\}))^{\Oh{N}}$ vertices, where $\phi=\Oh{N\log(\max\{\Delta,\norm{b}_{\infty}\})}$ is the maximum encoding length of the inequalities.
\end{theorem}

\section{Tool 1: The Balancing Result by Govzmann \etal}\label{sect:toolone}

How \autoref{prop:jkgr} can be useful for solving scheduling and packing problems can be seen at the example of \tf{\P}{}{\{\cmax,\cmin,\envy\}}. Remember that we are given load thresholds $\ell$ and/or $u$, depending on the objective: 
\begin{restatable}{lemma}{pqpcmaxminenvy}\label{lem:pqpcmaxminenvy}
    \tf{\P}{}{\{\cmax,\cmin,\envy\}} has a \pqmr with $m$ set to the number of machines, $Q_{\cmax}=\left\{\left.\begin{pmatrix}c\\t\end{pmatrix}\in\R_{\geq0}^{d+1}\,\right|\,c=n\right\}$, \\
    $Q_{\cmin}=Q_{\envy}=\left\{\left.\matrix{c\\t_1\\t_2}\in\R_{\geq0}^{d+2}\,\right|\,c=n\right\}$ and respectively
    
    \begin{align*}
        &P_{\cmax}=\left\{\left.\matrix{c\\t}\in\R_{\geq0}^{d+1}\,\right|\,p^Tc+t=u\right\},\\
        &P_{\cmin}=\left\{\left.\matrix{c\\t_1\\t_2}\in\R_{\geq0}^{d+2}\,\right|\,\matrix{p^T & 1 & 0 \\ -p^T & 0 & 1}\matrix{c\\t_1\\t_2}= \matrix{\ell+2\pmax \\ -\ell}\right\} \quad \text{and}\\
        &P_{\envy}=\left\{\left.\matrix{c\\t_1\\t_2}\in\R_{\geq0}^{d+2}\,\right|\,\matrix{p^T & 1 & 0 \\ -p^T & 0 & 1}\matrix{c\\t_1\\t_2}= \matrix{u \\ -\ell}\right\}.
    \end{align*}
\end{restatable}
\begin{proof}
    We only show the claim for \tf{\P}{}{\cmax}, so let $P:=P_{\cmax}$ and $Q:=Q_{\cmax}$. The other proofs work analogously, with the only difference being that we can restrict ourselves to configurations that have load at most $\ell+2\pmax$ in the case of $\cmin$, as we show in \autoref{sect:omittedproofs}.
    First note that the $t$-entry fulfills the role of a slack variable, making it so we have an equality constraint.
    Clearly, $P$ is bounded. Consider an instance $I$ of \tf{\P}{}{\cmax} and suppose that it is positive, i.e., there is a schedule of all the jobs with makespan at most $u$. For each machine $k\in[m]$, create a $(d+1)$-dimensional vector $y^{(k)}$ such that $y^{(k)}_i$ is the number of jobs with processing time $p_i$ scheduled on machine $k$ for all $i\in[d]$ and $y^{(k)}_{d+1}=u-\sum_{i=1}^dp_ic_i$. Clearly, summing up the all the vectors $y^{(k)}$ projected to the first $d$ variables yields $\begin{pmatrix}n_1&\hdots&n_d\end{pmatrix}^T$, since all jobs are scheduled. Moreover, we constructed $m$ such vectors. Since the $y^{(k)}$ are integral, we have $\sum_{k=1}^my^{(k)}\in\textup{int.cone}(P\cap\Z^{d})\cap Q$.

    Conversely, suppose there is a vector $y\in\textup{int.cone}(P\cap\Z^{d})\cap Q$. Since $y\in\textup{int.cone}(P\cap\Z^{d})$, $y$ can be written as a conic combination of integral vectors in $P$. Let $S$ be the multi-set of these integral vectors, i.e., $\sum_{y^{(k)}\in S}y^{(k)}=y$. Each $y^{(k)}\in S$ corresponds to an integral $d+1$-dimensional vector. For each $y^{(k)}\in S$, pick an (idle) machine $k\in[m]$ and for every $i\in[d]$, schedule $y^{(k)}_i$ jobs of type $i$ on $k$. Since $y$ is a sum of $m$ vectors, we know that $|S|=m$ and since $y\in Q$, $\sum_{y^{(k)}\in S}y^{(k)}_i=n_i$ holds for every $i\in[d]$. So indeed we schedule all the jobs and have enough machines to do so. Since each $y^{(k)}\in S$ also lies in $P$, the jobs scheduled on the same machine have total processing time at most $u$. So the constructed schedule is feasible and the given instance $I$ is positive.
\end{proof}

So a straightforward application of \autoref{prop:jkgr} to the above formulations would yield algorithms running in time
$|V|^{2^{\Oh{d}}}\enc{P}^{\Oh{1}}\enc{Q}^{\Oh{1}}(\log(m))^{\Oh{1}}$, where $|V|$ is the number of vertices of $P$'s integer hull or in the other case running time
$\enc{P}^{2^{\Oh{d}}}\enc{Q}^{\Oh{1}}(\log(m))^{\Oh{1}}$.
Note however that without further consideration, bounds for $|V|$ either contain a term $(\log(\pmax))^{\Oh{\log(\pmax)}}$ or $\log(\ell)\leq\log(u)$, depending on whether we use \autoref{thm:berndtgeneral} or \autoref{thm:hartmann}. Moreover, $\enc{P}$ also depends on $\log(\ell)\leq\log(u)$. So these do not yet yield a running time $(\log(\pmax))^{2^{\Oh{d}}}\enc{I}^{\Oh{1}}$. This is where the balancing result by Govzmann \etal (\autoref{lem:balancing}) can help: By paying $\Oh{d}$ time (which vanishes in the other terms), we can assume that $\ell\leq u\leq 2d\pmax^2$, as there are at most $2\pmax$ jobs of each type on any machine. Then the running time $\enc{P}^{2^{\Oh{d}}}\enc{Q}^{\Oh{1}}(\log(m))^{\Oh{1}}$ becomes
\begin{align*}
    (d\log(\pmax)+\log(u))^{2^{\Oh{d}}}\enc{I}^{\Oh{1}}
    &= (d\log(\pmax)+\log(2d\pmax^2))^{2^{\Oh{d}}}\enc{I}^{\Oh{1}}
\end{align*}
and with the considerations about optimization (see \autoref{sect:omittedproofs}) in mind, we get the following result:
\pcmaxminenvybalancing*

\section{Tool 2: Proximity}\label{sect:tooltwo}
In some cases, we might not have such a fine preprocessing algorithm like the one by Govzmann~\cite{GMO23} for our problem but we still want to reduce the encoding length of $P$. In this section, we show how a similar preprocessing can be done for general \pqmrs. Each \pqmr also has an equivalent ILP and solving a relaxation of this ILP allows us to remove the dependency on the right-hand-side (e.g.\ $u$ in the case of \tf{\P}{}{\cmax}):
\begin{lemma}\label{lem:pqilp}
    Let $(P,Q,m)$ be a given \pqmr with $P=\left\{\left.x\in\R_{\geq0}^{N}\,\right|\,A^{(P)}x=b^{(P)}\right\}$ and $Q=\left\{\left.x\in\R_{\geq0}^{N}\,\right|\,A^{(Q)}x=b^{(Q)}\right\}$. Then the points $y\in\textup{int.cone}(P\cap\Z^N)\cap Q$ have a one-to-one correspondence to the solutions of the following ILP, projected to block $m+1$ and the other blocks correspond to vectors $y^{(1)},\hdots,y^{(m)}\in P\cap\Z^N$ such that $y=y^{(1)}+\hdots+y^{(m)}$:
    \begin{align*}
        \begin{pmatrix}
        \;
        \begin{array}{@{}*{5}{c}@{}}
        \cline{1-1}\cline{2-2}\cline{4-4}\cline{5-5}
        \lrr{0} & \rr{0} & \hdots & \lrr{0} & \rr{A^{(Q)}}\\
        \lrr{I} & \rr{I} & \hdots & \lrr{I} & \rr{-I} \\
        \cline{1-1}\cline{2-2}\cline{4-4}\cline{5-5}
        \lrr{A^{(P)}} &  &  &  & \\
        \cline{1-1}\cline{2-2}
         & \lrr{A^{(P)}} &  &  & \\
         \cline{2-2}
         &  & \ddots &  & \\
         \cline{4-4}
         &  &  & \lrr{A^{(P)}} & \\
         \cline{4-4}\cline{5-5}
         &  &  &  & \lrr{0} \\
         \cline{5-5}
         \end{array}
         \;\;
        \end{pmatrix}
        x &= \matrix{b^{(Q)}\\\mathbf{0}\\b^{(P)}\\b^{(P)}\\\vdots\\b^{(P)}\\\mathbf{0}}\\
        x&\in\N^{(m+1)N}
    \end{align*}
\end{lemma}
\begin{proof}
    Suppose $y=y^{(1)}+\hdots+y^{(m)}\in\textup{int.cone}(P\cap\Z^N)\cap Q$ such that for all $i\in[m]$, $y^{(i)}\in P\cap\Z^N$. Then $x:=\matrix{y^{(1)}&\hdots&y^{(m)}&y}^T$ is a solution of the above ILP. 
    Now suppose we have a solution $x=\matrix{y^{(1)}&\hdots&y^{(m)}&y}^T$. Then by the local, non-negativity and integrality constraints, we have $y^{(i)}\in P\cap\Z^N$ for all $i\in[m]$. By the first set of constraints, we have $y\in Q$ and by the second set of constraints, we have $y=y^{(1)}+\hdots+y^{(m)}$. Hence, $y\in\textup{int.cone}(P\cap\Z^N)\cap Q$.
\end{proof}

The ILP in \autoref{lem:pqilp} has a so-called $n$-fold structure, as indicated by the framing (see e.g.~\cite{EHK18}).
We use a proximity result by Clsovjecsek \etal~\cite{CEW20} to bound the vectors in $P$, which allows us to reduce the coefficients in the right-hand-side.\footnote{We cite the first version of their paper here, because it is much easier to see how the approach for solving the relaxation can be adapted to the case where many blocks are identical.} The key idea is to solve a relaxation of the ILP corresponding to the \pqmr.
In general, this \textit{convexified relaxation} of an $n$-fold ILP looks like this:
\begin{align*}
    \max\,\, &\sum_{i=1}^n(c^{(i)})^Tx^{(i)} \\
    \sum_{i=1}^n C^{(i)}x^{(i)} &= b^{(0)}\\
    x^{(i)}\in Q^{(i)}&:=\text{int.hull}\left(\{x\in\R^t\,|B^{(i)}x=b^{(i)},\,\mathbf{0}\leq x\leq u^{(i)}\}\right) \quad \forall i\in[n]
\end{align*}
where the $C^{(i)}\in\Z^{r\times t}$ are the blocks at the top, $b^{(0)}\in\Z^r$ is the corresponding right-hand-side, the $B^{(i)}\in\Z^{s\times t}$ are the blocks on the diagonal, the $b^{(i)}\in\Z^s$ are the corresponding right-hand-sides, the $c^{(i)}\in\Z^t$ are the parts of the objective function and the $u\in\Z^t$ are upper bounds for the variables.
Cslovjecsek \etal~\cite{CEW20} show that the Lagrangean dual (w.r.t. the first $r$ constraints) of this problem can be solved with at most $r^{\Oh{1}}\log^2(ntK)$ calls to a separation oracle, where $K$ is the largest number appearing in the objective function, matrix, right-hand-side and upper bounds. 
The separation oracle has to solve
\begin{align*}
    L(\lambda)=\max_{\substack{x^{(i)}\in Q^{(i)}\\i=1,\hdots,n}}\left\{\sum_{i=1}^n(c^{(i)})^Tx^{(i)}+\lambda^T(\sum_{i=1}^nC^{(i)}x^{(i)}-b^{(0)})\right\}
\end{align*}
for a given $\lambda\in\R^r$. In general, this can be done by solving ILPs corresponding to the $n$ blocks $Q^{(i)}$ separately, which may take time $n t^{\Oh{t}} (st\log(K))^{\Oh{1}}$ using the algorithm by Kannan~\cite{Kannan87} $n$ times.\footnote{Cslovjecsek \etal~\cite{CEW20} use the algorithm by Eisenbrand and Weismantel~\cite{EW19} instead, but here it makes more sense to use the algorithm by Kannan~\cite{Kannan87}, as our running time will be exponential in $t$ anyway.} However, note that that if for two blocks $i$ and $j$, we have $c^{(i)}=c^{(j)}$, as well as $B^{(i)}=B^{(j)}$, $C^{(i)}=C^{(j)}$, $u^{(i)}=u^{(j)}$ and $b^{(i)}=b^{(j)}$, the ILP has the same optimal solution(s). 

When we have computed a solution of the Lagrangean dual, we still have to compute a solution of the primal, i.e., the convexified relaxation. This part is described in more detail in the second version of the paper by Cslovjecsek \etal~\cite{CEHRW20}: While solving the dual with the cutting plane algorithm, each call to the separation oracle yields an optimizer $z^{(j)}\in\Z^{nt}$ of the maximum in iteration $j$. One then only has to find a vector $x^*\in\R^{nt}$ in the convex hull of the $z^{(j)}$ that also fulfills the linking constraints $\sum_{i=1}^n C^{(i)}x^{(i)} = b^{(0)}$ and maximizes the objective $\max\sum_{i=1}^n(c^{(i)})^Tx^{(i)}$. This LP has many variables (equal to the number of calls to the separation oracle) but only $r+1$ constraints. So it can be solved e.g.\ with the algorithm by Megiddo~\cite{Megiddo84} (applied to the dual of this LP; it also provides a primal solution) in a time that is polynomial in the number of oracle calls and $2^{\Oh{r^2}}$.
Altogether, for $\tau$ different blocks, this yields running time
    \begin{align*}
        &\underbrace{r^{\Oh{1}}\log^2(ntK)}_{\text{number of oracle calls}}\underbrace{\tau t^{\Oh{t}} (st\log(K))^{\Oh{1}}}_{\text{time for one oracle call}}+\underbrace{(r^{\Oh{1}}\log^2(ntK))^{\Oh{1}}2^{\Oh{r^2}}}_{\text{time for solving the LP}}\\
        \leq&\tau t^{\Oh{t}}2^{\Oh{r^2}} (sr\log(Kn))^{\Oh{1}}
    \end{align*}
and we get the following result:
\begin{proposition}\label{prop:relaxation}
    Let $\tau$ be the number of different $\left(c^{(i)},B^{(i)},C^{(i)},b^{(i)},u^{(i)}\right)$-pairs, $n$ the number of blocks, $r$, $s$, $t$ the block dimensions and $K$ the largest number appearing in the objective function, matrix, right-hand-side and the upper bounds. Then the convexified relaxation can be solved in time
    $$\tau t^{\Oh{t}}2^{\Oh{r^2}} (sr\log(Kn))^{\Oh{1}}.$$
\end{proposition}
Cslovjecsek \etal~\cite{CEHRW20} show that the convexified relaxation has nice proximity properties, namely they show the following:
\begin{theorem}[Cslovjecsek \etal~\cite{CEHRW20}]\label{thm:convrelproximity}
    For every solution $x^*$ of the convexified relaxation, there exists a solution $z^*$ of the corresponding integer program such that $\norm{x^*-z^*}_1\leq(r\Delta (s\Delta)^{\Oh{s}})^{\Oh{r}}$, where 
    $\Delta$ is the largest absolute value in the matrix.
\end{theorem}

By first solving the convexified relaxation using \autoref{prop:relaxation} and then applying the proximity bound from \autoref{thm:convrelproximity}, we can bound the solutions of the ILP. For a full proof, see \autoref{sect:omittedproofs}.
\begin{restatable}{theorem}{balancinggeneral}\label{thm:balancinggeneral}
    Let $P=\left\{\left.x\in\R_{\geq0}^N\,\right|\,A^{(P)}x=b^{(P)}\right\}$ with $M^{(P)}$ constraints and $Q=\left\{\left.x\in\R_{\geq0}^N\,\right|\,A^{(Q)}x=b^{(Q)}\right\}$ with $M^{(Q)}$ constraints, and consider the \pqmr $(P,Q,m)$. Then by solving the convexified relaxation of the corresponding ILP in time 
    \begin{align*}
        &2^{\Oh{(M^{(Q)}+N)^2}} \left(M^{(P)}\log\left(\max\left\{\norm{A^{(P)}}_\infty, \norm{A^{(Q)}}_\infty, \right.\right.\right.\\
        &\left.\left.\left.\norm{b^{(P)}}_\infty,\norm{b^{(Q)}}_\infty\right\}m\right)\right)^{\Oh{1}}
    \end{align*}
    and modifying $P$, we can assume that the points in $P$ have $\ell_\infty$-norm at most
    \[\left((M^{(Q)}+N)M^{(P)}\max\left\{\norm{A^{(P)}}_\infty, \norm{A^{(Q)}}_\infty\right\}\right)^{\Oh{(M^{(Q)}+N)M^{(P)}}}.\]
\end{restatable}

We now have all the tools to prove \autoref{thm:main}. The key idea is to combine \autoref{thm:balancinggeneral} with the result by Frank and Tardos (\autoref{cor:franktardos}). This reduces the coefficients in the inequalities describing $P$. In particular, it makes their size independent of the original right-hand-side. For a full proof, see \autoref{sect:omittedproofs}.
\maintheorem*

Note how this result also yields an algorithm for \tf{\P}{}{\{\cmax,\cmin,\envy\}} with running time $(\log(\pmax))^{2^{\Oh{d}}}\enc{I}^{\Oh{1}}$, but with worse constants.

\section{Tool 3: More Fitting Bounds for the Vertices of the Integer Hull}\label{sect:toolthree}
As mentioned above, given a polytope $P=\{x\in\R^N\,|\,Ax\leq b\}$, there are upper bounds for the number of vertices of $P$'s integer hull. But they either depend on $\log(\norm{b}_\infty)$ (like \autoref{thm:hartmann}) or are exponential in $\log(\norm{A}_\infty)$ (like \autoref{thm:berndtgeneral}). For using the algorithm by Jansen and Klein~\cite{JK20} -- which takes time $|V|^{2^{\Oh{d}}}\enc{P}^{\Oh{1}}\enc{Q}^{\Oh{1}}\log(m)^{\Oh{1}}$ to solve a \pqmr -- it would be nice to have a bound that does not depend on $\log(\norm{b}_\infty)$ but is also not exponential in $\log(\norm{A}_\infty)$. Fortunately, the proof by Berndt \etal~\cite{BJK21} can be adapted to get the following result:
\vertexboundimproved*
\begin{proof} (Sketch).
    Denote by $P_I$ the integer hull of $P$. It is well-known that $P$ has at most $\binom{N}{M}$ vertices (the basic feasible solutions). Eisenbrand and Weismantel~\cite{EW19} have shown that for each vertex $z^*$ of $P_I$, there exists a vertex $x^*$ of $P$ such that $\norm{x^*-z^*}_1\leq M(2M\Delta+1)^M$.\footnote{To be precise, they showed that for each optimal solution of the associated LP, there exists an optimal solution of the corresponding ILP with this distance. But by perturbing the linear objective function, every vertex of $P_I$ can be made a unique optimum.} 
    The idea of the proof by Berndt \etal~\cite{BJK21} is to use this proximity result to remove the dependency on $\norm{b}_\infty$. Each vertex of $P_I$ is associated with at least one vertex of $P$; if we were to draw balls around each vertex of $P$ with the radius set to the proximity bound, each vertex of $P_I$ would be in at least one of these circles. Instead of rasterizing the whole $P_I$, Berndt \etal only rasterize the regions around the vertices of $P$. This yields at most $\Oh{\log(M(2M\Delta+1)^M)}^N=\Oh{M\log(M\Delta)}^N$ boxes per vertex of $P$, each of which contains at most one vertex of $P_I$. So $P_I$ has at most
    \begin{align*}
        \binom{N}{M}\Oh{M\log(M\Delta)}^N&\leq N^M\Oh{M\log(M\Delta)}^N
    \end{align*}
    vertices. Berndt \etal then proceed to reduce the exponent so that it does not depend on $N$, which might in some cases be rather large compared to $\Delta$ and $M$.
\end{proof}

We can use this result to solve \pqmrs with the algorithm by Jansen and Klein~\cite{JK20}:
\mainvertices*
\begin{proof}
    We use the $|V|^{2^{\Oh{N}}}(\enc{P}\enc{Q}\log(m))^{\Oh{1}}$-time algorithm by Jansen and Klein~\cite{JK20} (see \autoref{prop:jkgr}), so we need to bound $|V|$, the number of vertices of $P$'s integer hull. For this, we use \autoref{thm:vertexboundimproved},
    which yields
    \begin{align*}
        |V|&\leq N^{M^{(P)}}O\left(M^{(P)}\log\left(M^{(P)}\norm{A^{(P)}}_\infty\right)\right)^{N}.
    \end{align*}
    In total, we get the following running time:
    \begin{align*}
        &|V|^{2^{\Oh{N}}}(\enc{P}\enc{Q}\log(m))^{\Oh{1}} \\
        =&\left(N^{M^{(P)}}O\left(M^{(P)}\log\left(M^{(P)}\norm{A^{(P)}}_\infty\right)\right)^{N}\right)^{2^{\Oh{N}}}(\enc{P}\enc{Q}\log(m))^{\Oh{1}}\\
        =&\left(N^{M^{(P)}}M^{(P)}\log\left(\norm{A^{(P)}}_\infty\right)\right)^{2^{\Oh{N}}}(\enc{P}\enc{Q}\log(m))^{\Oh{1}}
    \end{align*}
\end{proof}
Observe that this also yields an algorithm for \tf{\P}{}{\{\cmax,\cmin\envy\}} with running time $(\log(\pmax))^{2^{\Oh{d}}}\enc{I}^{\Oh{1}}$ (again, with worse constants than in \autoref{thm:pcmaxminenvy}).

\section{Conclusion and Future Work}\label{sect:conclusion}
In this work, we investigated three tools that improve the running time of the algorithm by Goemans and Rothvoss~\cite{GR20}. The first tool could be considered a problem-specific preprocessing, the second tool uses a special relaxation of an $n$-fold ILP and proximity and the third tool utilizes the framework by Jansen and Klein~\cite{JK20} together with new customized bounds for the vertices of the integer hull. Overall, each tool provides a $(\log(\pmax))^{2^{\Oh{d}}}\enc{I}^{\Oh{1}}$-time algorithm for \tf{\P}{}{\{\cmax,\cmin,\envy\}}, which is FPT w.r.t. $d$ if the processing times are given in unary. The second and third tool also yield similar algorithms for other problems.
The second tool provides a bound for the variables in $P$. In some cases, there are natural bounds for the variables in $P$ already, e.g.\, $\n$ for \tf{\P}{}{\{\cmax,\cmin,\envy\}}. Of course, one can also use such bounds instead of proximity, if they produce a better result.
On another note, it would be interesting to apply the improved framework to other scheduling and packing problems like the ones considered by Goemans and Rothvoss~\cite{GR20} and to see how eliminating the right-hand-side influences the running times. Another very interesting direction is incorporating our ideas into the framework by Knop \etal~\cite{KKLMO21} to handle multiple types of $P$-polytopes. This might then extend our results to scheduling on uniform and unrelated machines with and without deadlines or release dates, parallel task scheduling, 2-dimensional packing as well as other objectives in an efficient way.

%
%
%
\newpage
\bibliographystyle{splncs04}
%
\bibliography{bib}

\newpage

\appendix
\section{Applications}\label{sect:applications}

We now apply \autoref{thm:main} to \tf{\P}{\dj,\vec,\class,\capacity}{\{\cmax,\cmin,\envy\}} (a quite general problem with interesting special cases), the objective $\sumwu$, setup times, \mswbp and uniform $n$-fold ILPs.

The other objective we consider in this section is minimizing the weighted number of late jobs ($\sumwu$). Each job has a due date $d_j$ and a weight $w_j$ and the goal is to find a schedule such that the weight of the jobs missing their due date is minimized. Since it does not matter by how much late jobs miss their due date, they might as well not be scheduled at all (or way after the other jobs) and we may focus on finding a set of jobs that can be scheduled early. For those jobs, it then only matters on which machine they are scheduled, not in which order, as it is well-known that they may be scheduled in EDD (earliest due date) order~\cite{LS20}.

We also consider several job characteristics: Deadlines ($\dj$), multidimensional jobs ($\vec$), class constraints ($\class$), setup times ($\setup$) and capacity constraints ($\capacity$). If there are deadlines (not to be confused with the due dates from the $\sumwu$-objective), each job has a deadline $d_j$ and may not finish after it. In the $\vec$-variant (also referred to as \emph{vector scheduling}), the processing times $p_j\in\N^M_{>0}$ are $M$-dimensional vectors. If class constraints are present, each job has a class $i\in[K_{\#}]$, we are given a number $C^{\text{class}}\in\N_{>0}$ and there may only be jobs of at most $C^{\text{class}}$ different classes scheduled on each machine. This constraint has e.g.\ been studied in~\cite{JLM20}. Similarly, with setup times, each job has a class and starting a job from a new class invokes a setup time $s_j$ (see e.g.~\cite{DJ19}). The capacity constraints just forbid more than $C^{\text{cap}}\in\N_{>0}$ jobs to be scheduled on any machine. This has been studied in~\cite{CJWG16}.

Note that we can always apply \autoref{lem:logbound} to the running times stated in the theorems of this section, like we did in \autoref{sect:introduction} for \tf{\P}{}{\{\cmax,\cmin,\envy\}}.

We can simplify the running time of \autoref{thm:main} and \autoref{thm:mainvertices}. Usually, $N$ should be the largest of the three matrix dimensions, anyway:

\begin{corollary}\label{cor:pqviavertices}
    If a problem has a \pqmr $(P,Q,m)$ with $P=\{x\in\R_{\geq0}^{N}\,|\,A^{(P)}x= b^{(P)}\}$ and $Q=\{x\in\R_{\geq0}^{N}\,|\,A^{(Q)}x= b^{(Q)}\}$ with $A^{(P)}\in\Z^{M^{(P)}\times N}$, $A^{(Q)}\in\Z^{M^{(Q)}\times N}$, $b^{(P)}\in\Z^{M^{(P)}}$ and $b^{(Q)}\in\Z^{M^{(Q)}}$, it can be solved in time $$
        (\log(\norm{A^{(P)}}_\infty))^{2^{\Oh{K}}}(\enc{P}\enc{Q}\log(m))^{\Oh{1}},$$
    where $K:=\max\{N,M^{(P)},M^{(Q)}\}$.
\end{corollary}

\subsection{A Quite General Scheduling Problem}
In \tf{\P}{\dj,\vec,\class,\capacity}{\{\cmax,\cmin,\envy\}} we are given a lower bound $\ell$ ($\cmin$), an upper bound $u$ for the load of each machine ($\cmax$), or both ($\envy$). So the goal is to schedule all $\n$ jobs on the $m$ machines while satisfying the machine load bounds and several other constraints: 
\begin{itemize}
    \item Jobs are $M$-dimensional, i.e., their processing times $p_j$ are $M$-dimensional vectors, as are the lower and upper bounds on the machine loads,
    \item each job $j$ has a deadline $d_j$ (which is \emph{not} an $M$-dimensional vector; otherwise, we might not be able to order jobs by their due date) it has to be finished before,
    \item each job has a class $i\in[K_{\#}]$ and there may at most be $C^{\textup{class}}$ different classes represented on each machine and
    \item each machine can only schedule $C^{\textup{cap}}$ jobs.
\end{itemize}
So the input consists of the load bounds $\ell$ and $u$ (depending on the objective), the number of machines, the capacities $C^{\textup{class}}$ and $C^{\textup{cap}}$ and finally the information about the jobs: For each distinct tuple of deadline, processing time vector and class, we have a multiplicity $n_i$. In this context, the number of job types $d$ refers to the number of these distinct tuples. So $d$ might be as large as $d_{\#}\pmax^MK_{\#}$, where $d_{\#}$ is the number of different deadlines, $\pmax$ is the largest entry in all the processing time vectors, $M$ is the dimension of the processing time vectors and $K_{\#}$ is the number of classes. We assume that the deadlines are ordered non-decreasingly, i.e., $d_1$ is the smallest deadline. 

\begin{lemma}\label{lem:pqcomplex}
    \tf{\P}{\dj,\vec,\class,\capacity}{\{\cmax,\cmin,\envy\}} has a \pqmr.
\end{lemma}
\begin{proof}
    Set $m$ to the number of machines. Let $N:=d+K_{\#}+K_{\#}+(2+d_{\#}M+2M+3K_{\#})$ be the dimension of $P$, which is defined as
    \[P=\left\{\left.\begin{pmatrix}c\\x\\z\\s\end{pmatrix}\in\R_{\geq0}^{N}\,\right|\,A\begin{pmatrix}c\\x\\z\\s\end{pmatrix}=b\right\}\] and let \[Q=\left\{\left.\begin{pmatrix}c&x&z&s\end{pmatrix}^T\in\R_{\geq0}^N\,\right|\,c=n\right\}.\]
    We describe $A$ and $b$ by a system of equalities and inequalities. The part $s\in\R^{2+d_{\#}M+2M+3K_{\#}}_{\geq0}$ of the vectors in $P$ are slack variables that we will not explicitly write down for the sake of visibility. Just note that there are $2+d_{\#}M+2M+3K_{\#}$ constraints that have inequalities instead of equalities.
    \begin{align*}
        \mathbf{1}^Tc&\leq C^{\textup{cap}} && &&&\textup{(cap)} \\
        \sum_{k=1}^j(p_{k,i})^Tc^{(k)}&\leq d_j && \forall j\in[d_{\#}],i\in[M] &&&\textup{($d_j$)} \\
        \mathbf{1}^Tc^{(i)}&=x_i && \forall i\in[K_{\#}] &&&\textup{(class)}\\
        x_i\n&\geq z_i && \forall i\in[K_{\#}] &&&\textup{(class)}\\
        z_i\n&\geq x_i && \forall i\in[K_{\#}] &&&\textup{(class)}\\
        z_i&\leq 1 && \forall i\in[K_{\#}] &&&\textup{(class)}\\
        \sum_{i=1}^{K_{\#}}z_i&\leq C^{\textup{class}} && &&&\textup{(class)}\\
        p_i^Tc&\leq \left\{u_i, \sum_{j=1}^dp_{j,i}\right\} && \forall i\in[M] &&&\textup{(obj.)}\\
        p_i^Tc&\geq \left\{\ell_i, 0\right\} && \forall i\in[M] &&&\textup{(obj.)}
    \end{align*}
    Here -- slightly abusing notation -- $c^{(k)}$ is the sub-vector of $c$ that corresponds to job types that have due date $d_k$, while $c^{(i)}$ is the sub-vector of $c$ that corresponds to job types that have class $i$. The processing time matrix corresponding to the jobs of due date $d_k$ is denoted by $p_k$ and $p_{k,i}$ refers to the processing times of the jobs with due date $k$ in dimension $i$. The actual values in the objective constraints depend on the objective function: For $\cmax$, we have $0\leq p_i^Tc\leq u_i$; for $\cmin$, we have $\ell_i\leq p_i^Tc\leq \sum_{j=1}^dp_{j,i}$ and for $\envy$, we have $\ell_i\leq p_i^Tc\leq u_i$. 
    
    To prove that $(P,Q,m)$ forms a \pqmr for our problem, we first argue that we have $y=\begin{pmatrix}c&x&z&s\end{pmatrix}^T\in P\cap\Z^{N}$ if and only if $c$ is a valid configuration. The rest of the proof is then easy.

    To this end, suppose that we have a point $y=\begin{pmatrix}c&x&z&s\end{pmatrix}^T\in P\cap\Z^{N}$. First, observe that due to the objective constraints, w.r.t. the load, $c$ is a valid configuration for either of the three objectives. Clearly, $c$ fulfills the capacity constraints. The deadline constraints $$\sum_{k=1}^j(p_{k,i})^Tc^{(k)}\leq d_j$$ assure that for each deadline, the jobs selected by $c$ that have this (or an earlier) deadline can be finished before the deadline, in each dimension. This constraint implicitly uses the fact that if a set of jobs can be scheduled on a machine without violating any due date, then they can be scheduled in EDD-order (i.e., in non-decreasing order of due dates). The class constraints are slightly more tricky: The first constraints set the value $x_i$ to the number of jobs selected by $c$ that belong to class $i$. The second, third and fourth groups of class constraints assure that whenever $x_i>0$, then $z_i=1$ and otherwise $z_i=0$: The fourth constraint just means that $z_i$ is either $0$ or $1$. If $x_i=0$, the second constraint forces $z_i=0$ and the third constraint is irrelevant. If $x_i>0$, the second constraint is irrelevant and the third one means that $z_i$ cannot be zero (and must hence be $1$). The fifth class constraint now assures that the number of classes where $x_i>0$ is at most $C^{\textup{class}}$. So $c$ fulfills the class constraints and hence is a valid configuration of the given problem.

    If we -- on the other hand -- are given a valid configuration, we can simply create a vector $y=\begin{pmatrix}c&x&z&s\end{pmatrix}^T$ by
    \begin{itemize}
        \item setting the entries in $c$ according to the number of jobs of each type appearing in the configuration,
        \item setting $x_i$ to the number of jobs of class $i$ that are in the configuration for all $i\in[K_{\#}]$,
        \item setting $z_i=1$ if $x_i>0$ and $z_i=0$ otherwise for all $i\in[K_{\#}]$ and
        \item setting the $s$-entries to the slack of the corresponding inequalities.
    \end{itemize}
    By the above arguments, $y\in P\cap\Z^{N}$.

    Now assume that there is a point $y=\sum_{i=1}^m y^{(i)}\in \textup{int.cone}(P\cap\Z^{N})\cap Q$ that is the sum of $m$ vectors $y^{(i)}\in P\cap\Z^{N}$. By the above observations, each $y^{(i)}$ corresponds to a set of jobs that may be scheduled together on a single machine. The number of configurations matches the number of machines. Moreover, the jobs in the configurations sum up to all the jobs in the scheduling instance (since the $c$-vectors sum up to the job vector $n$). So the vector $y$ gives rise to a feasible schedule.

    In the other direction, a schedule gives us $m$ configurations, each of which can be modelled as a vector $y^{(i)}\in P\cap\Z^N$. The job part sums up to $n$. So the sum of these vectors is a point in $\textup{int.cone}(P\cap\Z^{N})\cap Q$.

    It remains to show that $P$ is bounded. Due to the objective constraints, $c$ is bounded. Moreover, $x$ and $z$ are bounded because of the class constraints. If we add slack variables $s$, these are also bounded.
\end{proof}

\begin{theorem}\label{thm:complex}
    The problems \tf{\P}{\dj,\vec,\class,\capacity}{\{\cmax,\cmin,\envy\}} can be solved in time
    \[(\log(\max\{\n,\pmax\}))^{2^{\Oh{d+Md_{\#}}}}\enc{I}^{\Oh{1}},\]
    where the $\n$-term only appears due to the class constraints.
\end{theorem}
\begin{proof}
    Consider the constraints describing $P$. The parameters in our \pqmr in \autoref{lem:pqcomplex} are: $K=\max\{M^{(P)},M^{(Q)},N\}=\Oh{d+K_{\#}+Md_{\#}}=\Oh{d+Md_{\#}}$ (the number of job types $d$ is at least as large as the number of classes $K_{\#}$), $\norm{A^{(P)}}_{\infty}=\max\{\n,\pmax\}$ and finally $\enc{P},\enc{Q},\log(m)=\enc{I}^{\Oh{1}}$. Note that $\n$ only appears in the class constraints.
    
    \autoref{cor:pqviavertices} yields running time:
    \begin{align*}                          &\left(\log\left(\norm{A^{(P)}}_\infty\right)\right)^{2^{\Oh{K}}}\enc{P}^{\Oh{1}}\enc{Q}^{\Oh{1}}\log(\m)^{\Oh{1}}\\
    =&\left(\log\left(\max\{\n,\pmax\}\right)\right)^{2^{\Oh{d+Md_{\#}}}}\enc{I}^{\Oh{1}}
    \end{align*}
\end{proof}
As noted above, the problem \tf{\P}{\dj,\vec,\class,\capacity}{\{\cmax,\cmin,\envy\}} is quite general, so it captures many special cases, for some of which we now give corollaries of \autoref{thm:complex}:

\begin{corollary}\label{cor:class}
    \tf{\P}{\class}{\{\cmax,\cmin,\envy\}} can be solved in time
    \[(\log(\max\{\n,\pmax\}))^{2^{\Oh{d}}}\enc{I}^{\Oh{1}}.\]
\end{corollary}

\begin{corollary}\label{cor:capacity}
    \tf{\P}{\capacity}{\{\cmax,\cmin,\envy\}} can be solved in time
    \[(\log(\pmax))^{2^{\Oh{d}}}\enc{I}^{\Oh{1}}.\]
\end{corollary}

\begin{corollary}\label{cor:vector}
    \tf{\P}{\vec}{\{\cmax,\cmin,\envy\}} can be solved in time
    \[(\log(\pmax))^{2^{\Oh{d+M}}}\enc{I}^{\Oh{1}}.\]
\end{corollary}

\subsection{Setup Times}
Note that in the problem \tf{\P}{\setup}{\cmax,\cmin,\envy\}}, the order or starting time of jobs might matter in general, as one can increase the load of a machine by alternating between jobs of different classes to meet some $\cmin$ or $\envy$ requirement.\footnote{Of course, for the objective $\cmax$, it would not make sense to pay extra setups in an optimal solution.} In the following, we assume that no such unnecessary setups are paid (i.e., jobs of the same class are scheduled in a batch). 

So we assume that the setup $s_j$ of a class is only paid once per machine. In the context of \tf{\P}{\setup}{\cmax,\cmin,\envy\}}, the number $d$ of job types might be as large as the number of different processing times multiplied by the number of classes $s_{\#}$.

\begin{lemma}\label{cor:pqsetup}
    \tf{\P}{\setup}{\{\cmax,\cmin,\envy\}} has a \pqmr, assuming that each class induces only one setup per machine.
\end{lemma}
\begin{proof}
    We do not go into too much detail here, as the proof is quite similar to the one of \autoref{lem:pqcomplex}. The vectors in $P$ have a configuration part $c$ and parts $x$ and $z$ for the setups. These parts fulfill a similar role as in the class constraints: For a class $i\in[s_{\#}]$, the variable $x_i$ is set to the number of jobs from class $i$ that are in the configuration $c$. Then if $x_i>0$, $z_i$ is set to $1$ and otherwise $z_i=0$. This works with the same constraints as in \autoref{lem:pqcomplex}. But instead of upper bounding the sum of all $z_i$, we change the objective constraints to:
    \begin{align*}
        \left\{\ell_i, 0\right\}\leq p_i^Tc+\sum_{j=1}^{s_{\#}}s_j z_j\leq \left\{u_i, \sum_{j=1}^dp_j\right\} \quad \forall i\in[M]
    \end{align*}
    This way, the setups induced by a class being represented on a machine are included in the load bounds.
\end{proof}

\begin{theorem}\label{cor:setup}
    \tf{\P}{\setup}{\{\cmax,\cmin,\envy\}} can be solved in time
    \[\log(\max\{\n,\pmax,\smax\}))^{2^{\Oh{d}}}\enc{I}^{\Oh{1}},\]
    assuming that each class induces only one setup per machine.
\end{theorem}
\begin{proof}
    Again, the proof is similar to the one of \autoref{thm:complex}. Just note that the setups $s_i$ appear in the matrix $A^{(P)}$. There are $\Oh{d+s_{\#}}=\Oh{d}$ variables, so calculations similar to the ones in the proof of \autoref{thm:complex} yield the claimed running time.
\end{proof}

\subsection{Weighted Number of Late Jobs}
We now consider the objective $\sumwu$. In this context, we are only concerned with the early jobs (those that make their due date) and we can assume that these early jobs are scheduled in EDD order (i.e., in order of non-decreasing due dates) on each machine~\cite{LS20}. The late jobs may all be scheduled afterwards on an arbitrary machine. Note that in the problem \tf{\one}{\vec}{\sumwu}, a job $j$ has an $M$-dimensional processing time $p_j$, a (one-dimensional) due date $d_j$ and a weight $w_j$. So our parameter $d$ (the number of different job types) naturally extends to the number of different $(p_j,d_j,w_j)$-tuples. So if we denote the number of different processing times, due dates and weights by $p_{\#}$, $d_{\#}$ and $w_{\#}$, respectively, $d$ may be as large as $p_{\#}d_{\#}w_{\#}\leq \pmax^M\dmax\wmax$. We are given a threshold $T$ for the total weight of the late jobs.
\begin{lemma}\label{lem:pqpsumwu}
    The problem \tf{\P}{\vec}{\sumwu} has a \pqmr.
\end{lemma}
\begin{proof}
    Set $m$ to the number of machines,
    \[P=\left\{\left.\begin{pmatrix}c\\x\\s\end{pmatrix}\in\R_{\geq0}^{N}\,\right|\,
    A
    \begin{pmatrix}
        c\\x\\s
    \end{pmatrix}
    = b\right\}\]
    and
    \[Q=\left\{\left.\begin{pmatrix}c\\x\\s\end{pmatrix}\in\R_{\geq0}^N\,\right|\,x\geq\sum_{i=1}^dw_i-T,\,\mathbf{0}\leq c\leq n\right\},\]
    where $N:=d+1+d_{\#}M$ and $A$ and $b$ are given by the following constraints:
    \begin{align*}
        \sum_{k=1}^j(p_{k,i})^Tc^{(k)}&\leq d_j && \forall j\in[d_{\#}],i\in[M] &&&\textup{(due dates)} \\
        w^T c&= x &&  &&&\textup{(objective)}
    \end{align*}
    Here, $p_{k}$ is a matrix containing the vectors of processing times that correspond to the jobs of the $k$-th due date. Then $p_{k,i}$ contains the processing times in dimension $i$ corresponding to the jobs of due date $k$. Similarly, $c^{(k)}$ includes the job types that have due date $d_k$.

    As we showed similarly in the proof of \autoref{lem:pqcomplex}, the $c$-part of the integral vectors in $P$ are valid configurations. In this case, \enquote{valid} means that the jobs in $c$ can be scheduled on the same machine (in EDD order) without missing any due dates (in any of the $M$ dimensions). The objective constraint sets $x$ to the total weight of all the jobs selected by $c$. The due date constraints make sure that no job misses its due date in any dimension if they are scheduled in EDD order. Given a threshold $T$ for the desired weight, a solution of the scheduling problem boils down to a subset of jobs that can be scheduled early and their assignment to the machines such that their total weight is at least $\sum_{i=1}^dw_i-T$ (and hence the weight of all the late jobs is at most $T$). This is exactly what a point in $\textup{int.cone}(P\cap\Z^{N})\cap Q$ corresponds to: A set of configurations (containing the early jobs, hence the inequality $c\leq n$ instead of $c=n$ in the description of $Q$) that have total weight at least $\sum_{i=1}^dw_i-T$.
\end{proof}
We can now get a parameterized algorithm for \tf{\P}{}{\sumwu}:

\begin{theorem}\label{thm:vectorsumwu}
    The problem \tf{\P}{\vec}{\sumwu} can be solved in time
    \[(\log(\max\{\pmax,\wmax\}))^{2^{\Oh{d+d_{\#}M}}}\enc{I}^{\Oh{1}}.\]
\end{theorem}
\begin{proof}
    The parameters in our \pqmr in \autoref{lem:pqpsumwu} are: 
    \begin{itemize}
        \item $\enc{P},\enc{Q},\log(\m)=\enc{I}^{\Oh{1}}$,
        \item $\norm{A^{(P)}}_{\infty}=\max\{\pmax,\wmax\}$ and
        \item $K=\max\{M^{(P)},M^{(Q)},N\}=\Oh{d+d_{\#}M}$.
    \end{itemize}
    Hence, \autoref{cor:pqviavertices} yields running time:
    \begin{align*}                          &\left(\log\left(\norm{A^{(P)}}_\infty\right)\right)^{2^{\Oh{K}}}\enc{P}^{\Oh{1}}\enc{Q}^{\Oh{1}}\log(\m)^{\Oh{1}}\\
    =&(\log(\max\{\pmax,\wmax\}))^{2^{\Oh{d+d_{\#}M}}}\enc{I}^{\Oh{1}}
    \end{align*}
\end{proof}

\begin{corollary}\label{thm:sumwu}
    The problem \tf{\P}{}{\sumwu} can be solved in time
    \[(\log(\max\{\pmax,\wmax\}))^{2^{\Oh{d}}}\enc{I}^{\Oh{1}}.\]
\end{corollary}
Note that $d_{\#}\leq d$.

\subsection{Uniform N-Fold ILPs}
\autoref{cor:pqviavertices} can also be used to solve uniform $n$-fold ILPs.
In general, these ILPs have the following form:
\begin{align*}
    \max&\,\sum_{i=1}^n(c^{(i)})^Tx^{(i)}\\
    \matrix{C^{(1)} & \hdots & C^{(n)} \\ B^{(1)} & & \\ & \ddots & \\ & & B^{(n)}} \matrix{x^{(1)}\\\vdots\\x^{(n)}} &= \matrix{b^{(0)} \\ b^{(1)} \\\vdots\\b^{(n)}} \\
    \mathbf{0}\leq x^{(i)}&\leq u^{(i)} &&\forall i\in[n]\\
    x^{(i)}&\in\Z^t &&\forall i\in[n]
\end{align*}
where $C^{(i)}\in\Z^{r\times t}$, $B^{(i)}\in\Z^{s\times t}$, $b^{(0)}\in\Z^r$, $b^{(i)}\in\Z^s$, $c\in\Z^t$ and $u^{(i)}\in(\Z_{>0}\cup\{\infty\})^t$. 
Here, we consider the special case with $C^{(1)}=\hdots=C^{(n)}$, $B^{(1)}=\hdots=B^{(n)}$, $b^{(1)}=\hdots=b^{(n)}$, $c^{(1)}=\hdots=c^{(n)}$ and $u^{(1)}=\hdots=u^{(n)}$.
By standard bounds for the solutions of (I)LPs, we know that the objective value is bounded by some $U$ that has an encoding length bounded by $\enc{I}^{\Oh{1}}$ (see Schrijver~\cite{Schrijver86} for bounds on the $\ell_{\infty}$-norm of (I)LP solutions; one could even use proximity arguments like in \cite{JLR20}), if the bounds $u$ were infinite in some directions. So by doing a binary search over the objective values, we can assume that we are looking for a solution $x$ with $\sum_{i=1}^ncx^{(i)}\geq T$ for some $T\in[U]$. If the ILP is bounded, the objective value is also bounded and doing a binary search over the possible objective values only adds a factor of $\enc{I}^{\Oh{1}}$.
\begin{lemma}\label{lem:pqnfold}
    Uniform $n$-fold ILPs have a \pqmr with $m=n$,
    \[P=\left\{\left.\matrix{x\\z\\g\\v}\in\R_{\geq0}^{2t+r+1}\,\right|\,\matrix{B&0&0&0\\I&I&0&0\\C&0&-I&0\\c^T&0&0&-1}\matrix{x\\z\\g\\v}=\matrix{b^{(1)}\\u\\\mathbf{0}\\0}\right\}\]
    and
    \[Q=\left\{\left.\matrix{x\\z\\g\\v}\in\R_{\geq0}^{2t+r+1}\,\right|\,v\geq T,\,g=b^{(0)}\right\}.\]
\end{lemma}
\begin{proof}
    Suppose the given uniform $n$-fold ILP has a solution $x^*\in\Z^{nt}$ with objective value $\sum_{i=1}^nc^T(x^*)^{(i)}\geq T$. Then each of the $n$ bricks $(x^*)^{(i)}\in\Z^{t}$ can be augmented to a vector $\matrix{x&z&g&v}^T\in P$ by setting the value $v=c^T(x^*)^{(i)}$, the solution $x=(x^*)^{(i)}$, the slack $z=u-(x^*)^{(i)}$ and the global part $g=C(x^*)^{(i)}$. The sum of all these augmented vectors is in $Q$: The values for $x$ and $z$ are unconstrained; we have $\sum_{i=1}^nC(x^*)^{(i)}=b^{(0)}$, which is exactly the sum of all the $g$-entries; the value is $\sum_{i=1}^nc^T(x^*)^{(i)}$, which is at least $T$ by assumption. So $\textup{int.cone}(P\cap\Z^{2t+r+1})\cap Q$ is not empty.

    Now suppose $\textup{int.cone}(P\cap\Z^{2t+r+1})\cap Q$ is not empty but contains a vector $y=y^{(1)}+\hdots+y^{(n)}$ such that $$y^{(i)}=\matrix{x^{(i)}&z^{(i)}&g^{(i)}&v^{(i)}}^T\in P\cap\Z^{2t+r+1}$$ for all $i\in[n]$. Set $x^*$ to be the concatenation of the $x^{(i)}$-parts, i.e., $x^*=\matrix{x^{(1)},\hdots,x^{(n)}}^T$. We claim that $x^*$ is a solution of the uniform $n$-fold ILP that has objective value $\sum_{i=1}^nc^T(x^*)^{(i)}\geq T$. Due to the last constraint in $P$ that sets the last coordinate of $y^{(i)}$ to $c^Ty^{(i)}$ and the constraint in $Q$, we have $\sum_{i=1}^nc^T(x^*)^{(i)}=\sum_{i=1}^nc^Ty^{(i)}\geq T$. The third group of constraints in $P$ sets the $g$-part of the vectors to $Cy^{(i)}$ and since $y\in Q$, we have $\sum_{i=1}^nCy^{(i)}=b^{(0)}$, so the global constraints $\matrix{C&\hdots&C}x^*=b^{(0)}$ are fulfilled. The local constraints $B(x^*)^{(i)}=b^{(1)}$ are also fulfilled because of the first group of constraints in $P$ and $y^{(i)}\in P$ for all $i\in[n]$. Similarly, we have $0\leq (x^*)^{(i)}\leq u$ for each $i\in[n]$, thanks to the second group of constraints in $P$. Clearly, $x^*$ is integer. So $x^*$ is a solution of the given uniform $n$-fold ILP with good enough objective value.
\end{proof}
We can now use this to obtain an algorithm for uniform $n$-fold ILPs:

\begin{theorem}\label{thm:uniformnfoldilps}
    Uniform $n$-fold ILPs can be solved in time
    \[(\log(\max\{\norm{C}_{\infty},\norm{B}_{\infty},\norm{c}_{\infty}\}))^{2^{\Oh{r+s+t}}}\enc{I}^{\Oh{1}}.\]
\end{theorem}
\begin{proof}
    The parameters in our \pqmr in \autoref{lem:pqnfold} are: 
    \begin{itemize}
        \item $\enc{P},\enc{Q},\log(\m)=\enc{I}^{\Oh{1}}$,
        \item $\norm{A^{(P)}}_{\infty}=\max\{\norm{C}_{\infty},\norm{B}_{\infty},\norm{c}_{\infty}\}$ and
        \item $K=\max\{M^{(P)},M^{(Q)},N\}=\Oh{r+s+t}$.
    \end{itemize}
    Hence, \autoref{cor:pqviavertices} yields running time:
    \begin{align*}                          &\left(\log\left(\norm{A^{(P)}}_\infty\right)\right)^{2^{\Oh{K}}}\enc{P}^{\Oh{1}}\enc{Q}^{\Oh{1}}\log(\m)^{\Oh{1}}\\
    =&(\log(\max\{\norm{C}_{\infty},\norm{B}_{\infty},\norm{c}_{\infty}\}))^{2^{\Oh{r+s+t}}}\enc{I}^{\Oh{1}}
    \end{align*}
\end{proof}

\subsection{\texorpdfstring{\mswbp}{MinSumWeightedBinPacking}}
In the \mswbp problem, similar to scheduling, we are given an item vector $n\in\N^d$, a corresponding size vector $p\in\N^d_{>0}$, a weight vector $w\in\N^d_{>0}$ and a bin capacity $B\in\N_{>0}$. The goal is to pack all items into an unlimited number of bins such that the cost of all used bins is minimized; a bin with index $i$ costs $i$ times the total weight of the items in the bin. This problem has, e.g., been studied in~\cite{EL08}.


Note that it makes no sense to pack the $\n$ items into more than $\n$ bins, so we can via binary search (invoking an additional cost of $\log(\n)$) assume that we know how many bins are used. Like in the other problems, we are looking at the decision variant, so we are also given a threshold for the objective value. Let $T$ be this objective value and let $m$ be the number of bins used by some feasible packing. 
\begin{lemma}\label{lem:pqmswbp}
    \mswbp has a \pqmr.
\end{lemma}
\begin{proof}
    Set $m$ to the number of used bins (which we know via binary search) and let
    \[P=\left\{\left.\begin{pmatrix}c\\x\\t\\v\\z\end{pmatrix}\in\R_{\geq0}^{N}\,\right|\,A\matrix{c\\x\\t\\v\\z}\leq b\right\}\]
    and
    \[Q=\left\{\left.\begin{pmatrix}c&x&t&v&z\end{pmatrix}^T\in\R_{\geq0}^{N}\,\right|\,c=n,\,x=\mathbf{1},\,0\leq z\leq T\right\},\]
    where $N:=d+2m+2$ and $A$ and $b$ are defined by the following constraints:
    \begin{align*}
        p^Tc &\leq B \\
        \sum_{i=1}^m x_i &= 1 \\
        \sum_{i=1}^d w_ic_i &= v \\
        t_i+\wmax\min\{B,\n\}&\geq v+x_i\wmax\min\{B,\n\}  &&\forall i\in[m] \\ 
        t_i&\leq \wmax\min\{B,\n\}  &&\forall i\in[m] \\
        \sum_{i=1}^m i t_i &= z
    \end{align*}
    The integer vectors in $P$ fulfill the load constraint, i.e., the items (given by the vector $c$) fit into the bin. The second constraint assures that if $x_i$ is non-negative and integer, exactly one component is equal to $1$ and the others are zero. This selects the bin in which the items are to be packed. The third constraint just sets the variable $v$ to the weight of the items in the bin/configuration. The fourth constraint assures that if $x_i=1$, then $t_i\geq v$; otherwise, the constraint is always fulfilled, as $v\leq \wmax\min\{B,\n\}$. The sixth constraint stores the actual value of the configuration (depending on the bin index and the weight of the items in the bin), while the fifth constraint just assures the boundedness of $P$ (note that the other variables are also all bounded).

    A solution of a given \mswbp instance using $m$ bins can be modelled by creating $m$ vectors $y^{(1)},\hdots,y^{(m)}\in P\cap\Z^{N}$, each of which represents a packing of a singular bin of the form $\matrix{c&x&t&v&z}^T$. Vector $y^{(k)}$ is defined as follows: 
    \begin{itemize}
        \item The $c$-part represents the packed items,
        \item $x_k=1$ and $x_i$ for $i\neq k$ is set to $0$,
        \item $t_k=v$ is set to the total weight of the items in the configuration, i.e., $\sum_{i=1}^d w_ic_i$ and the other $t_i$ are set to $0$ and
        \item $z$ is set to $k v$, the weight of the configuration multiplied by its bin index.
    \end{itemize}
    Then each item is contained in one of the configurations, so the $c$-parts sum up to $(n_1,\hdots,n_d)^T$; there is exactly one $y^{(i)}$-vector for each bin, meaning the $x$-vectors sum up to $\mathbf{1}$; and finally, the objective values of all the configurations sum up to some value $\leq T$. So $y:=y^{(1)}+\hdots+y^{(m)}\in \textup{int.cone}(P\cap\Z^{N})\cap Q$.

    For the other direction, consider a vector $y\in \textup{int.cone}(P\cap\Z^{N})\cap Q$ with $y:=y^{(1)}+\hdots+y^{(m)}$ for some vectors $y^{(i)}\in P\cap\Z^{N}$. By the above observations, each of these $y^{(i)}$ corresponds to a valid packing of a bin such that the $z$-entry contains an upper bound for the value this bin contributes to the overall objective. Since the sum of the $c$-vectors is equal to $(n_1,\hdots,n_d)^T$, all items are packed and since the $z$-values sum up to a value $\leq T$, the packing has a good enough objective value.
\end{proof}

\begin{theorem}\label{thm:mswbp}
    \mswbp can be solved in time
    \[(\log(\max\{\pmax,\wmax\min\{B,\n\}\}))^{2^{\Oh{d+m}}}\enc{I}^{\Oh{1}}.\]
\end{theorem}
\begin{proof}
    The parameters in our \pqmr in \autoref{lem:pqmswbp} are:
    \begin{itemize}
        \item $\enc{P},\enc{Q},\log(\m)=\enc{I}^{\Oh{1}}$,
        \item $\norm{A^{(P)}}_{\infty}= \max\{\pmax,\wmax\min\{B,\n\}\}$ and 
        \item $K=\max\{M^{(P)},M^{(Q)},N\}=\Oh{d+m}$.
    \end{itemize}
    Hence, \autoref{cor:pqviavertices} yields running time:
    \begin{align*}                          &\left(\log\left(\norm{A^{(P)}}_\infty\right)\right)^{2^{\Oh{K}}}\enc{P}^{\Oh{1}}\enc{Q}^{\Oh{1}}\log(\m)^{\Oh{1}}\\
    =&(\log(\max\{\pmax,\wmax\min\{B,\n\}\}))^{2^{\Oh{d+m}}}\enc{I}^{\Oh{1}}
    \end{align*}
\end{proof}

\section{Additive Approximation Scheme}\label{sect:additiveapprox}
Approximation algorithms have an error measure $\eps$ as part of their input and they usually produce a solution with value at most $(1+\eps)\opt$, where $\opt$ is the optimum (for a minimization problem). An additive approximation algorithm on the other hand produces a solution with value at most $\opt + \eps k$, where $k$ is some parameter. A classical approximation algorithm can be seen as an additive approximation algorithm with additive error $\eps\opt$. 
So if we had two approximation algorithms with the same running time but one with additive error and one with multiplicative error, the one with additive error would be better, as long as $k<\opt$. 
Here, we consider parameter $k=\pmax$ for the problems \tf{\P}{}{\{\cmax,\cmin,\envy\}}. 

In the additive approximation schemes presented in \cite{buchem}, Buchem~\etal address load balancing problems on identical machines by introducing slot-MILPs, where jobs are organized in job size classes and fractionally assigned to machines such that the fractional components sum up to integral numbers of slots within each class. The idea is to first determine the number of slots to reserve for each class on each machine and to find an integral assignment to these slots in a second step. On each machine, the load is constrained to the target load interval $\closedinterval{\ell,u}$ for $\ell \in \Nwithzero$ and $u \in \Nwithoutzero$. The problem of finding an assignment for $\nlm$ jobs to $\mlm$ machines under the given load constraints is referred to as the \textit{target load balancing problem}. The decision versions of \tf{\P}{}{\{\cmax,\cmin,\envy\}} are special cases of this problem. An approximation scheme for each of these is included in~\cite{buchem} and also briefly described in \cref{subsec:application}. In this section, we present an adaptation of the algorithm to a high-multiplicity setting. We slightly improve the running time from $\originaladditiveapproxrt$ to $\additiveapproxrt$ and present a way to speed up the entire algorithm even more in case better running times can be achieved for the dynamic program making up the first part. First, we provide some more details on the original algorithm.

For $k \in \closedinterval{ \epsfrac}, \epsfrac \in \Nwithoutzero$, a job size class $\jk$ refers to the set of jobs $\{ j \in \J \mid p_j \in ((k-1)\eps \pmax, k \eps \pmax]\}$. Formally, the $\slotmilp$ is defined as follows. 

\begin{align*}
    \sum\limits_{i \in \M} x_{i,j} &= 1 & \forall j \in \J\\
    \ell \leq \sum\limits_{j \in \J} p_j x_{i,j} &\leq u & \forall i \in \M\\
    \sum\limits_{i \in \jk} x_{i,j} &= y_{i,k} & \forall i \in \M, k \in \closedinterval{\epsfrac}\\
    x_{i,j}&\geq 0 & \forall j \in \J, i \in \M\\
    y_{i,k}&\in \Nwithzero& \forall i \in \M, k \in \closedinterval{\epsfrac}
\end{align*}

\sloppy A weaker version, referred to as the $\slotmilpdelta$, admits an error to the upper and lower bound respectively, using the redefined bounds $u'  = u + \delta \pmax$ and $\ell' = \ell - \delta \pmax$ where $1 > \delta > 0$. While Buchem~\etal~\cite{buchem} dedicate most of the paper to the $\delta$-approximation, they also prove a structural property of the $\slotmilp$ allowing to compute an exact solution in $\mlm^{\O({2^{\epsfrac}})} \nlm^{\O({ \epsfrac 2^{\epsfrac}})}$. As our approach is based on the $\delta$-approximation, we will almost exclusively refer to the $\slotmilpdelta$ here.  

Roughly outlined, the algorithm proceeds as follows to obtain a solution to the target load balancing problem:
\begin{enumerate}
    \item For each $i \in \M$ compute 
    \item[] \begin{itemize}
        \item a vector $y \in {\closedinterval{\nlm}}^{\epsfrac}$ specifying how many jobs of each job size class are assigned to machine $m_i$ and
        \item a vector $z \in \R^{\epsfrac}$ specifying the average processing time of the jobs assigned to $m_i$ within each job size class.
    \end{itemize}
    \item[] This step is achieved in $\originaladditiveapproxrt$ time using a dynamic program which relies on some structural properties of an existing solution. If the $\slotmilpdelta$ is feasible, the DP is guaranteed to find a solution, otherwise it states that the $\slotmilpdelta$ is infeasible. As the vector $z$ is only needed for the computation steps of the DP, it is discarded afterwards. 
    \item From the obtained vector $y$, assign jobs integrally to the slots using a local search algorithm described in \cref{subsec:localsearch}. This step has a running time of $\O(\nlm^5)$ and introduces an additional absolute error of at most $\eps \pmax$.
\end{enumerate}
The overall running time is determined by the DP and amounts to $\originaladditiveapproxrt$. In some cases, we can substantially decrease the number of jobs to be considered. Key to our approach is the following result by Govzmann~\etal~\cite{GMO23}.

\balancinglemma*

As described in more detail in~\cref{subsec:balancinglemma}, this result follows from a couple of insights making it possible to preassign $\left \lceil \frac{n_p}{\mlm}\right \rceil - \pmax$ jobs for each $p \in \D$ to each machine and obtain a kernel consisting only of $\O(\mlm \jobtypes \pmax)$ jobs. Creating one single assignment vector to be shared among all machines is achieved in $\O(\jobtypes)$. Replacing $\nlm$ by $\mlm \jobtypes \pmax$ immediately yields the following: 

\begin{corollary}\label{cor:slotmilp}
        There exists an algorithm that, given an instance with $\jobtypes$ different processing times, computes a solution to the $\slotmilpdelta$ in time $\O(\dprt)$ or asserts that the $\slotmilpdelta$ is infeasible.
\end{corollary}

\begin{remark}
    The running time specified in~\cite{buchem} is $\mlm^2  \left(\frac{\nlm}{\delta \eps}\right)^{\O\left(\epsfrac\right)}$. The full version of the paper \cite{fullversion} includes $K$ different target load intervals $[\ell_i,u_i]$ for the machines. In other words, $K$ is defined by $|\{[\ell_i,u_i]\mid i \in \M\}|$. For each machine, the DP includes one additional step which consists of guessing the number of machines for each target load interval. This results in a running time of $\mlm^K  \mlm  \left(\frac{\nlm}{\delta \eps}\right)^{\O\left(\epsfrac\right)} = \mlm^{K + 1}  \left(\frac{\nlm}{\delta \eps}\right)^{\O\left(\epsfrac\right)}$. It appears that $K$ was replaced by 1 in the simplified version with only one target load interval, leading to the squared number of machines. However, this would imply guessing the number of machines for the remaining target load interval. Since the step of guessing the number of machines can be entirely omitted, it seems that the running time should be $\mlm  \left(\frac{\nlm}{\delta \eps}\right)^{\O\left(\epsfrac\right)}$ instead.
\end{remark}

As the running time of the DP hides the running time of the local search algorithm used to find an integral assignment to the slots, using \autoref{cor:slotmilp} directly yields (note that $d\leq\pmax$):
\thmadditiveapprox*

 However, it might be possible to adapt the DP even further such that the number of machines disappears from the exponential term. In the hope that further research will reveal ways of speeding up the DP enough for the running time of the local search to become relevant, we present an enhanced local search algorithm. 

\subsection{Local search algorithm}\label{subsec:localsearch}

Given an exact solution or a $\delta$-approximation to the $\slotmilp$, Buchem~\etal~\cite{buchem} use a local search algorithm to compute an integral solution with error at most $\eps  \pmax$ to the upper and lower bounds respectively. More precisely, from an exact solution to the $\slotmilp$, the algorithm computes an integral schedule such that $\sum\limits_{j \in \J} p_j x_{i,j} \in [\ell - \eps  \pmax, u + \eps  \pmax]$ holds for each machine $i \in \M$. 

Starting from a $\delta$-approximation, the local search adds an error of $\eps  \pmax$ to the error already existing due to the DP. Each machine has therefore a load in the interval $[\ell - \delta  \pmax- \eps  \pmax, u + \delta  \pmax + \eps  \pmax]$ on termination of the local search. The error contributed by the $\delta$-approximation can be compensated by setting $\delta = \frac{\eps^*}{2}, \eps = \frac{\eps^*}{2}$ for some final error parameter $\eps^* > 0$.

Based on the vector $y$, the algorithm arbitrarily assigns $y_{i,k}$ jobs of class $\jk$ to machine $i$ for each $i \in \M, k \in \closedinterval{\epsfrac}$. This does not guarantee a load within the interval $[\ell - \eps  \pmax, u + \eps  \pmax]$ on any machine though. Therefore, the algorithm proceeds by exchanging jobs on overloaded machines for smaller jobs of the same class on machines with load at most $u$. Machines with a load between $u+1$ and $u + \eps  \pmax$ cannot provide any smaller jobs for an exchange because they might become overloaded themselves. When such machines are encountered during the search, they are repaired in a recursive procedure. Dealing with the overloaded machines is referred to as the first stage. In a second stage, an analogous procedure is performed for the machines with a load smaller than $\ell - \eps  \pmax$. Since both stages follow the same principles, we will only discuss the first stage.

\input{localsearch}

Buchem \etal~\cite{buchem} provide a representation of the local search algorithm as a weighted directed graph $G = (V,E)$. For each slot on each machine, a vertex is created. A machine is described by a clique where each pair of vertices is connected by a bidirectional edge of weight 0. For each pair of slots $v,w \in V$, the graph contains an edge from $v$ to $w$ of weight 1 if and only if both slots are in the same job size class and on different machines and $w$ contains a job smaller than the one assigned to $v$, that is, if slot $w$ contains a possible swap candidate for $v$. Furthermore, $V$ contains an additional source vertex $s$ with an edge of weight 0 to all vertices on all machines with a load exceeding $u + \eps \pmax$. For an illustration, see \cref{subfig:localsearch1} and \cref{subfig:localsearch2}. The local search algorithm is equivalent to a breadth-first search (BFS) starting at $s$ where the vertices representing one machine are collected all at once as soon as one vertex of the clique is reached. A swap of jobs between two slots is expressed by redirecting all incoming and outgoing edges accordingly. Whenever two jobs are exchanged, the search restarts at vertex $s$. Denote by $\delta(v)$ the number of edges of weight 1 on a path from $s$ to vertex $v$ and let $j_1, \dots, j_{\n}$ be the jobs in increasing order of job size. In \cite{buchem}, it is shown that the value of the following potential function increases with each swap:

$$\sum\limits_{i = 1}^{\nlm} i  \delta(j_i)$$

From this function, an upper bound of $\nlm^2  \mlm$ can be derived ($\nlm^3$ in the original paper). A BFS can be carried out in $\O(\nlm^2)$. The preprocessing reduces the number of jobs considered in the local search algorithm. This improves the running time $\O(\nlm^4  \mlm)$ to $\O( (\mlm  \jobtypes  \pmax)^4  \mlm )$. However, we can do even better by exploiting some structural properties of the graph given in this specific setting.

\begin{lemma}\label{lem:bfspathlength}
    $\delta(v) < \jobtypes$ for each $v \in V$.
\end{lemma}

\begin{proof}
    For the proof, we have to distinguish between vertices adjacent to $s$ and the remaining ones. If a vertex $v \in V$ corresponds to a slot on a machine with a load exceeding $u + \eps \pmax$, then $\delta(v) = 0$ regardless of the processing time of the job stored in the slot. Let $v_{1}, v_{2} \in V$ and $ (s,v_{1}), (s,v_{2}) \notin E$ and denote by $p(v')$ the processing time of the job assigned to the slot represented by a vertex $v' \in V$. Let $w(e)$ be the weight of an edge $e \in E$. Then from $p(v_{1}) = p(v_{2})$ follows the equality of the sets $\{ x \in V \mid (x,v_{1}) \in E \wedge w((x,v_{1})) = 1 \}$ and $\{ x \in V \mid (x,v_{2}) \in E \wedge w((x,v_{2})) = 1 \}$, implying that vertices with distance $>0$ storing identical processing times belong to the same BFS layer. At the same time, all vertices inside a machine clique also belong to the same BFS layer, implying that a path does not contain any detours inside cliques. Therefore, each processing time can appear at most once along a path. From this, we can conclude that $\delta(v)$ is at most $\jobtypes-1$ for each $v \in V$.
\end{proof}

Based on this insight, we can make the following adjustment to the potential function. 

\begin{equation}
\label{eq:potentialfunction}
\sum\limits_{i = 1}^{\jobtypes} i \sum\limits_{j = 1}^{n_i}  \delta(j_i)    
\end{equation}

where $1, \dots, \jobtypes$ represents the sequence of indices of job sizes arranged in ascending order. 

Because at most $2  \pmax$ jobs of one specific size are assigned to each machine, we have $n_i \leq 2  \mlm  \pmax$ for each $1 \leq i \leq \jobtypes$. The first sum contributes a factor $\jobtypes^2$ and $\delta(j_i)$ is bounded by $\jobtypes$ due to \cref{lem:bfspathlength}. From this, we ac bound the potential function by $\jobtypes^3  2  \mlm  \pmax$. 

\begin{lemma}
\label{lem:numberofswaps}
The value of \labelcref{eq:potentialfunction} increases with each swap.
\end{lemma}

\begin{proof}
Let $j_{\fromvertex},j_{\tovertex} \in \J$ be two jobs to be swapped and let $\fromvertex,\tovertex \in V$ be the corresponding vertices. Let $p_{j_{\fromvertex}},p_{j_{\tovertex}}$ be the processing times of jobs $j_{\fromvertex}$ and $j_{\tovertex}$ and let $p_{j_{\fromvertex}} > p_{j_{\tovertex}}$. As in \cite{buchem}, we will start by showing that a swap does not decrease the distance to any vertex. Note that no edges outgoing from $s$ are added, since a swap is only carried out with machines with load at most $u$ and the difference in processing time between two jobs of one size class is at most $\eps \pmax - 1$.
Exchanging the jobs between $\fromvertex$ and $\tovertex$ causes a reversal of the edge $(\fromvertex,\tovertex)$ and a complete exchange of both the sets of outgoing edges of weight 1 and the sets of incoming edges of weight 1 between $\fromvertex$ and $\tovertex$. 

Let $\delta(u), \delta'(u)$ denote the distance from $s$ to some vertex $u \in V$ before and after the swap, respectively. It holds that $\delta(\tovertex) = \delta(\fromvertex) + 1$, otherwise both $\fromvertex$ and $\tovertex$ would belong to one and the same BFS layer, preventing the swap between $\fromvertex$ and $\tovertex$ in the first place. Another important observation is that the edges within the cliques do not change. A new path to $\fromvertex$ can therefore only be provided by an edge previously incident to $\tovertex$ and vice versa. 
Suppose that the distance to $\fromvertex$ decreases. Then there must have been a path leading to $\tovertex$ such that $\delta(\tovertex) \leq \delta(\fromvertex) - 1$, a contradiction to $\delta(\tovertex) = \delta(\fromvertex) + 1$. Suppose that the distance to $\tovertex$ decreases. Then one of the edges originally incident to $\fromvertex$ yields a new path to $\tovertex$ such that $\delta'(\tovertex) \leq \delta(\tovertex) - 1 = \delta(\fromvertex)$. Let $(\thirdvertex,\fromvertex)$ be this edge before the swap. Because of the transitive relation between all the edges of weight 1, there must have been an edge from $\thirdvertex$ to $\tovertex$ too, implying $\delta(\tovertex) = \delta(\fromvertex)$. This scenario is illustrated in \cref{subfig:swap1} and \cref{subfig:swap2}.

\input{swap}

All paths in the graph which include neither $\fromvertex$ now $\tovertex$ are unaffected by the swap. Those including at least one of the vertices $\fromvertex$ or $\tovertex$ cannot become shorter by the swap since neither $\delta(\fromvertex)$ nor $\delta(\tovertex)$ decreased.

We will now inspect the changes caused to the value of the potential function by the swap. Let $\pi : [\pmax] \longrightarrow [\jobtypes]$ be a function mapping a processing time to its position in an ascending sequence of all processing times in $\D$. As established before, $\delta'(u) \geq \delta(u)$ for all $u \in V$. Therefore, we have the following inequality for the potential without the swapped jobs. 
\begin{align*}
    &\sum\limits_{i = 1}^{\jobtypes} i \sum\limits_{j = 1}^{n_i} \delta(j_i) - \pi(j_{\fromvertex}) \delta(j_{\fromvertex}) - \pi(j_{\tovertex}) \delta(j_{\tovertex}) \\
    \leq &\sum\limits_{i = 1}^{\jobtypes} i \sum\limits_{j = 1}^{n_i} \delta'(j_i) - \pi(j_{\fromvertex})  \delta'(j_{\fromvertex}) - \pi(j_{\tovertex})  \delta'(j_{\tovertex})
\end{align*}

For the exchanged jobs, we can derive the following inequality from $\delta(j_{\tovertex}) = \delta(j_{\fromvertex}) + 1$. 

\begin{align}
&\pi(j_{\fromvertex})  \delta'(j_{\fromvertex}) + \pi(j_{\tovertex}) \delta'(j_{\tovertex})\\
= &\pi(j_{\fromvertex}) \delta(j_{\tovertex}) + \pi(j_{\tovertex}) \delta(j_{\fromvertex})\\
= &(\pi(j_{\fromvertex}) - \pi(j_{\tovertex})) \delta(j_{\tovertex}) + \pi(j_{\tovertex}) \delta(j_{\tovertex}) + \pi(j_{\tovertex}) \delta(j_{\fromvertex})\\
> &(\pi(j_{\fromvertex}) - \pi(j_{\tovertex})) \delta(j_{\fromvertex}) + \pi(j_{\tovertex}) \delta(j_{\tovertex}) + \pi(j_{\tovertex}) \delta(j_{\fromvertex})\\
= &\pi(j_{\fromvertex}) \delta(j_{\fromvertex}) - \pi(j_{\tovertex}) \delta(j_{\fromvertex}) + \pi(j_{\tovertex}) \delta(j_{\tovertex}) + \pi(j_{\tovertex}) \delta(j_{\fromvertex})\\
= &\pi(j_{\fromvertex}) \delta(j_{\fromvertex}) + \pi(j_{\tovertex}) \delta(j_{\tovertex})
\end{align}

Hence, the value of the potential function increases by  at least 1 with each swap, implying that the number of swaps is bounded by $\jobtypes^3  2  \mlm  \pmax$. 
\end{proof}

\begin{lemma}
\label{lem:oneswap}
    One BFS run and the subsequent swap can be completed in $\O(\swaprt)$ time.
\end{lemma}

\begin{proof}
The running time of a BFS is given by $\O(|V| + |E|)$. The naive approach consists of including all edges of the graph in the search. As $|E|$ is only bounded by $|V|^2$, this yields running time $\O((\mlm  \jobtypes   \pmax)^2)$ for one BFS run. However, we can speed up this procedure based on the following thoughts.
\begin{enumerate}
    \item The exact slots or size classes along one path are irrelevant before a machine with load at most $u$ is reached.\label{enum:item1}
    \item Edges of weight 1 exist only between slots of the same class.
    \item Let $\mathcal{P}_{i,k}$ be the partition assembling the vertices representing the $y_{i,k}$ slots of class $\jk$ for some machine $i \in \M$. The edges of weight 1 are based on the strict order $>$. Hence, it is sufficient to consider the outgoing edges of the vertex storing the biggest job of some class $\jk$ assigned to some machine $i \in \M$ in order to know all the destination machines of the outgoing edges from partition $\mathcal{P}_{i,k}$. \label{enum:item2}
    \item For the same reason, the incoming edges of the vertex storing the smallest job of some class $\jk$ assigned to some machine $i \in \M$ reveal all the machines from which partition $\mathcal{P}_{i,k}$ can be reached by at least one edge. \label{enum:item3}
\end{enumerate}

From observation \labelcref{enum:item1} we can conclude that it would be convenient to carry out the search on a graph similar to the one displayed in \cref{subfig:localsearch1}, that is, a graph where slots and size classes are abstracted away. 

In order to construct and maintain such a graph, we first create an auxiliary data structure consisting of $m$ arrays of size $\epsfrac$. Each array represents a machine and each cell represents a size class $\jk$ for $k \in \closedinterval{\epsfrac}$ and contains a min-max heap $T_{i,k}$ storing the jobs assigned to the $y_{i,k}$ slots on the corresponding machine $i \in \M$. Denote by $\min(T_{i',k'})$ $(\max(T_{i',k'}))$ the smallest (resp. biggest) element of some heap $T_{i',k'}$.
Let $i_1,i_2 \in \M, k_1,k_2 \in \closedinterval{\epsfrac}$. Denote by $C_i$ the clique describing machine $i \in \M$ in the original graph $G$. Then $G$ contains an edge leading from some vertex in $C_{i_1}$ to some vertex in $C_{i_2}$ if and only if $\max(T_{i_1,k_1}) > \min(T_{i_2,k_2})$ and $k_1 = k_2$.  Creating the auxiliary data structure requires the insertion of $\O(m  \jobtypes  \pmax)$ jobs into min-max heaps, which can be done in $\O(m  \jobtypes  \pmax  \log(m  \jobtypes  \pmax))$ time.  

We now have all the ingredients needed to construct our simplified graph $G' = (V',E')$ consisting of $m$ vertices, each one representing some machine $i \in \M$. Before each BFS run, we update $G'$ by iterating through the job size classes, comparing machines pairwise within each class. If the largest job on one machine is larger than the smallest job on some other machine, the edge in $G'$ is added accordingly.

The swap procedure as a whole now consists of the following steps:
\begin{itemize}
    \item Update $G'$ in $\O(\epsfrac  m^2)$.
    \item Perform the BFS in $\O(m^2)$.
    \item When some machine $i_1$ suitable for a swap is found, that is, when its load is at most $u$, then perform a swap between $i_1$ and its predecessor $i_2$ in the BFS. For this purpose, iterate through the $\epsfrac$ classes. For each $k \in \closedinterval{\epsfrac}$ compare the biggest job of $i_2$ with the smallest job on $i_1$ using the auxiliary data structure described above. If $\max(T_{i_2,k}) > \min(T_{i_1,k})$, then perform the swap of the two just compared elements. Iterating through the classes contributes an additive term of $\epsfrac$ to the running time. Swapping jobs involves two remove and insert operations, each of them running in $\O(\log(\jobtypes  \pmax))$. Therefore, this step adds $\O(\epsfrac + \log(\jobtypes  \pmax))$ to the running time.
\end{itemize}

Assembling all the terms yields running time $\O(\epsfrac  m^2) + \O(m^2) + \O(\epsfrac + \log(\jobtypes  \pmax))= \O(\swaprt)$ for one BFS run and the subsequent swap.

\end{proof}

\begin{proposition}
    Given a solution to the $\slotmilp$, an integral solution to the target load balancing problem with an absolute error of at most $\eps \pmax$ can be computed in time $\O(\localsearchrt)$.
\end{proposition}

\begin{proof}
Due to \cref{lem:numberofswaps}, the number of swaps is bounded by $\jobtypes^3  2  m  \pmax$. By \cref{lem:oneswap}, one swap and the preceding BFS run can be carried out in $\O(\swaprt)$. Constructing the auxiliary data structure needed for the graph can be achieved in $\O(m  \jobtypes  \pmax  \log(m  \jobtypes  \pmax))$. Putting this together, we obtain $\O(m  \jobtypes  \pmax  \log(m  \jobtypes  \pmax)) + \jobtypes^3  2  m  \pmax  \O(\swaprt) = \O(\localsearchrt)$.   
\end{proof}

\subsection{Application}\label{subsec:application}

The idea by Buchem~\etal~\cite{buchem} is to guess the target load intervals for either \tf{\P}{}{\cmax}, \tf{\P}{}{\cmin} or \tf{\P}{}{\envy} in steps of $\eps  \pmax$. The local search is then performed for the best solution obtained by the DP. Based on our insights, we can make the following adjustments to the running times presented in \cite{buchem}. 

For \tf{\P}{}{\cmax}, $\ell$ is set to 0 and $u$ is guessed within the interval $$\left[\frac{1}{m}\sum\limits_{j = 1}^{\nlm} p_j, \frac{1}{m}\sum\limits_{j = 1}^{\nlm} p_j + \pmax\right].$$
\begin{corollary}
    There exists an algorithm for \tf{\P}{}{\cmax} with running time $\additiveapproxrt$ computing a schedule with makespan at most $\opt + \eps \pmax$.
\end{corollary}

For \tf{\P}{}{\cmin}, $u$ is set to $\frac{1}{m}\sum\limits_{j = 1}^{\nlm} p_j$ and $\ell$ is guessed within the interval $\left[\frac{1}{m}\sum\limits_{j = 1}^{\nlm} p_j - \pmax, \frac{1}{m}\sum\limits_{j = 1}^{\nlm} p_j\right]$.

\begin{corollary}
     There exists an algorithm for \tf{\P}{}{\cmin} with running time $\additiveapproxrt$ computing a schedule where each machine has load at least $\opt - \eps \pmax$.
\end{corollary}

For \tf{\P}{}{\envy}, both $\ell$ and $u$ are guessed from the intervals above.

\begin{corollary}
     There exists an algorithm for \tf{\P}{}{\envy} with running time $\additiveapproxrt$ computing a schedule with envy value at most $\opt + \eps \pmax$.
\end{corollary}

\section{Omitted proofs}\label{sect:omittedproofs}
Here, we give some proofs that were omitted due to space constraints.
\subsection{Miscellaneous}\label{subsect:misc}
This part covers several smaller observations like solving optimization problems via decision problems.
\paragraph*{Frank Tardos}
\franktardos*
\begin{proof}
    Apply \autoref{thm:franktardos} to the vector $\begin{pmatrix}w \\ b\end{pmatrix}$ and the bound $N\Delta+1$ to obtain a vector $\begin{pmatrix}\bar{w} \\ \bar{b}\end{pmatrix}$. Now, let $x\in[-\Delta,\Delta]^N$ and consider the vector $x'=\begin{pmatrix}x\\-1\end{pmatrix}$. Then $\norm{x'}_1=\norm{x}_1+1\leq N\Delta+1$, so we have $\textup{sign}\left(\begin{pmatrix}w \\ b\end{pmatrix}^Tx'\right)=\textup{sign}\left(\begin{pmatrix}\bar{w} \\ \bar{b}\end{pmatrix}^Tx'\right)$ by \autoref{thm:franktardos}. Now:
    \begin{align*}
        w^Tx\leq b \iff w^Tx-b\leq 0 &\iff \matrix{w \\ b}^T\matrix{x \\ -1} \leq 0 \\
        &\iff \matrix{w \\ b}^Tx' \leq 0 \\
        &\iff \textup{sign}\left(\begin{pmatrix}w \\ b\end{pmatrix}^Tx'\right)\in\{0,+1\}\\ 
        &\iff \textup{sign}\left(\begin{pmatrix}\bar{w} \\ \bar{b}\end{pmatrix}^Tx'\right)\in\{0,+1\} \\
        &\iff \matrix{\bar{w} \\ \bar{b}}^Tx' \leq 0 \\
        &\iff \matrix{\bar{w} \\ \bar{b}}^T\matrix{x \\ -1} \leq 0 \\
        &\iff \bar{w}^Tx-\bar{b}\leq 0 \\
        &\iff \bar{w}^Tx\leq \bar{b} 
    \end{align*}
    The running time is the same as in \autoref{thm:franktardos}, namely $(N+1)^{\Oh{1}}=N^{\Oh{1}}$. We have $\norm{\bar{w}}_{\infty},|b|\leq (N\Delta+1)^{\Oh{(N+1)^3}}=(N\Delta)^{\Oh{N^3}}$.
\end{proof}

\paragraph*{The Exercise From \cite{CFKLMPPS15} (Hint 3.18)}
\begin{lemma}\label{lem:paralgbound}
    For computable $f$ and any $\beta,\gamma\in\N$, we have $\log(\beta)^{f(\gamma)}\leq 2^{\frac{(f(\gamma))^2}{2}}\beta^{o(1)}$.
\end{lemma}
\begin{proof}
    We have 
    \begin{align*}
        \log(\beta)^{f(\gamma)}=2^{\log(\log(\beta))f(\gamma)}\leq2^{\frac{(\log(\log(\beta)))^2+(f(\gamma))^2}{2}}&=2^{\frac{(\log(\log(\beta)))^2}{2}}2^{\frac{(f(\gamma))^2}{2}}\\
        &=2^{\frac{(f(\gamma))^2}{2}}\beta^{o(1)}.
    \end{align*}
    The second transformation uses that $ab\leq\frac{a^2+b^2}{2}$ (which is equivalent to $(a-b)^2\geq0$) and the last transformation uses $\log(\log(\beta))^2=\oh{\log(\beta)}$.
\end{proof}

\paragraph*{Optimization vs. Decision Problems}\label{par:optvsdec}
The optimization problems considered in this paper can be easily solved by solving several decision problems, simply by doing a binary search over some interval $[0,v]$. The key ingredient is that the objective value is always bounded by some value $v$ that depends on the numbers in the input and hence $\log(v)\leq\enc{I}^{\Oh{1}}$.
For uniform $n$-fold ILPs, \mswbp and $\sumwu$, we just have to find a solution that has at least (or at most) some value $T$ that is given by the binary search. For $\cmax$ and $\cmin$, we then have an upper bound $u$ or a lower bound $\ell$ for the load of the machines and have to decide whether a schedule exists with load values in $[0,u]$ or $[\ell,v]$. For $\envy$, it is a little more tricky than for $\cmax$ and $\cmin$, because a \enquote{guess} $T$ for the envy-value only gives us the difference $u-\ell$, not the actual values $\ell$ and $u$ we will need for the framework. The following lemma solves this predicament:
\begin{lemma}\label{lem:completionbounds}
    For \tf{\P}{}{\{\cmax,\cmin,\envy\}}, there always exists an optimal solution in which each machine $i$ has load $L_i\in[L-\pmax,L+\pmax]$, where $L$ is the sum of all processing times divided by $\m$.
\end{lemma}
\begin{proof}
    Suppose there is an optimal schedule $\sigma$ that does not fulfill this property. Then $\sigma$ can be improved by moving any job from a machine $j$ with the highest load to a machine $k$ with the lowest load. This increases the load of $k$ by at most $\pmax$ and reduces the load of $j$ by at most $\pmax$. Note that the load of $k$ has to be $\leq L$ and the load of $j$ has to be $\geq L$. So this modification of the schedule can only either move the load value of the two machines into the interval $[L-\pmax,L+\pmax]$ or move it around inside; a load value cannot fall out of this interval.
    Note that this modification can only improve the objective value (w.r.t. either of the three objectives). So iterating this procedure will lead to a schedule that has the desired property and is still optimal.
\end{proof}
This implies that we can assume a given upper/lower bound $u$ or $\ell$ for the load value to lie in the interval $[L-\pmax,L+\pmax]$. Otherwise, we can decide the feasibility of the problem in polynomial time (with a greedy algorithm). 
In particular, for the $\cmin$-objective, we only have to consider configurations with a load value that is at most $\ell+2\pmax$.
Note that by the same argument, one can even show that the difference between the largest load and the smallest load can be at most $\pmax$ (without the factor $2$).

For objectives $\cmax$ and $\cmin$, reducing optimization to decision adds a factor $\log(\pmax)\leq\enc{I}^{\Oh{1}}$
to the running time. For $\envy$, we guess the minimum load $\ell$, which can have at most $\Oh{\pmax}$ different values. We then set $u=\ell+T$ and solve the decision problem. Overall, this adds the same factor $\enc{I}^{\Oh{1}}$ for the binary search, and a factor $\Oh{\pmax}$ to the running time, which yields a total factor of $\enc{I}^{\Oh{1}}\pmax$. Note that this only works if there are no additional constraints like class or capacity constraints or setup times, as \autoref{lem:completionbounds} no longer works directly. We might have to try out more values in those cases to solve the optimization problem.

\subsection{Tool 2: Proximity}
We now give proofs for \autoref{thm:balancinggeneral} and \autoref{thm:main}.

\balancinggeneral*
\begin{proof}
    First, consider the formulation of the ILP in \autoref{lem:pqilp} and its parameters.
    We have $n=\Oh{m}$ as the number of blocks, $\tau=2$ as the number of different block combinations, $t=N$ as the number of variables per block, $s=M^{(P)}$ as the number of constraints per block, $r=M^{(Q)}+N$ as the number of global constraints and the largest entry in the matrix is $\Delta=\max\left\{\norm{A^{(P)}}_\infty, \norm{A^{(Q)}}_\infty\right\}$.  It should be noted that there are upper bounds $x^{(i)}\leq u^{(i)}$ in the convexified relaxation, but the variables in the ILP do not have explicit upper bounds. However, as we know by the definition of \pqmrs, $P$ itself is bounded. Hence, with Cramer's rule~\cite{Cramer1750} we can bound the points in $P$ and indirectly also the solutions of the ILP. In particular, the solutions are bounded by $\Delta^{N-1}\max\left\{\norm{b^{(P)}}_1,\norm{b^{(Q)}}_1\right\}$ if we use the Hadamard inequality~\cite{Hadamard1893} to bound the determinants. This yields uniform artificial upper bounds $u^{(i)}=u$ with $\norm{u}_\infty=\Delta^{N-1}\max\left\{\norm{b^{(P)}}_1,\norm{b^{(Q)}}_1\right\}$.
    Hence, the largest number in the formulation is bounded by 
    $$K\leq\max\left\{\norm{A^{(P)}}_\infty, \norm{A^{(Q)}}_\infty, \norm{b^{(P)}}_\infty,\norm{b^{(Q)}}_\infty\right\}^{O(N)}.$$
    
    So using \autoref{prop:relaxation}, we can solve the convexified relaxation in time:
    \begin{align*}
        &\tau t^{\Oh{t}}2^{\Oh{r^2}} (sr\log(Kn))^{\Oh{1}}\\
        =&2 N^{\Oh{N}}2^{\Oh{(M^{(Q)}+N)^2}}
        \left(M^{(P)}(M^{(Q)}+N)\log\bigg(\max\Big\{\Delta, \right.\\
        &\left.\left.\left.         \norm{b^{(P)}}_\infty,\norm{b^{(Q)}}_\infty\right\}^{O(N)}m\right)\right)^{\Oh{1}}\\
        =& 2^{\Oh{(M^{(Q)}+N)^2}} \left(M^{(P)}\log\Big(\max\Big\{\Delta,\norm{b^{(P)}}_\infty, \right.\left.\left.\left.\norm{b^{(Q)}}_\infty\right\}m\right)\right)^{\Oh{1}}
    \end{align*}

    Suppose that $x^*$ is a solution of the convexified relaxation such that $(x^*)^{(1)}=\hdots=(x^*)^{(m)}$. This can be obtained from any solution computed by the algorithm described in \autoref{prop:relaxation}: 
    Just replace the first $m$ solution parts by their overall average. As they all lie in the convex hull of the same points (see the description of the algorithm before \autoref{prop:relaxation}), the average also lies in the convex hull and is hence feasible. 
    
    By \autoref{thm:convrelproximity}, there is a solution $z^*$ to the ILP from \autoref{lem:pqilp} such that 
    \begin{align*}
        \norm{x^*-z^*}_1
        \leq (r\Delta (s\Delta)^{\Oh{s}})^{\Oh{r}}
        &= (rs\Delta)^{\Oh{rs}} \\
        &\leq \left((M^{(Q)}+N)M^{(P)}\Delta\right)^{\Oh{(M^{(Q)}+N)M^{(P)}}}
    \end{align*}
    The rest of this proof is illustrated by \autoref{fig:proximity}.
    Let $$D:=\left((M^{(Q)}+N)M^{(P)}\Delta\right)^{\Oh{(M^{(Q)}+N)M^{(P)}}}.$$ 
    Then we can assume that the variables in the ILP are bounded by $(x^*)^{(i)}-D\leq z^{(i)}\leq (x^*)^{(i)}+D$, meaning that we can assume that $\max\left\{(x^*)^{(i)}-D,\mathbf{0}\right\}$ is certainly part of the solution (where the $\max$ is component-wise). We then only need to look for the rest of the solution, say $y^{(i)}$, which corresponds to $z^{(i)}-\max\left\{(x^*)^{(i)}-D,\mathbf{0}\right\}$. Then $z^{(i)}=y^{(i)}+\max\left\{(x^*)^{(i)}-D,\mathbf{0}\right\}$ is the overall solution.
    Hence, we can focus on the following ILP, where $\max\left\{(x^*)^{(i)}-D,\mathbf{0}\right\}$ is already a fixed part of the solution and we are looking for the remaining part $y^{(i)}$, which has to lie inside the $\ell_\infty$-box with side length $D$:
    \begin{align*}
        A^{(P)}\left(y^{(i)}+\max\left\{(x^*)^{(i)}-D,\mathbf{0}\right\}\right)&=b^{(P)} \\
        y^{(i)}&\in\N^N
    \end{align*}
    And $y^{(i)}$ is implicitly upper bounded by $2D=D$. Hence, we can modify the right-hand-side from $b^{(P)}$ to $b^{(P)}-A^{(P)}\max\left\{(x^*)^{(i)}-D,\mathbf{0}\right\}$ and then the variables are bounded by $D$ in $\ell_\infty$-norm. Note that by assumption, the solutions $(x^*)^{(i)}$ are the same for all blocks but the last one. So modifying the right-hand-side and re-translating the ILP to a \pqmr (using \autoref{lem:pqilp}) yields one modified polytope $P$.
\end{proof}

\begin{figure}
    \centering
    \begin{tikzpicture}
        \draw[-{Stealth[length=3mm]}] (0,0) -- (7.5,0);
        \draw[-{Stealth[length=3mm]}] (0,0) -- (0,4);

        \draw (-0.5,3) -- (6,-0.5);
        
        \draw (5.08,0) pic[rotate = 0] {cross=4pt};
        \node at (5.4,0.3) {$x^*$};
        
        \draw (3,0) pic[rotate = 0] {cross=4pt};
        \node at (3,-0.3) {$\max\{x^*-D,\mathbf{0}\}$};
        
        \draw (3.98,0.6) pic[rotate = 0] {cross=4pt};
        \node at (4.3,0.9) {$z^*$};

        \draw[-{Stealth[length=3mm]}] (3,0) -- (3.98,0.6) node[midway, above, sloped] {$y$};
        
        \draw[dotted] (3,-2) -- (7,-2) -- (7,2) -- (3,2) -- (3,-2);
        \draw [decorate, decoration = {calligraphic brace}] (7,-2.1) -- (5,-2.1) node[midway, below] {$D$};
    \end{tikzpicture}
    \caption{After computing the solution $x^*$ of the relaxation, $\max\{x^*-D,\mathbf{0}\}$ is taken as a fixed part of the solution $z^*$ and only the bounded rest $y$ has to be computed.}
    \label{fig:proximity}
\end{figure}

\maintheorem*
\begin{proof}
    Let $\Delta:=\max\left\{\norm{A^{(P)}}_\infty,\norm{A^{(Q)}}_\infty\right\}$.
    First, we use \autoref{thm:balancinggeneral} to solve the convexified relaxation of the ILP for $(P,Q,m)$, modify $b^{(P)}$ in time 
    \begin{align*}
        &2^{\Oh{(M^{(Q)}+N)^2}} \left(M^{(P)}\log\left(\max\left\{\norm{A^{(P)}}_\infty,\norm{A^{(Q)}}_\infty, \right.\right.\right.\\
        &\left.\left.\left.\norm{b^{(P)}}_\infty,\norm{b^{(Q)}}_\infty\right\}m\right)\right)^{\Oh{1}}
    \end{align*}
    and restrict ourselves to points in $P$ that have $\ell_\infty$-norm at most 
    \[D:=\left((M^{(Q)}+N)M^{(P)}\Delta\right)^{\Oh{(M^{(Q)}+N)M^{(P)}}}\]
    by adding explicit upper bounds to the constraints in $A^{(P)}x= b^{(P)}$. Note that this increases the number of constraints to $M^{(P)}+N$ and the number of variables to $2N$ (due to adding slack variables). We then use the result by Frank and Tardos (see \autoref{cor:franktardos}) to reduce the coefficients in time $N^{\Oh{1}}$. Afterwards, the coefficients in $A^{(P)}$ and $b^{(P)}$ are bounded by $(ND)^{\Oh{N^3}}$.
    Using \autoref{prop:jkgr}, the given problem can be solved in time 
    $\enc{P'}^{2^{\Oh{N}}}\enc{Q}^{\Oh{1}}\log(m)^{\Oh{1}}$.
    Note that the modified polytope $P'$ has encoding length at most 
    \begin{align*}
        \enc{P'} &\leq O\left(2N(M^{(P)}+N)\log\left(\max\left\{\norm{A^{(P)}}_{\infty},\norm{b^{(P)}}_{\infty}\right\}\right)\right)\\
        &=O\left(2N(M^{(P)}+N)\log\left((ND)^{\Oh{N^3}}\right)\right)\\
        &=\Oh{N^5M^{(P)}\log(ND)}
    \end{align*}

    The problem can now be solved in time: 
    
    \begin{align*}
        &\enc{P'}^{2^{\Oh{N}}}\enc{Q}^{\Oh{1}}\log(m)^{\Oh{1}}\\
        =&(N^5M^{(P)}\log(ND))^{2^{\Oh{N}}}\enc{Q}^{\Oh{1}}\log(m)^{\Oh{1}}\\
        \leq&\bigg(M^{(P)}\log\bigg(\left((M^{(Q)}+N)M^{(P)}\max\left\{\norm{A^{(P)}}_\infty, \right.\right.\\
        &\left.\left.\left.\left.\norm{A^{(Q)}}_\infty\right\}\right)^{\Oh{(M^{(Q)}+N)M^{(P)}}}\right)\right)^{2^{\Oh{N}}}\enc{Q}^{\Oh{1}}\log(m)^{\Oh{1}}\\
        \leq&\left(M^{(P)} M^{(Q)}\log\left(\Delta\right)\right)^{2^{\Oh{N}}}(\enc{Q}\log(m))^{\Oh{1}}
    \end{align*}
    Altogether (with the preprocessing), this yields an algorithm running in time
    \begin{align*}
        &\left(M^{(P)} M^{(Q)}\log\left(\Delta\right)\right)^{2^{\Oh{N}}}
        (\enc{Q}\log(m))^{\Oh{1}}+2^{\Oh{(M^{(Q)}+N)^2}} \\
        &\left(M^{(P)}\log\left(\max\left\{\norm{A^{(P)}}_\infty,\right.\right.\right.\left.\left.\left.\norm{A^{(Q)}}_\infty,\norm{b^{(P)}}_\infty,\norm{b^{(Q)}}_\infty\right\}m\right)\right)^{\Oh{1}}\\
        &+N^{\Oh{1}} \\
        =&\left(\left(M^{(P)} M^{(Q)}\log\left(\max\left\{\norm{A^{(P)}}_\infty, \norm{A^{(Q)}}_\infty\right\}\right)\right)^{2^{\Oh{N}}}\right.+2^{\Oh{(M^{(Q)})^2}}\bigg)\\
        &(\enc{P}\enc{Q}\log(m))^{\Oh{1}}
    \end{align*}
    Afterwards, we re-add the solution parts $\max\left\{(x^*)^{(i)}-D,\mathbf{0}\right\}$ and obtain a solution.
\end{proof}

As a side note, \autoref{thm:balancinggeneral} is essentially a preprocessing algorithm like the one by Govzmann \etal~\cite{GMO23} and produces a kernel for certain problems that have a \pqmr:
\begin{corollary}
    Consider a problem with a \pqmr $(P,Q,m)$, given by polyhedra $P=\left\{\left.x\in\R_{\geq0}^N\,\right|\,A^{(P)}x=b^{(P)}\right\}$ and $Q=\left\{\left.x\in\R_{\geq0}^N\,\right|\,A^{(Q)}x=b^{(Q)}\right\}$ with $M^{(P)}$ and $M^{(Q)}$ constraints, respectively. Suppose that this \pqmr resembles a \confip in the sense that vectors in $P$ are (extended) configurations and that $Q$ contains a job/item vector.
    Then by solving the convexified relaxation of the corresponding ILP from \autoref{lem:pqilp} in time 
    \begin{align*}
        &2^{\Oh{(M^{(Q)}+N)^2}} \left(M^{(P)}\log\left(\max\left\{\norm{A^{(P)}}_\infty, \norm{A^{(Q)}}_\infty, \right.\right.\right.\\
        &\left.\left.\left.\norm{b^{(P)}}_\infty,\norm{b^{(Q)}}_\infty\right\}m\right)\right)^{\Oh{1}},
    \end{align*}
    we obtain a reduced instance where the number of jobs/items in a configuration is bounded by 
    \[\left((M^{(Q)}+N)M^{(P)}\max\left\{\norm{A^{(P)}}_\infty, \norm{A^{(Q)}}_\infty\right\}\right)^{\Oh{(M^{(Q)}+N)M^{(P)}}}.\]
\end{corollary}


\subsection{Balancing lemma}\label{subsec:balancinglemma}
Here, we give a proof of the balancing result by Govzmann \etal~\cite{GMO23} (\autoref{lem:balancing}). We follow their general proof structure, but notation has been aligned with the rest of this work.

In the following, $L_{\pi,i}$ denotes the total load assigned to some machine $i$ by some schedule $\pi$ and $x_i^{(\pi)}$ refers to the vector of job multiplicities assigned to machine $i$ by schedule $\pi$. For $j \in [d]$, the $j$\textsuperscript{th} component of $x_i^{(\pi)}$ is denoted by $x_{i,j}^{(\pi)}$. 

    For machines $i_1,i_2 \in \M$, schedule $\pi$ and some job type $j$, define the \textit{gap} $\Delta_{i_1,i_2}^{\pi,j} = x_{i_1,j}^{(\pi)} - x_{i_2,j}^{(\pi)}$. The \textit{gap size} of $\pi$ is defined by $\mu(\pi) = \max_{i_1,i_2,j} \Delta_{i_1,i_2}^{\pi,j}$ and the set of \textit{gap constellations} by
    
    $$I(\pi) = \{ (i_1,i_2,j) \in \M \times \M \times \D \mid \mu(\pi) = \Delta_{i_1,i_2}^{\pi,j}\}.$$

     Govzman \etal~\cite{GMO23} show that in an arbitrary schedule where the gap size exceeds $\pmax$, the assignment vectors can be rebalanced by performing a finite sequence of exchanges between pairs of machines without affecting the load on either of the machines. The result is a schedule with gap size at most $\pmax$. A consequence is that there exists an optimal integral schedule such that the assignment vectors are component-wise close to those of a fractional schedule where all jobs are equally distributed among all machines. The original proof is broken down into a number of steps starting with one important observation.

    \begin{lemma}[Govzmann \etal~\cite{GMO23}]
        \label{obs:obs2}
        Let $b, n \in \Nwithoutzero$ with $b \leq n$ and let $S = \{s_1 , \dots , s_n \} \subset \Nwithoutzero$ be a multiset of natural numbers. Then there exists a non-empty subset $S' \subseteq S$ with $|S'| \leq b$ and some $\alpha \in \Nwithoutzero$ such that $\sum\limits_{s \in S'} s = \alpha  b$.
    \end{lemma}

\begin{proof}
For $\ell \in \{0, \dots, b\}$, define $r_{\ell} = \sum\limits_{k = 1}^{\ell} s_k \mod b$. Because $r_{\ell} \in \{0, \dots, b - 1\}$ for all $\ell$, by pigeonhole principle there exist $\gamma \in \Nwithzero, \delta \in \Nwithoutzero$ with $0 \leq \gamma < \delta \leq b$ such that $r_{\gamma} = r_{\delta}$. Therefore, $\sum\limits_{k = \gamma + 1}^{\delta} s_k = \alpha  b$ for some $\alpha \in \N$. Choosing $S' = \{s_{\gamma+ 1}, \dots, s_{\delta}\}$ yields the claim.         
\end{proof}

    The idea is to apply this result to a vector of job multiplicities $z = (z_1, \dots, z_d)$, which can also be described by a multiset with $d$ different elements and $n'$ elements in total for some $d \leq n'$. Provided that the multiset has at least $p_j$ elements, it is possible by \autoref{obs:obs2} to find a subset of jobs whose processing times sum up to a multiple of $p_j$. This is expressed by the following lemma. 

    \begin{lemma}[Govzmann \etal~\cite{GMO23}]
        \label{lem:govzmanl3}
        Fix some job type $j \in \D$ and let $z = (z_1, \dots, z_d)$ be a vector with $\sum\limits_{j' \in D} z_{j'} \geq p_{j}$. Then there exists a vector $z' \neq \mathbf{0}$ such that
        \begin{itemize}
            \item $0 \leq z'_{j'} \leq z_{j'}$ for all $j' \in [d]$,
            \item $\sum\limits_{j' \in [d]}p_{j'}z'_{j'} = \alpha p_j$ for some $\alpha \in \N$ and
            \item $\sum\limits_{j' \in [d]}z'_{j'} \leq p_j$.
        \end{itemize}
    \end{lemma}

    \begin{proof}
        From $z$, construct the multiset $S$ containing $z_{j'}$ copies of processing time $p_{j'}$ for each $j' \in [d]$. Set $b = p_j$. Now \autoref{obs:obs2} can be applied, yielding a non-empty subset $S' \subseteq S$ with $|S'| \leq p_j$ such that $\sum\limits_{p \in S'} p = \alpha  b$ for some $\alpha \in \N$. Build a vector $z'$ where $z'_{j'}$ equals the number of copies of $p_{j'}$ in $S'$ for each $j' \in \D$. As $S'$ contains a subset of the jobs in $z$, $0 \leq z'_{j'} \leq z_{j'}$ holds for all $j' \in [d]$. The condition $z' \neq \mathbf{0}$ is met because $S'$ is non-empty. $\sum\limits_{j' \in [d]}z'_{j'} \leq p_j$ holds because of $|S'| \leq p_j$.
    \end{proof}

    These insights can now be used for the proof of one of the key results.   

    \begin{lemma}[Govzmann \etal~\cite{GMO23}]
        \label{lem:govzmanl4}
        For \tf{\P}{}{\{\cmax,\cmin,\envy\}}\footnotemark{}, there exists an optimal schedule $\pi$ such that the gap size $\mu(\pi)$ is at most $\pmax$.
    \end{lemma}

\footnotetext{The original document includes one more objective, $\sum\limits_{i \in \M} f(L_i^{\pi})$, where $f$ is some convex oracle function evaluating the load on machine $i \in \M$ for some schedule $\pi$.}

    \begin{proof}
        Let $\pi_1$ be an optimal schedule with $\mu(\pi_1) > \pmax$. Govzman \etal~\cite{GMO23} show that $\pi_1$ can be transformed into a new schedule $\pi_2$ such that 
        \begin{enumerate}
            \item $\mu(\pi_2) \leq \mu(\pi_1)$ and \label[cons]{cons:gapsize}
            \item if $\mu(\pi_2) = \mu(\pi_1)$ then the set $I(\pi_2)$ of gap constellations is strictly smaller than $I(\pi_1)$. \label[cons]{cons:gapconst}
        \end{enumerate}

        As $\mu(\pi_1) \leq \max_{j \in [d]}n_j$ and $|I(\pi_1)| \leq m^2\pmax$, the number of such transformation steps needed to obtain the target schedule is bounded. 
        For one transformation step, choose $i_1,i_2 \in \M, j \in [d]$ such that $\Delta_{i_1,i_2}^{\pi_1,j} = x_{i_1,j}^{(\pi_1)} - x_{i_2,j}^{(\pi_1)} = \mu(\pi_1)$. There are two cases to consider:

        If $L_{\pi_1,i_1} - L_{\pi_1,i_2} \geq p_j$, then one job of type $j$ can be transferred from $i_1$ to $i_2$. This yields a schedule $\pi_2$ with $L_{\pi_2,i_1} = L_{\pi_1,i_1} - p_j$ and $L_{\pi_2,i_2} = L_{\pi_1,i_2} + p_j$. Note that for the objective functions $\cmin,\cmax$ and $\envy$, the value of $\pi_2$ can only improve or remain unchanged compared to $\pi_1$.  Hence, $\pi_2$ is still an optimal schedule. Moreover, $\Delta_{i_1,i_2}^{\pi_2,j}$ is strictly smaller than $\Delta_{i_1,i_2}^{\pi_1,j}$. Therefore, $\mu(\pi_2) \leq \mu(\pi_1)$ and if $\mu(\pi_2) = \mu(\pi_1)$, then $|I(\pi_2)| < |I(\pi_1)|$.

        In the other case, where $L_{\pi_1,i_1} - L_{\pi_1,i_2} < p_j$, Govzman \etal~\cite{GMO23} extract a vector $z$ of jobs which can safely be moved from $i_2$ to $i_1$ without widening the gap size. This vector is described by $z_{j'} = \max\{0, x_{i_2,j'}^{(\pi_1)} - x_{i_1,j'}^{(\pi_1)}\}$ for each $j' \in [d]$. In other words, it contains only jobs of those types where $i_2$ outnumbers $i_1$. Next, the following inequality is derived for the processing time of $z$:

        \begin{align}
            \sum\limits_{j' \in [d]}z_{j'}p_{j'} &= \sum\limits_{j' \in [d]}\max\{0, x_{i_2,j'}^{(\pi_1)} - x_{i_1,j'}^{(\pi_1)}\}p_{j'}\label[ineq]{ineq:i1} \\
            &= \sum\limits_{j' \in [d] \setminus \{j\}}\max\{0, x_{i_2,j'}^{(\pi_1)} - x_{i_1,j'}^{(\pi_1)}\}p_{j'} + \max\{0, x_{i_2,j}^{(\pi_1)} - x_{i_1,j}^{(\pi_1)}\}p_{j}\\
            &\geq \sum\limits_{j' \in [d] \setminus \{j\}}( x_{i_2,j'}^{(\pi_1)} - x_{i_1,j'}^{(\pi_1)})p_{j'} + \max\{0, x_{i_2,j}^{(\pi_1)} - x_{i_1,j}^{(\pi_1)}\}p_{j}\\
            &= \sum\limits_{j' \in [d] \setminus \{j\}} x_{i_2,j'}^{(\pi_1)}p_{j'} - \sum\limits_{j' \in [d] \setminus \{j\}} x_{i_1,j'}^{(\pi_1)}p_{j'} + \max\{0, x_{i_2,j}^{(\pi_1)} - x_{i_1,j}^{(\pi_1)}\}p_{j}\\
            &= L_{\pi_1,i_2} - x_{i_2,j}^{(\pi_1)}p_{j} - (L_{\pi_1,i_1} - x_{i_1,j}^{(\pi_1)}p_{j}) + \max\{0, x_{i_2,j}^{(\pi_1)} - x_{i_1,j}^{(\pi_1)}\}p_{j}\\
            &= L_{\pi_1,i_2} - L_{\pi_1,i_1} + (x_{i_1,j}^{(\pi_1)} - x_{i_2,j}^{(\pi_1)})p_j + \max\{0, x_{i_2,j}^{(\pi_1)} - x_{i_1,j}^{(\pi_1)}\}p_{j}\\
            &\geq L_{\pi_1,i_2} - L_{\pi_1,i_1} + (x_{i_1,j}^{(\pi_1)} - x_{i_2,j}^{(\pi_1)})p_j\\
            &\geq L_{\pi_1,i_2} - L_{\pi_1,i_1} + (\pmax + 1)p_j \label[ineq]{ineq:i8}\\
            &> -p_j + (\pmax + 1)p_j\label[ineq]{ineq:i9}\\ 
            &= \pmax p_j
        \end{align}

        \cref{ineq:i8} holds because of the assumption $x_{i_2,j}^{(\pi_1)} - x_{i_1,j}^{(\pi_1)}  = \Delta_{i_1,i_2}^{\pi_1,j} = \mu(\pi_1) > \pmax$ and \cref{ineq:i9} follows from $L_{\pi_1,i_1} - L_{\pi_1,i_2} < p_j$. 

        The next step is to show that the requirement for \autoref{lem:govzmanl3} is met:

        $$\sum\limits_{j' \in [d]}z_{j'} = \sum\limits_{j' \in [d]}z_{j'} \frac{p_{j'}}{p_{j'}} \geq \sum\limits_{j' \in [d]}z_{j'} \frac{p_{j'}}{\pmax} = \frac{1}{\pmax} \sum\limits_{j' \in [d]}z_{j'} p_{j'} \geq p_j$$
    
       At this point, \autoref{lem:govzmanl3} can be applied, yielding a vector $z' \in \N^d$ with $z' \neq \mathbf{0}$ such that $0 \leq z'_{j'} \leq z_{j'}$ for all $j' \in [d]$ and $\sum\limits_{j' \in [d]}p_{j'}z'_{j'} = \alpha p_j$ for some $\alpha \in \N$ and $\sum\limits_{j' \in [d]}z'_{j'} \leq p_j$. This vector represents the set of jobs from $i_2$ to be be swapped with $\alpha$ jobs of type $p_j$ from machine $i_1$. Since the load remains unchanged on both machines, $\pi_2$ is also an optimal schedule. 

       Recall \cref{cons:gapsize,cons:gapconst}, claiming a gap size of at most $\mu(\pi_1)$ for $\pi_2$ and a decreased number of gap constellations in case the gap size did not shrink. Govzman \etal~\cite{GMO23} proceed by showing that the swap did not move too many jobs of type $j$ from $i_2$ to $i_1$ and thereby reverse the gap. More formally:
       $$\alpha = \frac{1}{p_j} \sum\limits_{j' \in [d]}p_{j'}z'_{j'} \leq \frac{1}{p_j} \sum\limits_{j' \in [d]}\pmax z'_{j'} = \frac{\pmax}{p_j} \sum\limits_{j' \in [d]} z'_{j'} \leq \pmax$$

       The last inequality holds because of $\sum\limits_{j' \in [d]}z'_{j'} \leq p_j$. Note that $\alpha$ cannot exceed the number of jobs of type $p_j$ because of $\alpha \leq \pmax < \Delta_{i_1,i_2}^{\pi_1,j} = x_{i_1,j}^{(\pi_1)} - x_{i_2,j}^{(\pi_1)} = \mu(\pi_1)$.
       Now recall that $\Delta_{i_1,i_2}^{\pi_1,j}$ was strictly greater than $\pmax$. For $\Delta_{i_2,i_1}^{\pi_2,j}$, this implies:
       \begin{align*}
       x_{i_2,j}^{(\pi_2)} - x_{i_1,j}^{(\pi_2)} = x_{i_2,j}^{(\pi_1)} + \alpha - (x_{i_1,j}^{(\pi_1)} - \alpha) &= x_{i_2,j}^{(\pi_1)} - x_{i_1,j}^{(\pi_1)} + 2 \alpha \\ &< - \pmax + 2 \alpha \\ &\leq - \pmax + 2 \pmax \\ &= \pmax 
       \end{align*}
       
       Furthermore, the number of jobs on $i_1$ strictly decreases because of $\alpha > 0$. 
       It remains to prove that the number of gap constellations involving other machines did not increase in $\pi_2$. Let $\Delta_{i_3,i_4}^{\pi_2,j'}$ be some gap constellation introduced by the swap. Then either $i_3 = i_1$ and $x_{i_3,j'}^{(\pi_2)} = x_{i_2,j'}^{(\pi_1)}$ or $i_4 = i_2$ and $x_{i_4,j'}^{(\pi_2)} = x_{i_2,j'}^{(\pi_1)}$. In other words, the transformation of $\pi_1$ into $\pi_2$ only caused some gap constellations to be swapped between $i_1$ and $i_2$ instead of creating new ones. This settles the proof of \autoref{lem:govzmanl4}.

    \end{proof}

    Let $\hat{\pi}$ be a fractional assignment where all jobs are equally distributed among all machines, i.e., where each machine gets assigned $\frac{n_p}{m}$ jobs of processing time $p$ for each $p \in D$. Govzman \etal~\cite{GMO23} use this optimal fractional schedule $\hat{\pi}$ as a reference for an optimal integral schedule and establish the following.
    \begin{lemma}[Govzmann \etal~\cite{GMO23}]
        \label{lem:govzmanl5}
        For \tf{\P}{}{\{\cmax,\cmin,\envy\}}, there exists an optimal schedule $\pi$ such that $|x_{i,j}^{(\pi_1)} - x_{i,j}^{(\hat{\pi})}| \leq \pmax$ for all $j \in [d], i \in \M$.
    \end{lemma}

    \begin{proof}

        Consider an optimal integral schedule $\pi$ with $\mu(\pi) \leq \pmax$. Such a schedule exists by \autoref{lem:govzmanl4}. Let $j \in [d]$ and let $i_1, i_2 \in \M$ be two distinct machines such that $x_{i_1,j}^{(\pi)} \leq x_{i_1,j}^{(\hat{\pi})}$ and $x_{i_2,j}^{(\pi)} \geq x_{i_2,j}^{(\hat{\pi})}$. Note that such a pair of machines exists for each job type. For any machine $i \in \M$,         
        \begin{align}
            x_{i,j}^{(\hat{\pi})} - x_{i,j}^{(\pi)} &= x_{i,j}^{(\hat{\pi})} - x_{i_2,j}^{(\pi)} + x_{i_2,j}^{(\pi)} - x_{i,j}^{(\pi)}\label[equal]{eq:i20} \\
            &\leq x_{i,j}^{(\hat{\pi})} - x_{i_2,j}^{(\pi)} + \pmax\label[ineq]{ineq:i21}\\
            &= x_{i_2,j}^{(\hat{\pi})} - x_{i_2,j}^{(\pi)} + \pmax\label[equal]{eq:i22}\\
            &\leq \pmax \label[ineq]{ineq:i23}
        \end{align}
        \cref{ineq:i21} follows from \autoref{lem:govzmanl4}. \cref{eq:i22} holds because $x_{i',j'}^{(\hat{\pi})} = x_{i'',j'}^{(\hat{\pi})}$ for each $j' \in [d], i', i'' \in \M$. Analogously:
        \begin{align}
             x_{i,j}^{(\pi)} - x_{i,j}^{(\hat{\pi})} &= 
             x_{i,j}^{(\pi)} - x_{i_1,j}^{(\pi)} + x_{i_1,j}^{(\pi)} -  x_{i,j}^{(\hat{\pi})}\label[equal]{eq:i24}\\
            &\leq \pmax + x_{i_1,j}^{(\pi)} -  x_{i,j}^{(\hat{\pi})}\\
            &= \pmax + x_{i_1,j}^{(\pi)} - x_{i_1,j}^{(\hat{\pi})}\\
            &\leq \pmax\label[ineq]{ineq:i27}
        \end{align}
        
    \end{proof}

From these insights, the balancing result can finally be derived. 

    \balancinglemma*

\begin{proof}

As in the proof of \autoref{lem:govzmanl5}, consider an optimal integral schedule $\pi$ with $\mu(\pi) \leq \pmax$ and let $i_1, i_2 \in \M$ be two distinct machines such that $x_{i_1,j}^{(\pi)} \leq x_{i_1,j}^{(\hat{\pi})}$ and $x_{i_2,j}^{(\pi)} \geq x_{i_2,j}^{(\hat{\pi})}$ for some $j \in \D$. Then by \labelcref{eq:i20} to \labelcref{ineq:i23}, $x_{i,j}^{(\pi)}$ is at least $x_{i,j}^{(\hat{\pi})} - \pmax = \frac{n_j}{m} - \pmax$. Because the jobs are integrally assigned and $\left \lceil \frac{n_j}{m}\right \rceil - 1 - \pmax$ is strictly smaller than $\frac{n_j}{m} - \pmax$, $x_{i,j}^{(\pi)}$ must be at least $\left \lceil \frac{n_j}{m}\right \rceil - \pmax$. For the same reason, $x_{i,j}^{(\pi)}$ cannot exceed $\left \lfloor \frac{n_j}{m}\right \rfloor + \pmax$. This makes it possible to preprocess problem instances \tf{\P}{}{\{\cmax,\cmin,\envy\}} by assigning $\left \lceil \frac{n_j}{m}\right \rceil - \pmax$ jobs of each type to each machine. The outcome is a kernel with at most $2\pmax  m$ jobs of each job type. Moreover, because of $x_{i,j}^{(\pi)} \leq \left \lfloor \frac{n_j}{m}\right \rfloor + \pmax$, at most $2  \pmax$ of the remaining jobs of each type $j \in \D$ need to be assigned to any machine $i \in \M$. For $d$ job types, each with processing time at most $\pmax$, this limits the load on every machine on the remaining instance to $2  \pmax^2d$. The running time of this preprocessing step is $\Oh{d}$. 

\end{proof}

\section{Complexity}\label{sect:complexity}
This section covers lower bounds and a connection between uniform and identical machines. In the context of uniform machines ($\Q$), we have $\tau$ different machine types (and hence $m\in\N^\tau$ is a vector), each with its own speed $s_i$, and processing a job of type $j$ on a machine of type $i$ induces load $\frac{p_j}{s_i}$. So $\P$ is the special case of $\Q$ where $\tau=1$.

\subsection{Connection Between Uniform and Identical Machines}
It is a long standing question whether \tf{\P}{}{\cmax} is FPT w.r.t. parameter $d$ (with high-multiplicity encoding). The following result shows that this question is equivalent to asking whether \tf{\Q}{}{\cmax} is FPT with parameters $d$ and $\tau$. 

\begin{restatable}{theorem}{connectioncmax}
    The following statements are equivalent:
    \begin{enumerate}
        \item \tf{\P}{}{\cmax} is FPT w.r.t. $d$ (in the high-multiplicity setting).
        \item \tf{\Q}{}{\cmax} is FPT w.r.t. $d$ and $\tau$ (in the high-multiplicity setting, for jobs and machines).
    \end{enumerate}
\end{restatable}
\begin{proof}
    $(2.\implies 1.)$ This direction is trivial: An algorithm with running time $f(d,\tau)\enc{I}^{\Oh{1}}$ for \tf{\Q}{}{\cmax} solves \tf{\P}{}{\cmax} in time $f(d,1)\enc{I}^{\Oh{1}}=f(d)\enc{I}^{\Oh{1}}$.

    $(1.\implies 2.)$ We give a reduction from \tf{\Q}{}{\cmax} to \tf{\P}{}{\cmax}. For an instance $I$ of \tf{\Q}{}{\cmax}, consisting of a processing time vector $p\in\N^d_{>0}$, job multiplicity vector $n\in\N^d_{>0}$, speed vector $s\in\N^\tau_{>0}$, machine multiplicity vector $m\in\N^\tau_{>0}$ and makespan threshold $u\in\N_{>0}$, we construct an instance $I'$ of \tf{\P}{}{\cmax} like this: We set the new makespan threshold to $u':=2\smax u+1$. The new processing times are $p_1,\hdots,p_d,u'-s_1u,\hdots,u'-s_\tau u$, the multiplicities are $n_1,\hdots,n_d,$ $m_1,\hdots,m_\tau$ and there are $m_1+\hdots+m_\tau$ (identical) machines.

    Suppose that the \tf{\Q}{}{\cmax} instance $I$ is feasible, i.e., there is a schedule $\sigma$ with makespan at most $u$. We construct a schedule $\sigma'$ for the \tf{\P}{}{\cmax} instance $I'$ in the following way: For every machine $k$ of type $i\in[\tau]$ in $I$, take all the jobs that $\sigma$ schedules on $k$ and place them on any one of the machines in $I'$, together with one of the dummy jobs that has processing time $u'-s_iu$. Note that there are as many dummy jobs as there are machines in $I$, matching the number of machines of each type. We schedule the jobs from all machines this way, obtaining schedule $\sigma'$; now, since $\sigma$ is feasible (so each machine has jobs with total processing time at most $s_iu$), the jobs scheduled on some machine $l\in[m_1+\hdots+m_\tau]$ in $\sigma'$ have total processing time at most $u'-s_iu+s_iu=u'$, if the original machine had type $i$. So in $\sigma'$, each machine has completion time at most $u'$ and hence $I'$ is feasible. 

    For the other direction, suppose that the constructed instance $I'$ is feasible, i.e., there is a schedule $\sigma'$ with makespan at most $u'$. The key observation is that no two dummy jobs can be scheduled on the same machine, since $u'-s_iu+u'-s_ju\geq2u'-2s_{\min}u=4\smax u+2-2s_{\min}u\geq2\smax u+2>2\smax u+1=u'$ holds for any machine types $i,j\in[\tau]$. Construct a schedule $\sigma$ for $I$ as follows: For each machine $k\in[m_1+\hdots+m_\tau]$, take all non-dummy jobs that are scheduled on this machine and place them on a machine of type $i\in[\tau]$, where $i$ is the type corresponding to the (singular) dummy job placed on machine $k$ in $\sigma'$. Doing this for each machine, we obtain a schedule $\sigma$ for all jobs. Now, consider some machine $k$ of type $i\in[\tau]$ and the jobs scheduled on it by $\sigma$. Their total processing time is at most $u'-(u'-s_iu)=s_iu$, which means that in $\sigma$, for every type $i\in[\tau]$, each machine of type $i$ has load at most $s_iu$, and hence every job is completed within time $u$. So $I$ is feasible.

    Now, consider the parameters of the constructed \tf{\P}{}{\cmax} instance. We have parameters
    \begin{itemize}
        \item $\pmax'=\Oh{\smax u}$,
        \item $d'=d+\tau$,
        \item $u'=\Oh{\smax u}$,
        \item $m'=\Oh{\tau m_{\max}}$ and
        \item $\nmax'=\max\{\nmax,m_{\max}\}$.
    \end{itemize}
    So the encoding length $\enc{I'}$ is bounded by:
    \begin{align*}
        \enc{I'}\leq&\Oh{\log(m')+d'(\log(\pmax')+\log(\nmax'))+\log(u')} \\
        =&\Oh{\log(\tau m_{\max})+(d+\tau)(\log(\smax u)\\
        &+\log(\max\{\nmax,m_{\max}\}))+\log(\smax u)} \\
        =&\enc{I}^{\Oh{1}}
    \end{align*}
    Now, suppose there exists an algorithm for \tf{\P}{}{\cmax} with running time $f(d')\enc{I'}^{\Oh{1}}$. Given an instance $I$ of \tf{\Q}{}{\cmax}, we can construct an instance $I'$ of \tf{\P}{}{\cmax} using the above reduction, in time $\enc{I}^{\Oh{1}}$. We then apply the algorithm for \tf{\P}{}{\cmax} to $I'$ in time $f(d')\enc{I'}^{\Oh{1}}=f(d+\tau)\enc{I}^{\Oh{1}}$ and transform the solution into a solution for $I'$, in time $\enc{I}^{\Oh{1}}$. In total, this yields an algorithm for \tf{\Q}{}{\cmax} that runs in time $\enc{I}^{\Oh{1}}+f(d+\tau)\enc{I}^{\Oh{1}}+\enc{I}^{\Oh{1}}=f(d+\tau)\enc{I}^{\Oh{1}}$, which concludes the proof.
\end{proof}

A similar reduction also works for the objective $\cmin$:
\begin{restatable}{theorem}{connectioncmin}
    The following statements are equivalent:
    \begin{enumerate}
        \item \tf{\P}{}{\cmin} is FPT w.r.t. $d$ (in the high-multiplicity setting).
        \item \tf{\Q}{}{\cmin} is FPT w.r.t. $d$ and $\tau$ (in the high-multiplicity setting, for jobs and machines).
    \end{enumerate}
\end{restatable}
\begin{proof}
    $(2.\implies 1.)$ This direction is again trivial: An algorithm with running time $f(d,\tau)\enc{I}^{\Oh{1}}$ for \tf{\Q}{}{\cmin} can solve an instance of \tf{\P}{}{\cmin} in time $f(d,1)\enc{I}^{\Oh{1}}=f(d)\enc{I}^{\Oh{1}}$.

    $(1.\implies 2.)$ We give a reduction from \tf{\Q}{}{\cmin} to \tf{\P}{}{\cmin}. For an instance $I$ of \tf{\Q}{}{\cmin}, consisting of a processing time vector $p\in\N^d_{>0}$, job multiplicity vector $n\in\N^d_{>0}$, speed vector $s\in\N^\tau_{>0}$, machine multiplicity vector $m\in\N^\tau_{>0}$ and completion time lower bound $\ell$, we construct an instance $I'$ of \tf{\P}{}{\cmin} like this: Let $S:=\sum_{i=1}^dn_ip_i$ be the sum of all processing times in the given instance. We define the new completion time lower bound as $\ell':=S+\smax \ell+1$ and the new processing times as $p_1,\hdots,p_d,\ell'-s_1l,\hdots,\ell'-s_\tau \ell$. The multiplicities are $n_1,\hdots,n_d,m_1,\hdots,m_\tau$ and there are $m_1+\hdots+m_\tau$ (identical) machines.

    Suppose that the \tf{\Q}{}{\cmin} instance $I$ is feasible, i.e., there is a schedule $\sigma$ with completion time at least $\ell$ on each machine. We construct a schedule $\sigma'$ for the \tf{\P}{}{\cmin} instance $I'$ in the following way: For every machine $k$ of type $i\in[\tau]$ in $I$, take all the jobs that $\sigma$ schedules on $k$ and place them on any one of the machines in $I'$, together with one of the dummy jobs that has processing time $\ell'-s_i\ell$. Note that there are as many dummy jobs as there are machines in $I$, matching the number of machines of each type. We schedule the jobs from all machines this way, obtaining schedule $\sigma'$; now, since $\sigma$ is feasible, the jobs scheduled on some machine $k\in[m_1+\hdots+m_\tau]$ in $\sigma'$ have completion time at least $s_i\ell+\ell'-s_i\ell=\ell'$, if the original machine had type $i$. So in $\sigma'$, each machine has completion time at least $\ell'$ and hence $I'$ is feasible. 

    For the other direction, suppose that the constructed instance $I'$ is feasible, i.e., there is a schedule $\sigma'$ with completion time at least $\ell'=S+\smax \ell+1$ on each machine. The key observation is that there has to be at least one dummy job on each machine, as all non-dummy jobs together have processing time $S$, which is strictly smaller than $\ell'$ (even without taking the speeds into account). Since there are as many dummy jobs as machines, there has to be exactly one dummy job on each machine. Construct a schedule $\sigma$ for $I$ as follows: For each machine $k\in[m_1+\hdots+m_\tau]$, take all non-dummy jobs that are scheduled on this machine and place them on a machine of type $i\in[\tau]$, where $i$ is the type corresponding to the (singular) dummy job placed on machine $k$ in $\sigma'$. Doing this for each machine, we obtain a schedule $\sigma$ for all jobs. Now, consider some machine $k$ of type $i\in[\tau]$ and the jobs scheduled on it by $\sigma$. Their total processing time is at least $\ell'-(\ell'-s_i\ell)=s_i\ell$, which means that in $\sigma$, for every type $i\in[\tau]$, each machine of type $i$ has load at least $s_i\ell$, and hence the time needed to process these jobs is at least $\ell$. So $I$ is feasible.

    Now, consider the parameters of the constructed \tf{\P}{}{\cmin} instance. We have parameters
    \begin{itemize}
        \item $\pmax'=S+\smax \ell+1=\Oh{\n\pmax+\smax \ell}$,
        \item $d'=d+\tau$,
        \item $\ell'=S+\smax \ell+1=\Oh{\n\pmax+\smax \ell}$,
        \item $m'=\Oh{\tau m_{\max}}$ and
        \item $\nmax'=\max\{\nmax,m_{\max}\}$.
    \end{itemize}
    So the encoding length $\enc{I'}$ is bounded by:
    \begin{align*}
        \enc{I'}\leq&\Oh{\log(m')+d'(\log(\pmax')+\log(\nmax'))+\log(\ell')} \\
        =&\Oh{\log(\tau m_{\max})+(d+\tau)(\log(\n\pmax+\smax \ell)\\
        &+\log(\max\{\nmax,m_{\max}\}))\\
        &+\log(\n\pmax+\smax \ell)} \\
        =&\enc{I}^{\Oh{1}}
    \end{align*}
    Now, suppose there exists an algorithm for \tf{\P}{}{\cmin} with running time $f(d')\enc{I'}^{\Oh{1}}$. Given an instance $I$ of \tf{\Q}{}{\cmin}, we can construct an instance $I'$ of \tf{\P}{}{\cmin} using the above reduction, in time $\enc{I}^{\Oh{1}}$. We then apply the algorithm for \tf{\P}{}{\cmin} to $I'$ in time $f(d')\enc{I'}^{\Oh{1}}=f(d+\tau)\enc{I}^{\Oh{1}}$ and transform the solution into a solution for $I'$, in time $\enc{I}^{\Oh{1}}$. In total, this yields an algorithm for \tf{\Q}{}{\cmin} that runs in time $\enc{I}^{\Oh{1}}+f(d+\tau)\enc{I}^{\Oh{1}}+\enc{I}^{\Oh{1}}=f(d+\tau)\enc{I}^{\Oh{1}}$, which concludes the proof.
\end{proof}
Note that these reductions also yield algorithms for \tf{\Q}{}{\{\cmax,\cmin\}} by applying the reduction and then solving the resulting instance with an algorithm for \tf{\P}{}{\{\cmax,\cmin\}}.

For the $\envy$-objective, the above reductions do not seem to work directly. Having only the one dummy job from the $\cmin$-reduction is not enough, as we cannot upper bound the processing time of jobs coming from machines of a specific speed value $s_i$. One could fix this by adding another dummy job for each machine (so that there are two dummy jobs on each machine). Making sure that there are exactly two dummy jobs on each machine should be easy by setting the processing times and the thresholds accordingly, but we have to assure that there are always two dummy jobs \emph{of the same type $i$} on the same machine (essentially making it a machine of type $i$). This seems to be much more complicated.

\subsection{Lower Bounds for \texorpdfstring{\tf{\P}{}{\cmax}}{P||Cmax}}
Chen \etal~\cite{CJZ18} gave two reductions from 3-SAT to \tf{\P}{}{\cmax} to show lower bounds for approximation algorithms:
\begin{theorem}[First reduction by Chen \etal~\cite{CJZ18}]\label{thm:chenreductionone}
	For every $\delta>0$ with $\frac{1}{\delta}$ integer, \textsc{3-Sat}' (a slightly restricted variant of \textsc{3-Sat} that retains the same ETH-hardness as \textsc{3-Sat}) with $N$ variables can be reduced to \tf{\P}{}{\cmax} with:
	\begin{itemize}
		\item $d=\n=\Oh{\frac{N}{\delta}}=\Oh{N}$ jobs (each with its own processing time),
		\item $m=\Oh{\frac{N}{\delta}}=\Oh{N}$ machines,
		\item makespan threshold $u=\Oh{2^{\frac{3}{\delta}}N^{1+\delta}}=\Oh{N^{1+\delta}}$,
		\item largest processing time $\pmax=\Oh{N^{1+\delta}}$ and
		\item encoding length
  \begin{align*}
      \enc{I}&=\Oh{d\log(\pmax)+\log(m)}+\log(u)\\
      &=O\left(\frac{N}{\delta}\log(N^{1+\delta})+\log\left(\frac{N}{\delta}\right)\right)+\log(N^{1+\delta})\\
      &=\Oh{N\log(N)},
  \end{align*}
	\end{itemize}
	assuming that $\delta$ is set to some constant.
\end{theorem}

\begin{theorem}[Second reduction by Chen \etal~\cite{CJZ18}]\label{thm:chenreductiontwo}
	\textsc{3-Sat} with $N$ variables can be reduced to \tf{\P}{}{\cmax} with:
	\begin{itemize}
		\item An arbitrary number $m$ of machines (this can be freely chosen),
		\item $\n=\Oh{N+m}$ jobs,
		\item $d=\Oh{N}$ different processing times,
        \item $\nmax=\Oh{m}$ as the highest multiplicity,
		\item $\pmax=2^{O\left(\frac{N\log^2(m)}{m}\right)}$ and hence 
		\item encoding length:
  \begin{align*}
      \enc{I}&=\Oh{d\log(\pmax+\nmax)+\log(m)}\\
      &=O\left(N\log\left(2^{O\left(\frac{N\log^2(m)}{m}\right)}+m\right)+\log(m)\right)\\
      &=O\left(N\frac{N\log^2(m)}{m}\log(m)+\log(m)\right)\\
      &=O\left(\frac{N^2\log^3(m)}{m}\right)
  \end{align*}
	\end{itemize}
        Moreover, the constructed instances have total processing time at most $P\leq Nm^{O\left(\frac{N\log(m)}{m}\right)}$ if $m>1$. The makespan threshold $u$ can also be bounded by this value, hence adding $\log(u)$ to the encoding length does not change $\enc{I}$.
\end{theorem}

These reductions can be used to show the following lower bounds:
\exactlowerbound*
\begin{proof}
        Let $c\in\R$ be a constant such that $$\pmax^{\Oh{1}}=\left(2^{O\left(\frac{N\log^2(m)}{m}\right)}\right)^{\Oh{1}}=2^{c\frac{N\log^2(m)}{m}}.$$ Then for any constant $\delta>0$, we can set $m$ so large that $2^{c\frac{N\log^2(m)}{m}}\leq2^{\delta N}$. Hence, multiplying $\pmax^{\Oh{1}}$ with any of the terms below yields a better-than-ETH running time.
	
        As we set $m$ large enough, we also get $\frac{\log^3(m)}{m}=\Oh{1}$ and hence $\enc{I}=\Oh{N^2}$.
	
	1.) We have for some constant $c$ that
	\begin{align*}
		\enc{I}^{\Oh{d^{1-\eps}}}=N^{\Oh{N^{1-\eps}}}=N^{cN^{1-\eps}}=(2^{\log(N)})^{cN^{1-\eps}}=2^{\oh{N}}
	\end{align*}
	and such a running time would contradict the ETH.
	
	2.) We have for some constants $c_1,c_2,c_3$ that
	\begin{align*}
		\enc{I}^{\Oh{1}}d^{\Oh{d^{1-\eps}}}&=N^{\Oh{1}}\Oh{N}^{\Oh{N}^{1-\eps}}\\
  &=N^{c_1}(c_2N)^{(c_3N)^{1-\eps}}\\
  &=2^{c_1\log(N)}(2^{\log(c_2N)})^{(c_3N)^{1-\eps}}\\
		&=2^{c_1\log(N)}2^{\log(c_2N)(c_3N)^{1-\eps}}\\
		&=2^{\oh{N}}
	\end{align*}
	and such a running time would contradict the ETH.
	
	3.) We have
	\begin{align*}
		\enc{I}^{\oh{\frac{d}{\log(d)}}}=N^{\oh{\frac{N}{\log(N)}}}=(2^{\log(N)})^{\oh{\frac{N}{\log(N)}}}&=2^{\log(N)\oh{\frac{N}{\log(N)}}}\\
        &=2^{\oh{N}}
	\end{align*}
	and such a running time would contradict the ETH.

    4.) Similarly, we have
    \begin{align*}
        \pmax^{O(d^{1-\eps})}=(N^{1+\delta})^{O(N^{1-\eps})}=2^{\log(N)(1+\delta)O(N^{1-\eps})}=2^{o(N)}
    \end{align*}
    using the first reduction.
\end{proof}

With the first reduction (\autoref{thm:chenreductionone}), one can also show the following:
\begin{corollary}\label{cor:lb2exact}
	Let $\delta>0$. Then \tf{\P}{}{\cmax} cannot be solved in time $2^{\Oh{d+m+l+u+\pmax+\enc{I}}^{1-\delta}}$ unless the ETH fails. 
\end{corollary}
\begin{proof}
    Note that in \autoref{thm:chenreductionone}, $$d+m+l+u+\pmax+\enc{I}=O(N^{1+\delta'}\log(N)),$$ where $\delta'$ is the constant from the reduction. So a running time
    $$2^{\Oh{d+m+l+u+\pmax+\enc{I}}^{1-\delta}}=2^{o(N)}$$ would contradict the ETH.
\end{proof}
Note that one can also replace the $2$ in the basis by any constant $>2$, in fact even by anything that is at most $N$ in the reduction.

\begin{corollary}\label{cor:lb2approx}
Let $\delta>0$. Then \tf{\P}{}{\cmax} cannot be approximated with additive error $\varepsilon \pmax$ in time $2^{\Oh{\frac{1}{\varepsilon}+d+m+l+u+\pmax+\enc{I}}^{1-\delta}}$ unless the ETH fails. 
\end{corollary}
\begin{proof}
	Setting $\varepsilon:=\frac{1}{\pmax+1}$ yields $OPT+\varepsilon \pmax=OPT+\frac{1}{\pmax+1} \pmax<OPT+1$ and hence for this choice of $\varepsilon$, such an additive approximation yields an optimal solution. So with $\varepsilon:=\frac{1}{\pmax+1}$, a running time $2^{\Oh{\frac{1}{\varepsilon}+d+m+u+\pmax+\enc{I}}^{1-\delta}}$ would give an exact algorithm with running time $2^{\Oh{d+m+l+u+\pmax+\enc{I}}^{1-\delta}}$, contradicting the ETH via \autoref{cor:lb2exact}.
\end{proof}

Note that the running time of the additive approximation algorithm is $\additiveapproxrt$. So the following lower bound shows that the dependency on $\eps$ and $m$ is almost optimal:
\thmlowerboundapprox*
\begin{proof}
    Again, setting $\varepsilon:=\frac{1}{\pmax+1}$ yields $OPT+\varepsilon \pmax=OPT+\frac{1}{\pmax+1} \pmax<OPT+1$ and hence for this choice of $\varepsilon$, such an additive approximation yields an optimal solution. So with $\varepsilon:=\frac{1}{\pmax+1}$, an algorithm with running time $(m\pmax)^{\Oh{(\frac{1}{\varepsilon})^{1-\delta}}}$ for the approximation problem would give an exact algorithm for \textsc{3-Sat} with running time $2^{\Oh{\log(NN^{1+\delta'})(N^{1+\delta'})^{1-\delta}}}=2^{\oh{N}}$ via \autoref{thm:chenreductionone} (with $\delta'$ being the constant from the reduction), contradicting the ETH.
\end{proof}

Similarly, we can also show a lower bound that has the form of an EPTAS running time, i.e., $f(\frac{1}{\eps})\enc{I}^{O(1)}$:
\begin{theorem}\label{thm:eptaslb}
    Let $\delta>0$. Then \tf{\P}{}{\cmax} cannot be approximated with additive error at most $\eps \pmax$ in time $\left(\frac{1}{\eps}\right)^{O(\frac{1}{\eps})^{1-\delta}}\enc{I}^{O(1)}$, unless the ETH fails.
\end{theorem}
\begin{proof}
    An algorithm for \tf{\P}{}{\cmax} with running time $\left(\frac{1}{\eps}\right)^{O(\frac{1}{\eps})^{1-\delta}}\enc{I}^{O(1)}$ and additive error $\eps\pmax$, where $\eps:=\frac{1}{\pmax+1}$ would yield an exact algorithm for \textsc{3-Sat} with running time:
    \begin{align*}
        \left(\frac{1}{\eps}\right)^{O(\frac{1}{\eps})^{1-\delta}}\enc{I}^{O(1)} &=\pmax^{O(\pmax)^{1-\delta}}\enc{I}^{O(1)} \\
        &=(N^{1+\delta'})^{O(N^{1+\delta'})^{1-\delta}}(N\log(N))^{O(1)} \\
        &=(N^{1+\delta'})^{O(N^{1+\delta'})^{1-\delta}}(N\log(N))^{O(1)} \\
        &=2^{o(N)}
    \end{align*}
    Again, $\delta'$ is the constant from \autoref{thm:chenreductionone}.
\end{proof}

Koutecký and Zink showed the following:\footnote{The algorithm even works for objectives $\sumwc$ and $\elltwo$ (minimizing the $\ell_2$-norm of the machine load vector).}
\begin{theorem}[Koutecký \& Zink~\cite{KZ20}]
    \tf{\{\unrel,\Q\}}{}{\{\cmax\}} can be solved in time $m\n^{\Oh{d}}$ via dynamic programming.
\end{theorem}
With either reduction from \cite{CJZ18}, one can show that this running time is essentially tight, even for identical machines:
\begin{restatable}{corollary}{kzlowerbound}
    \tf{\P}{}{\cmax} cannot be solved in time $m\n^{\Oh{d^{1-\eps}}}$ for any $\eps>0$, unless the ETH fails.
\end{restatable}
\begin{proof}
    If we use the first reduction (\autoref{thm:chenreductionone}), $m=\Oh{N}$ and $\n=d=\Oh{N}$. This yields
    \[m\n^{\Oh{d^{1-\eps}}}=NN^{\Oh{N^{1-\eps}}}=2^{\oh{N}},\]
    which contradicts the ETH. Similarly, with the second reduction (\autoref{thm:chenreductiontwo}), we have some number of machines $m$ that we can freely choose and set to a constant, $\n=\Oh{N+m}=\Oh{N}$ jobs and $d=\Oh{N}$ different processing times. This yields
    \[m\n^{\Oh{d^{1-\eps}}}=N^{\Oh{N^{1-\eps}}}=2^{\oh{N}},\]
    which also contradicts the ETH.
\end{proof}

\end{document}

%% file: localsearch.tex
\begin{figure}
  \captionsetup[subfigure]{font=footnotesize}
  \centering
  \subcaptionbox{A graph representation of the local search algorithm. A:~Machines with load at most $u$. B:~Machines with a load in the interval $\closedinterval{u+1,u+\eps \pmax}$. C:~Machines with a load exceeding $u + \eps \pmax$. In this example, $m_1$ cannot exchange any jobs before $m_4$ or $m_2$ is repaired.\label{subfig:localsearch1}}[.47\textwidth]{%

  \begin{tikzpicture}
    [scale=0.76,auto=left,every node/.style={circle,draw,minimum size=10mm}]
    \node (s) at (0,0) {s};

    \foreach \id/\x/\y/\name in {
      a/-2/2.5/1,
      b/0/2.5/2,
      c/2/2.5/3,
      d/-2/5.2/4,
      e/0/5.2/5,
      f/2/5.2/6,
      g/-2/7.9/7,
      h/0/7.9/8,
      i/2/7.9/9
    }
    \node (\id) at (\x,\y) {$m_{\name}$};

    \foreach \from/\to/\weight/\x/\y/\br in {
      s/a/0/0.2/0.1/0,
      s/b/0/0.3/-0.1/0,
      s/c/0/0/-0.5/0,
      b/c/1/0/-0.3/0,
      b/e/1/-0.3/0/0,
      e/g/1/0.2/0.8/0,
      a/d/1/0.3/0/0,
      f/c/1/-0.4/0/0,
      b/a/1/0/0.3/0,
      a/b/1/0/-0.3/-20, 
      d/g/1/0.3/0/0,
      c/h/1/0.4/0.5/0,
      h/e/1/-0.6/0/0,
      i/f/1/-0.3/0/0,
      h/f/1/-0.1/-0.3/0,
      e/b/1/0/0/45,
      d/e/1/0/-0.2/0,
      b/h/1/0.6/-0.3/20,
      f/h/1/0.1/0.9/20,
      f/i/1/0.6/0/30  
      }
      \draw [-latex] (\from) to [bend right=\br,edge node={node [shift={(\x,\y)},draw=none] {\weight}}] (\to);

      \node[draw,dashed,inner sep=10pt,thick, rectangle,fit= (g) (h) (i)] (A) {};

      \node[draw,dashed,inner sep=10pt,thick,rectangle, fit= (d) (e) (f)] (B) {};

      \node[draw,dashed,inner sep=10pt,thick,rectangle,fit= (a) (b) (c)] (C) {};

      \foreach \rect in {A,B,C}
      \node[anchor = north,shift={(-0.3,0)},draw = none] at (\rect.north west) {\rect}; 
      
  \end{tikzpicture}
  }
  \hfill
  %
  %
  %
  %
  %
  %
  %
  \subcaptionbox{A close up view showing the connections between two machines $m_1$ and $m_2$. Bullet points: Jobs assigned to slots. Gray lines: Edges of weight 0 connecting jobs on one machine. In this example, the job of class $J_3$ in row 3 on $m_1$ has the same size as the one of the same class in row 1 on $m_2$. Therefore, there is no edge between the corresponding bullet points.\label{subfig:localsearch2}}[.47\textwidth]{

  \begin{tikzpicture}
  \pgfdeclarelayer{background}
\pgfdeclarelayer{foreground}
\pgfsetlayers{background,main,foreground}

\tikzset{
    col/.style = {
        draw=none,minimum width =1.2cm, minimum height =0.4cm,anchor=base
    },
    row/.style = {
        draw=none,minimum width =0.6cm, minimum height =0.5cm
    },
    hidden/.style = {
        gray!60
    }
    }

    \matrix [nodes=draw,very thin,minimum width =1.2cm, minimum height =0.6cm,ampersand replacement=\&] (machine0)
    {
      \node[col] {$\mathcal{J}_1$}; \& \node[col] {$\mathcal{J}_2$}; \& \node[col] {$\mathcal{J}_3$}; \& \node[col] {$\mathcal{J}_4$};  \\
      \node[] (a0) { $\bullet$}; \&  \node[] (b0) { $\bullet$}; \&  \node[] (c0) { $\bullet$}; \&  \node[] (d0) { $\bullet$};  \\

      \node[] (e0) { $\bullet$}; \& \node[] (f0) {}; \&  \node[] (g0) { $\bullet$}; \&  \node[] (h0) {};  \\

      \node[] (i0) { $\bullet$}; \&  \node[] (j0) {}; \&  \node[] (k0) { $\bullet$}; \&  \node[] (l0) {};  \\

      \node[] (m0) {}; \&  \node[] (n0) {}; \&  \node[] (o0) {}; \&  \node[] (p0) {};  \\
    };

    \foreach \from in {a0,b0,c0,d0,e0,g0,i0,k0}{
    \foreach \to in {a0,b0,c0,d0,e0,g0,i0,k0}{

    \ifthenelse{\equal{\to}{\from}}{}{
    
    \begin{pgfonlayer}{background} 
    \draw[hidden] (\from.center) to [bend right=20,edge node={node [draw=none] {}}] (\to.center);
    \end{pgfonlayer}{background} 
    
    }
    }}

    \matrix [below=1.5cm of machine0,nodes=draw,very thin,minimum width =1.2cm, minimum height =0.6cm,ampersand replacement=\&] (machine1)
    {
      
      \node[] (a1) { $\bullet$}; \&  \node[] (b1)  {}; \&  \node[] (c1) { $\bullet$}; \&  \node[] (d1) {};  \\

      \node[] (e1) { $\bullet$}; \& \node[] (f1) {}; \&  \node[] (g1) { $\bullet$}; \&  \node[] (h1) {};  \\

      \node[] (i1) {}; \&  \node[] (j1) {}; \&  \node[] (k1) { $\bullet$}; \&  \node[] (l1) {};  \\

      \node[] (m1) {}; \&  \node[] (n1) {}; \&  \node[] (o1) {$\bullet$}; \&  \node[] (p1) {};  \\
      \node[col] {$\mathcal{J}_1$}; \& \node[col] {$\mathcal{J}_2$}; \& \node[col] {$\mathcal{J}_3$}; \& \node[col] {$\mathcal{J}_4$};  \\
    };

    \foreach \from in {a1,c1,e1,g1,k1,o1}{
    \foreach \to in {a1,c1,e1,g1,k1,o1}{

    \ifthenelse{\equal{\to}{\from}}{}{
    
    \begin{pgfonlayer}{background} 
    \draw[hidden] (\from.center) to [bend right=20] (\to.center);
    \end{pgfonlayer}{background} 
    
    }
    }}

    \foreach \from/\to/\x/\y in {
      a0/a1/-0.1/-0.4,
      e0/a1/0.1/-0.3,
      e1/i0/0.1/0}{
    \draw[-{Stealth[length=3mm, width=2mm]}] (\from.center) to [bend right=20,edge node={node [shift={(\x,\y)},draw=none] {1}}] (\to.center);}

    \foreach \from/\to/\x/\y/\br in {
      c1/c0/-0.4/-0.4/-20,
      g1/c0/-0.4/-0.4/-20,
      k1/c0/-0.3/-0.4/-20,
      o1/c0/-0.2/-0.4/-20,
      c1/g0/-0.1/-0.1/20,
      g1/g0/-0.2/-0.1/20,
      k1/g0/-0.3/-0.1/20,
      o1/g0/-0.5/-0.1/20,
      g1/k0/0.4/0.1/60,
      k1/k0/0.3/0.1/60,
      o1/k0/0.2/0.1/60
      }{
       
    \draw[-{Stealth[length=3mm, width=2mm]}] (\from.center) to [bend right=\br,edge node={node [shift={(\x,\y)},draw=none] {1}}] (\to.center);}

    \node[anchor = north,shift={(0.3,0)},draw = none] at (d0.north east) {$\mathbf{m_1}$}; 
    \node[anchor = north,shift={(0.3,0)},draw = none] at (d1.north east) {$\mathbf{m_2}$}; 

  \end{tikzpicture}
  }
  \end{figure}

%% file: swap.tex
\begin{figure}[ht]
  \captionsetup[subfigure]{font=footnotesize}
  \centering
  \subcaptionbox{Two machines $m_1$ and $m_2$ before the swap. The edge $\mathbf{(\thirdvertex,\tovertex)}$ is implied by $\mathbf{(\thirdvertex,\fromvertex)}$ and $\mathbf{(\fromvertex,\tovertex)}$ because all edges of weight 1 are based on the \hyphenation{tran-si-ti-ve}transitive relation $>$. $\gamma_1$:~shortest path to $\mathbf{\fromvertex}$. $\gamma_2$:~shortest path to $\mathbf{\thirdvertex}$. For the swap to be carried out, $\gamma_1$ must be shorter than $|\gamma_2| + 1$.\label{subfig:swap1}}[.47\textwidth]{%

  \begin{tikzpicture}
  \pgfdeclarelayer{background}
\pgfdeclarelayer{foreground}
\pgfsetlayers{background,main,foreground}

\tikzset{
    col/.style = {
        draw=none,minimum width =1.2cm, minimum height =0.4cm,anchor=base
    },
    row/.style = {
        draw=none,minimum width =0.6cm, minimum height =0.5cm
    },
    hidden/.style = {
        gray!60
    }
    }

    \matrix [nodes=draw,very thin,minimum width =1.2cm, minimum height =0.6cm,ampersand replacement=\&] (machine0)
    {
      \node[col] {$\mathcal{J}_1$}; \& \node[col] {$\mathcal{J}_2$}; \& \node[col] {$\mathcal{J}_3$}; \& \node[col] {$\mathcal{J}_4$};  \\
      \node[] (a0) { $\bullet$}; \&  \node[] (b0) { }; \&  \node[] (c0) { $\bullet$}; \&  \node[] (d0) { $\bullet$};  \\

      \node[] (e0) { $\bullet$}; \& \node[] (f0) {}; \&  \node[] (g0) { $\bullet$}; \&  \node[] (h0) {};  \\

      \node[] (i0) { }; \&  \node[] (j0) {}; \&  \node[font = {\huge}] (k0) {$\bullet$}; \&  \node[] (l0) {};  \\

      \node[] (m0) {}; \&  \node[] (n0) {}; \&  \node[] (o0) {}; \&  \node[] (p0) {};  \\
    };

    \foreach \from in {a0,c0,d0,e0,g0,k0}{
    \foreach \to in {a0,c0,d0,e0,g0,k0}{

    \ifthenelse{\equal{\to}{\from}}{}{
    
    \begin{pgfonlayer}{background} 
    \draw[hidden] (\from.center) to [bend right=20] (\to.center);
    \end{pgfonlayer}{background} 
    
    }
    }}

    \matrix [below=1cm of machine0,nodes=draw,very thin,minimum width =1.2cm, minimum height =0.6cm,ampersand replacement=\&] (machine1)
    {
      
      \node[] (a1) { $\bullet$}; \&  \node[] (b1)  {$\bullet$}; \&  \node[] (c1) { $\bullet$}; \&  \node[] (d1) {};  \\

      \node[] (e1) { $\bullet$}; \& \node[] (f1) {}; \&  \node[] (g1) { $\bullet$}; \&  \node[] (h1) {};  \\

      \node[] (i1) {}; \&  \node[] (j1) {}; \&  \node[font = {\huge}] (k1) {$\bullet$}; \&  \node[] (l1) {};  \\

      \node[] (m1) {}; \&  \node[] (n1) {}; \&  \node[] (o1) {}; \&  \node[] (p1) {};  \\
      \node[col] (j1) {$\mathcal{J}_1$}; \& \node[col] (j2) {$\mathcal{J}_2$}; \& \node[col] (j3) {$\mathcal{J}_3$}; \& \node[col] (j4) {$\mathcal{J}_4$};  \\
    };

    \node[above=8mm of c0,draw = none] (n) {$\bullet$}; 
    \node[below=5mm of j1,draw = none,font = {\Huge}] (s) {$\bullet$}; 
    \node[right=-2mm of s,draw = none,font = {\Large}] {$\mathbf{s}$}; 
    \node[below=5mm of j3,draw = none] (u) {$\bullet$};
    \node[right=-2mm of u,draw = none,font = {\Large}] {$\mathbf{\thirdvertex}$}; 
    \node[right=of k1.center,draw = none,font = {\Large},shift={(-1,0)}] {$\mathbf{\fromvertex}$}; 
    \node[right=of k0.center,draw = none,font = {\Large},shift={(-1,0)}] {$\mathbf{\tovertex}$}; 

    \foreach \from in {a1,b1,c1,e1,g1,k1}{
    \foreach \to in {a1,b1,c1,e1,g1,k1}{

    \ifthenelse{\equal{\to}{\from}}{}{
    
    \begin{pgfonlayer}{background} 
    \draw[hidden] (\from.center) to [bend right=20] (\to.center);
    \end{pgfonlayer}{background} 
    
    }
    }}

    \begin{pgfonlayer}{background} 
    \foreach \from/\to/\x/\y/\br in {
      a1/a0/-0.1/-0.5/-10,
      a1/e0/0.1/-0.4/10,
      e1/e0/0.1/-0.4/20
      }{
    \draw[-{Stealth[length=5mm, width=2mm]},gray!60] (\from.center) to [bend right=\br,edge node={node [shift={(\x,\y)},draw=none] {1}}] (\to.center);}
    \end{pgfonlayer}{background} 

    \begin{pgfonlayer}{background} 
    \foreach \from/\to/\x/\y/\br in {
      c1/c0/-0.3/-0.5/-20,
      g1/c0/-0.3/-0.5/-20,
      k1/c0/-0.2/-0.5/-20,
      c1/g0/-0.2/-0.2/20,
      g1/g0/-0.3/-0.2/20,
      k1/g0/-0.4/-0.2/20,
      g1/k0/0.1/0.2/40
      }{
    \draw[-{Stealth[length=5mm, width=2mm]},hidden] (\from.center) to [bend right=\br,edge node={node [shift={(\x,\y)},draw=none] {1}}] (\to.center);}
    \end{pgfonlayer}{background} 
     \foreach \from/\to/\x/\y/\br in {
        k1/k0/-0.2/0/-45,
        u/k1/0.1/0.4/20,
        u/k0/0.2/-0.2/40,
        n/k0/-0.1/-0.1/20
    }{
    \draw[-{Stealth[length=5mm, width=2mm]}] (\from.center) to [bend right=\br,edge node={node [shift={(\x,\y)},draw=none] {1}}] (\to.center);}

    \draw (b1.center) to [bend right=20,edge node={node [shift={(-0.2,0)},draw=none] {0}}] (k1.center) {};

     \draw[-{Stealth[length=3mm, width=2mm]},dotted,ultra thick] (s.center) to [bend right=20,edge node={node [shift={(-0.1,0.6)},draw=none,font = {\Large}] {$\gamma_1$}}] (b1.center);

     \draw[-{Stealth[length=3mm, width=2mm]},dotted,ultra thick] (s.center) to [bend right=-30,edge node={node [shift={(0,-0.3)},draw=none,font = {\Large}] {$\gamma_2$}}] (u.center);

    \node[anchor = north,shift={(0.3,0)},draw = none] at (d0.north east) {$\mathbf{m_1}$}; 
    \node[anchor = north,shift={(0.3,0)},draw = none] at (d1.north east) {$\mathbf{m_2}$}; 

  \end{tikzpicture}

  }
  \hfill
    \subcaptionbox{The machines after the swap. All the incoming edges of weight 1 have been exchanged between $\mathbf{\fromvertex}$ and $\mathbf{\tovertex}$ and the edge $\mathbf{(\fromvertex,\tovertex)}$ has been flipped.\label{subfig:swap2}}[.47\textwidth]{%

  \begin{tikzpicture}
  \pgfdeclarelayer{background}
\pgfdeclarelayer{foreground}
\pgfsetlayers{background,main,foreground}

\tikzset{
    col/.style = {
        draw=none,minimum width =1.2cm, minimum height =0.4cm,anchor=base
    },
    row/.style = {
        draw=none,minimum width =0.6cm, minimum height =0.5cm
    },
    hidden/.style = {
        gray!60
    }
    }

    \matrix [nodes=draw,very thin,minimum width =1.2cm, minimum height =0.6cm,ampersand replacement=\&] (machine0)
    {
      \node[col] {$\mathcal{J}_1$}; \& \node[col] {$\mathcal{J}_2$}; \& \node[col] {$\mathcal{J}_3$}; \& \node[col] {$\mathcal{J}_4$};  \\
      \node[] (a0) { $\bullet$}; \&  \node[] (b0) { }; \&  \node[] (c0) { $\bullet$}; \&  \node[] (d0) { $\bullet$};  \\

      \node[] (e0) { $\bullet$}; \& \node[] (f0) {}; \&  \node[] (g0) { $\bullet$}; \&  \node[] (h0) {};  \\

      \node[] (i0) { }; \&  \node[] (j0) {}; \&  \node[font = {\huge}] (k0) {$\bullet$}; \&  \node[] (l0) {};  \\

      \node[] (m0) {}; \&  \node[] (n0) {}; \&  \node[] (o0) {}; \&  \node[] (p0) {};  \\
    };

    \foreach \from in {a0,c0,d0,e0,g0,k0}{
    \foreach \to in {a0,c0,d0,e0,g0,k0}{

    \ifthenelse{\equal{\to}{\from}}{}{
    
    \begin{pgfonlayer}{background} 
    \draw[hidden] (\from.center) to [bend right=20] (\to.center);
    \end{pgfonlayer}{background} 
    
    }
    }}

    \matrix [below=1cm of machine0,nodes=draw,very thin,minimum width =1.2cm, minimum height =0.6cm,ampersand replacement=\&] (machine1)
    {
      
      \node[] (a1) { $\bullet$}; \&  \node[] (b1)  {$\bullet$}; \&  \node[] (c1) { $\bullet$}; \&  \node[] (d1) {};  \\

      \node[] (e1) { $\bullet$}; \& \node[] (f1) {}; \&  \node[] (g1) { $\bullet$}; \&  \node[] (h1) {};  \\

      \node[] (i1) {}; \&  \node[] (j1) {}; \&  \node[font = {\huge}] (k1) {$\bullet$}; \&  \node[] (l1) {};  \\

      \node[] (m1) {}; \&  \node[] (n1) {}; \&  \node[] (o1) {}; \&  \node[] (p1) {};  \\
      \node[col] (j1) {$\mathcal{J}_1$}; \& \node[col] (j2) {$\mathcal{J}_2$}; \& \node[col] (j3) {$\mathcal{J}_3$}; \& \node[col] (j4) {$\mathcal{J}_4$};  \\
    };

    \node[above=8mm of c0,draw = none] (n) {$\bullet$}; 
    \node[below=5mm of j1,draw = none,font = {\Huge}] (s) {$\bullet$}; 
    \node[right=-2mm of s,draw = none,font = {\Large}] {$\mathbf{s}$}; 
    \node[below=5mm of j3,draw = none] (u) {$\bullet$};
    \node[right=-2mm of u,draw = none,font = {\Large}] {$\mathbf{\thirdvertex}$}; 
    \node[right=of k1.center,draw = none,font = {\Large},shift={(-1,0)}] {$\mathbf{\fromvertex}$}; 
    \node[right=of k0.center,draw = none,font = {\Large},shift={(-1,0)}] {$\mathbf{\tovertex}$}; 

    \foreach \from in {a1,b1,c1,e1,g1,k1}{
    \foreach \to in {a1,b1,c1,e1,g1,k1}{

    \ifthenelse{\equal{\to}{\from}}{}{
    
    \begin{pgfonlayer}{background} 
    \draw[hidden] (\from.center) to [bend right=20] (\to.center);
    \end{pgfonlayer}{background} 
    
    }
    }}

    \begin{pgfonlayer}{background}
    \foreach \from/\to/\x/\y/\br in {
      a1/a0/-0.1/-0.5/-10,
      a1/e0/0.1/-0.4/10,
      e1/e0/0.1/-0.4/20
      }{
    \draw[-{Stealth[length=5mm, width=2mm]},hidden] (\from.center) to [bend right=\br,edge node={node [shift={(\x,\y)},draw=none] {1}}] (\to.center);}

    \foreach \from/\to/\x/\y/\br in {
      c1/c0/-0.3/-0.5/-20,
      g1/c0/-0.3/-0.5/-20,
      k1/c0/-0.2/-0.5/-20,
      c1/g0/-0.2/-0.2/20,
      g1/g0/-0.3/-0.2/20,
      k1/g0/-0.4/-0.2/20,
      g1/k0/0.1/0.2/40
      }{
    \draw[-{Stealth[length=5mm, width=2mm]},hidden] (\from.center) to [bend right=\br,edge node={node [shift={(\x,\y)},draw=none] {1}}] (\to.center);}
    \end{pgfonlayer}{background}

     \foreach \from/\to/\x/\y/\br in {
        k0/k1/-0.2/0/50,
        u/k1/0.1/0.4/20,
        u/k0/0.2/-0.2/40,
        n/k1/0.1/-0.1/15
    }{
    \draw[-{Stealth[length=5mm, width=2mm]}] (\from.center) to [bend right=\br,edge node={node [shift={(\x,\y)},draw=none] {1}}] (\to.center);}

    \draw (b1.center) to [bend right=20,edge node={node [shift={(-0.2,0)},draw=none] {0}}] (k1.center) {};

     \draw[-{Stealth[length=3mm, width=2mm]},dotted,ultra thick] (s.center) to [bend right=20,edge node={node [shift={(-0.1,0.6)},draw=none,font = {\Large}] {$\gamma_1$}}] (b1.center);

     \draw[-{Stealth[length=3mm, width=2mm]},dotted,ultra thick] (s.center) to [bend right=-30,edge node={node [shift={(0,-0.3)},draw=none,font = {\Large}] {$\gamma_2$}}] (u.center);

    \node[anchor = north,shift={(0.3,0)},draw = none] at (d0.north east) {$\mathbf{m_1}$}; 
    \node[anchor = north,shift={(0.3,0)},draw = none] at (d1.north east) {$\mathbf{m_2}$}; 

  \end{tikzpicture}

  }
  \end{figure}